\documentclass[12pt]{article}

\usepackage{setspace}
\usepackage{amsmath,amsfonts}
\usepackage{theorem}

\usepackage{epsfig}
\usepackage{xcolor}

\usepackage{lscape}

\usepackage{multirow}

\usepackage{rotating}

\usepackage{tikz}

\usepackage{makecell}
\newcolumntype{x}[1]{>{\centering\arraybackslash}p{#1}}
\usepackage[hyphens]{url}
\usepackage[breaklinks]{hyperref}
\hypersetup{
    colorlinks=true,       % false: boxed links; true: colored links
    linkcolor=red,          % color of internal links
    citecolor=blue,        % color of links to bibliography
    filecolor=magenta,      % color of file links
    urlcolor=cyan           % color of external links
}

\usepackage[sort&compress]{natbib}
\bibpunct{(}{)}{,}{a}{,}{,}

\usepackage{appendix}

\theoremstyle{plain} { \theorembodyfont{\rmfamily}
\newtheorem{definition}{Definition}[section]

}
\newtheorem{proposition}{Proposition}[section]
\newtheorem{lemma}{Lemma}[section]
\newtheorem{corollary}{Corollary}[section]
\newtheorem{theorem}{Theorem}[section]

\def\qed{\quad Q.E.D.}
\newenvironment{proof}{\vspace{1ex}\noindent{\bf Proof}\hspace{0.5em}}
    {\qed\vspace{1ex}}
\newenvironment{pfof}[1]{\vspace{1ex}\noindent{\bf Proof of #1}\hspace{0.5em}}
    {\qed\vspace{1ex}}

\newcommand{\expect}{\mathbb{E}}

\begin{document}

\title{Equilibrium Policy on Dividend and Capital Injection under Time-inconsistent Preferences}
\author{Sang Hu\thanks{School of Data Science, The Chinese University of Hong Kong, Shenzhen, Guangdong Province, China 518172.
Email: \texttt{husang@cuhk.edu.cn}. %The author would like to acknowledge the funding of National Natural Science Foundation of China (Grant No. 12271462).
}
\and Zihan Zhou\thanks{School of Data Science, The Chinese University of Hong Kong, Shenzhen, Guangdong Province, China 518172.
Email: \texttt{zihanzhou@link.cuhk.edu.cn}.}
}

\maketitle

\abstract{ This paper studies the dividend and capital injection problem under a diffusion risk model with general discount functions. A proportional cost is imposed when injecting capitals. For exponential discounting as time-consistent benchmark, we obtain the closed-form solutions and show that the optimal strategies are of threshold type. Under general discount function which leads to time-inconsistency, we adopt the definition of weak equilibrium and obtain the extended HJB equation system. An explicit solution is derived under pseudo-exponential discounting where three cases of the dividend and capital injection thresholds are obtained. Numerical examples show that large capital injection cost may lead to no capital injection at all, while larger difference in group discount rate leads to higher equilibrium value function. }\\


\section{Introduction}

Dividend optimization has been a major topic in finance and actuarial science for a long time. It aims to find the optimal strategy for an insurance company to maximize the expected cumulative discounted dividends. \cite{de1957impostazione} was the first one to propose the dividend optimization problem, considering the dividend strategy in a simple discrete random walk model. Since then, the problem has been extensively studied under various models, such as the classical Cram\'{e}r-Lundberg (C-L) model, diffusion model, and regime-switching model; see, e.g., \cite{gerber1969entscheidungskriterien}, \cite{jeanblanc1995optimization}, \cite{asmussen1997controlled} for more details.

In addition to dividend policy, risk management is also of great interest in actuarial science. Capital injection, also known as issuance of equity, is one of the commonly used tool to prevent the company from bankruptcy; see, e.g., \cite{dickson2004some}. \cite{sethi2002optimal} then considered dividend optimization with capital injections under the diffusion risk model without taking the possibility of ruin into consideration and showed the optimal strategy is barrier type.
\cite{kulenko2008optimal} addressed the same problem under the C-L model where bankruptcy is not allowed either. \cite{zhou2012optimal} considered the optimal reinsurance and dividend strategy under the diffusion model, where they also used the capital injection to avoid bankruptcy.
Following these works, \cite{lokka2008optimal} reconsidered the problem with ruin time included and showed that capital injection is not always compulsory. Some further studies include \cite{he2008optimal} with additional proportional reinsurance strategy, \cite{avanzi2011optimal} under a dual model, and \cite{zhu2016optimal} under a more general diffusion model with nonconstant drift and volatility terms.

This paper follows \cite{lokka2008optimal} to consider total expected discounted dividends with penalized capital injections until ruin time, but under time-inconsistent preferences. The capital injection strategy in most of the literatures takes the form of singular control, that is, the capital injection strategy is a non-decreasing c\`{a}gl\`{a}d process. Therefore, capital injection in the presence of discounting is never optimal until the ruin time. Then the problem degenerates to determine the optimal dividend strategy only, where the injection process is obtained automatically.
In this paper, we assume the capital is injected at a restricted rate and therefore, the capital injection strategy may no longer be a c\`{a}gl\`{a}d process.
In other words, it is not feasible to inject capital with arbitrary size at the ruin time to prevent bankruptcy because of the restricted injection rate. Therefore, there still exists a positive probability that the insurance company will go bankrupt even with the capital injections, which is more reasonable and consistent with reality.

Whereas previous literatures studied the dividend problem under classical exponential discounting, some empirical evidence show that human's dynamic decision is more likely to be made under non-exponential discounting; see, e.g., \cite{laibson1998life} and references therein.
Non-exponential discounting may lead to the time-inconsistent decision problem, that is, the optimal strategy at time $s$ may no longer be optimal at some latter time $t$. Therefore, standard dynamic programming and classical Hamilton-Jacobi-Bellman (HJB) equation cannot be directly used to analyze the stochastic control problem because Bellman's principle of optimality does not hold, which is the major technical difficulty faced by the time-inconsistent problem.

\cite{strotz1955myopia} first noticed the time-inconsistent problem and put forth the precommitted strategy and equilibrium strategy in face of the inconsistency.
Equilibrium strategy in continuous time was first introduced in \cite{ekeland2006being} and \cite{ekeland2008investment}. Following their definition, \cite{bjork2010general} derived an extended HJB system, both in discrete and continuous time, for solving equilibrium value function which is the objective value with respect to the equilibrium strategy.
Thereafter, plenty of studies worked on various types of time-inconsistent problems, e.g., \cite{bjork2014mean} in mean-variance portfolio optimization problem, \cite{hu2012} in time-inconsistent LQ stochastic control, \cite{heweightedCPT} in optimal stopping under cumulative prospect theory, etc.

In the time-inconsistent dividend problems, \cite{chen2014optimal} studied both naive and sophisticated strategy with quasi-hyperbolic discounting under diffusion model and \cite{chen2017optimal} under C-L model. \cite{zhao2014dividend} studied the general time-inconsistent dividend problem under the diffusion risk model and derived the equilibrium dividend strategy from the extended HJB which is a threshold strategy for two special discount functions. \cite{li2016equilibrium} considered the time-inconsistent dividend problem under a dual model. In particular, \cite{chen2016optimal} studied the dividend and capital injection problem under the dual risk model in the presence of a time-inconsistent preference, where a quasi-hyperbolic discount function is used to describe the time preference.

To the best of our knowledge, this work is the first to consider equilibrium dividend and capital injection policy under time-inconsistent preferences. We adopt the weak equilibrium definition of \cite{bjork2010general} to seek equilibrium strategy under a diffusion risk model with general discounting.
An extended HJB equation system with rigorous proof of verification theorem are provided for the equilibrium value function and the solution under a special discount function, pseudo-exponential discount function, is derived.
Not that for the time-consistent counterpart, the concavity of the equilibrium value function is directly obtained from the optimality of the value function; see \citet[Proposition 1.1]{jgaard1999controlling}. However, it is not the case for the equilibrium value function.
With delicate analysis, we prove that the equilibrium value function in our problem is indeed concave, which implies that the equilibrium dividend and capital injection strategies are of threshold form.

The equilibrium value function is solved from the extended HJB equation together with the thresholds for dividend and capital injection. The equilibrium strategy can be classified into three cases. First, the capital is injected at the maximal rate once the surplus is below the threshold $x_1$ and the dividend is paid at the maximal rate once the surplus exceeds the threshold $x_2$, where $x_2 \ge x_1 > 0$. Second, the dividend is paid at the maximal rate once the surplus exceeds the threshold $x_2 > 0$ but the capital is never injected. Third, the dividend is always paid at the maximal rate regardless of the surplus state.
Note that our capital injection policy is different from the strategy in \cite{chen2016optimal} in which the capital injection is a non-decreasing c\`{a}gl\`{a}d process. Because ruin is not allowed in \cite{chen2016optimal}, the equilibrium strategy was always to inject capital at ruin time and pay dividends at the barrier.
In our setting, however, it is possible that the equilibrium strategy does not require capital injection even if the insurance company is at ruin time.

Furthermore, we find that capital injection is never a choice if its cost is too high. However, even for a reasonable injection cost, one cannot inject capital arbitrarily at the ruin time because of the restricted injection rate, which is different from the singular control setting. Our equilibrium capital injection strategy is a threshold type that injects capital when the surplus level is relatively low and stops injection once the surplus grows above the threshold. There still exists a positive probability that the insurance company will go bankrupt even with the capital injections. In an extreme case when there is no cost for capital injection, the threshold for capital injection coincides with that for dividends and the insurance company is either injecting capital or paying dividends.

Numerical analysis also shows that when the maximal dividend paying rate is restricted to be relatively low, it is better for the shareholders to always receive the dividend regardless of the bankruptcy risk of the insurance company. If the maximal dividend paying rate is not too restrictive but the proportional capital injection cost is too high, then the shareholders under the equilibrium strategy never inject capital even if the surplus of the insurance company is very low. Both the dividend paying threshold $x_1$ and capital injecting threshold $x_2$ are increasing with respect to the maximal dividend rate $\bar l$; on the other hand, $x_2$ is increasing in the capital injection cost $\phi$ and $x_1$ is decreasing in $\phi$. It is also shown that in a group with larger difference in the shareholders' discount rate, the equilibrium value function is higher than that in a group with smaller difference.

The rest of the paper is organized as follows. In Section \ref{se:Model}, we introduce the basic formulation of the dividend and capital injection problem under general discounting. We first study the time-consistent counterpart and then derive the extended HJB equation system with the verification theorem under general discounting and the equilibrium dividend and capital injection strategy in Section \ref{sec:time-inconsistent}. The explicit solution to the extended HJB equation under the pseudo-exponential discounting is thoroughly analyzed in Section \ref{se:pseudo} with numerical examples presented in \ref{se:numericalexample} to illustrate the results. Section \ref{sec:conclusion} concludes. The proofs are included in Appendix \ref{proof}.

\section{Model}\label{se:Model}

\subsection{Surplus process}

Denote by $X_t$ the surplus level of the insurance company at time $t$. Suppose that the surplus process $\{X_t\}_{t \ge 0}$ follows the dynamics of drifted Brownian motion:
\begin{align*}
d X_t = \mu dt + \sigma dW_t, \quad X_s = x,
\end{align*}
where $\mu, \sigma > 0$ and $\{W_t\}_{t \ge 0}$ is the standard Brownian motion.

Denote by $l_t$ the dividend paying rate at time $t$. Suppose that the dividend paying rate is chosen from the interval $[0,\bar l]$, where $\bar l \in \mathbb{R}_+$.
Denote by $r_t$ the capital injection rate at time $t$. Suppose that the capital injection rate is chosen from the interval $[0,\bar r]$, where $\bar r \mathbb{R}_+$.
Let ${\bf u} = \{l_t,r_t\}_{t \ge 0}$ be the control process.
Then the controlled surplus process is
\begin{equation}\label{eq:surplus_process}
d X^{\bf u}_t = (\mu - l_t + r_t )dt+ \sigma dW_t, \quad X^{\bf u}_s = x.
\end{equation}

The shareholder's objective is to maximize the total expected discounted net cash inflows, which is the dividends minus the penalized capital injections, until the ruin time:
\begin{align}\label{eq:objective}
\sup_{{\bf u} \in \mathcal{A}} J(s,x,{\bf u}) := \expect_{s,x}\left[\int_s^{\tau_{s,x}^{\bf u}} \delta(s,t) \left( l_t - \phi r_t \right) dt \right],
\end{align}
where the notation $\expect_{s,x}$ stands for the expectation taken at time $s$ with state $X^{\bf u}_s = x$. In \eqref{eq:objective}, $\phi > 1$, implying that a proportional cost occurs when injecting capital, and hence, $\phi$ stands for the level of how costly the capital injection is.
$\delta(s,\cdot): [s, \infty) \mapsto (0, 1]$ standing for the discounting function used at time $s$, which is strictly decreasing and satisfies $\delta(s,s) = 1$.
$\tau_{s,x}^{\bf u}$ denotes the ruin time, which is the first time the surplus level hits zero, i.e.,
\begin{align*}
\tau_{s,x}^{\bf u} = \inf \{ t \ge s: X^{\bf u}_t \le 0, X^{\bf u}_s = x \}.
\end{align*}
$\mathcal{A}$ is the set of admissible controls.
To be precise, we give the definition of the admissible controls in the following.
\begin{definition}\label{def:admissible_control}
${\bf u} = \{l_t,r_t\}$ is admissible if
\begin{itemize}
\item[(i)] For any $t \ge 0$, $l_t \in [0,\bar l]$ and $r_t \in [0,\bar r]$;
\item[(ii)] $\{l_t,r_t\}$ is ${\cal F}_t$-measurable;
\item[(iii)] the stochastic differential equation \eqref{eq:surplus_process} has a unique solution when applying ${\bf u}$;
\item[(iv)] $\mathbb{E}_{s, x}\left[ \int_s^{\tau_{s, x}^{\bf  u}} \delta(s,t) \left|l_t - \phi r_t \right| dt  \right] < \infty$.
\end{itemize}
\end{definition}

\subsection{Time-inconsistency}
The dividend and capital injection optimization problem \eqref{eq:objective} is time-inconsistent in general, due to non-exponential discount function $\delta(s,\cdot)$. In other words, the optimal policy that is obtained at time $s$ is no longer optimal at later time $t > s$.
In this case, Bellman's principle of optimality does not hold, and the conventional dynamic programming is not directly applicable.

\cite{strotz1955myopia} studies the time inconsistent problem within a game theoretic framework by using Nash equilibrium points.
\cite{ekeland2006being} and \cite{ekeland2008investment} provide the definition of the equilibrium concept in continuous time.
\cite{bjork2010general} considered the time-inconsistent control problem in a general Markov framework. Accordingly, there are three approaches to treat time-inconsistency in the dividend and capital injection problem \eqref{eq:objective}, based on whether the shareholder is aware of the time-inconsistency and whether he can commit to the plan.

First, if the shareholder is not aware of time-inconsistency, he then solves the optimization problem \eqref{eq:objective} at each time point and implements the current optimal strategy at current time only. For illustration, the optimal strategy at time $s_1$ with surplus $x_1$ is ${\bf u}_1 := \sup_{{\bf u} \in \mathcal{A}} J(s_1, x_1,{\bf u})$, which is implemented at time $s_1$. Later, at time $s_2$ with surplus $x_2$, solving \eqref{eq:objective} leads to the optimal strategy ${\bf u}_2 := \sup_{{\bf u} \in \mathcal{A}} J(s_2, x_2,{\bf u})$. Then ${\bf u}_2$ is implemented at time $s_2$ instead of the original ${\bf u}_1$. This is called the naive strategy and it is straightforward to see that naive strategy is not time-consistent as it always deviates from the previous plan.

Second, if the shareholder is aware of time-inconsistency and is able to commit to his pre-determined plan, he then solves the optimization problem at initial time $s$, denoted as ${\bf u}_s:= \sup_{{\bf u} \in \mathcal{A}} J_c(s, x,{\bf u})$, and implements this strategy ${\bf u}_s$ at all the later times $t \ge s$, regardless of its optimality later. Such a strategy is called precommitted and it requires self-control or some commitment device such that the shareholder is able to stick to his plan even faced with a better choice later.

Third, if the shareholder is aware of time-inconsistency but does not have self-control or certain commitment device, he is not able to commit to any pre-determined plan and need to seek a consistent plan such that he and his future selves are all willing to follow. In such a sense, one formulates a time-inconsistent problem as a non-cooperative game where himself at each time is regarded as a player. It aims to find the subgame perfect Nash equilibrium and thus, it is called equilibrium strategy.\footnote{This is a situation in which no players will deviate from their current strategy. Equivalently, given the strategies of all the players to be fixed, there is no benefit for any player to change his strategy.} In this work, the equilibrium strategy is the focus afterwards.
We give the formal definition of the equilibrium strategy in the continuous-time setting in the next section.

\subsection{Equilibrium}

Whereas the definition of equilibrium in discrete time is straightforward from game-theoretic view, it is not easily defined in continuous time because an instant change of control variable at current time $t$ only does not influence the process at time $t+$ or thereafter. Instead, the first-order criterion is adopted in the definition of equilibrium in continuous time such that the rate of change in the objective value after perturbation is non-negative. In other words, if a control is changed in arbitrarily small time interval but the corresponding value is not improved, then this control is an equilibrium.
The following presents the formal definition of equilibrium.
\begin{definition}\label{def:eq}
Given an admissible policy ${\bf u} = \{l_t, r_t\}_{t \ge 0}$, define its counterpart with $\epsilon$-perturbation:
\begin{align*}
u^{\epsilon,l,r}_t =
\begin{cases}
(l, r), & t \in [s, s + \epsilon) \\
(l_t, r_t ), & t \in [s + \epsilon, \infty),
\end{cases}
\end{align*}
where $l \in [0,\bar l]$ and $r \in [0,\bar r]$.
Then ${\bf u}$ is an equilibrium if
\begin{align*}
\liminf_{\epsilon \downarrow 0} \frac{J(s,x, {\bf u}) - J(s,x,{\bf u}^{\epsilon,l,r})}{\epsilon} \ge 0,
\end{align*}
for any $s \ge 0, x > 0$, $l \in [0,\bar l]$ and $r \in [0,\bar r]$.
\end{definition}
The $\epsilon$-perturbation policy ${\bf u}^{\epsilon,l,r}$ is different from ${\bf u}$ in a small time interval $[s,s+\epsilon]$ only during which it takes constant control variable $l$ and $r$. Then ${\bf u}$ is an equilibrium if any such perturbation does not improve the objective value.

\section{Extended HJB}\label{sec:time-inconsistent}

\subsection{Time-consistent benchmark}

We first consider the time-consistent counterpart before formally investigating the general time-inconsistent problem. Suppose constant discount rate $\rho > 0$. In this case, the objective in \eqref{eq:objective} becomes
\begin{equation}\label{eq:classical_objective}
\sup_{{\bf u} \in {\cal A}} J_c(s,x,{\bf u}) := \expect_{s,x}\left[\int_s^{\tau_{s, x}^{\bf u}} e^{-\rho (t-s)} \left(l_t - \phi r_t \right) dt \right].
\end{equation}
Since the problem is time-homogeneous, we omit the time argument and denote $V_c(x) := V_c(0,x) = \sup_{\bf u \in {\cal A}} J_c(0,x,{\bf u})$ to be the corresponding value function.
We first show some properties of the value function $V_c(x)$, which will be used later.
\begin{proposition}\label{prop:Vc_property}
Suppose that $V_c(x)$ is the value function of problem \eqref{eq:classical_objective}. Then
\begin{enumerate}
\item[(i)] $V_c(x)$ is bounded;
\item[(ii)] $V_c(x)$ is concave;
\item[(iii)] the limit of $V_c(x)$ is given by
\begin{equation}\label{eq:Vc_limit}
\lim_{x\to\infty} V_c(x) = \frac{\bar l}{\rho}.
\end{equation}
\end{enumerate}
\end{proposition}

Combining with the fact in Proposition \ref{prop:Vc_property} that the value function $V_c$ is concave with boundary conditions $V_c(0) = 0$ and \eqref{eq:Vc_limit}, it is straightforward to obtain the following corollary.
\begin{corollary}\label{corollary:Vc_bound}
Assume that $V_c$ is continuously differentiable. Then $V_c$ is increasing and $V_c^\prime$ is bounded.
\end{corollary}

Assume that $V_c(x)$ is twice-continuously differentiable. The standard dynamic programming approach can be applied to derive the following HJB equation:
\begin{equation}\label{eq:Vc_HJB}
\rho V_c(x) = \sup_{l\in[0,\bar l], r\in[0,\bar r]} \left\{\left(l - \phi r \right) + \left(\mu - l + r\right) V^\prime_c(x) + \frac{1}{2}\sigma^2 V^{\prime\prime}_c(x)\right\},
\end{equation}
with boundary condition $V_c(0) = 0$ and \eqref{eq:Vc_limit}.
Combined with the verification theorem below, we assert that the solution to HJB equation \eqref{eq:Vc_HJB} is indeed the value function of problem \eqref{eq:classical_objective}.
\begin{theorem}\label{thm:verification_exp}
Assume there exists twice-continuously differentiable function $V_c$ that solves \eqref{eq:Vc_HJB} with boundary condition $V_c(0) = 0$ and \eqref{eq:Vc_limit}. Suppose $\vert V_c \vert$ and $\vert V^\prime_c \vert$ are bounded. Then $V_c$ is the value function of problem \eqref{eq:classical_objective}.
\end{theorem}

Note that Proposition \ref{prop:Vc_property} along with Corollary \ref{corollary:Vc_bound} that state the boundedness of $V_c$ and $V_c^\prime$, which in turn verify the assumption in Theorem \ref{thm:verification_exp}. By directly solving the supremum on the right-hand side of \eqref{eq:Vc_HJB}, we obtain that
\begin{equation*}
u^*(x) = (l^*(x), r^*(x)) = \left\{\begin{aligned}
& (0, \bar r), & V^\prime_c(x) > \phi,\\
& (0, 0), \qquad & 1 < V^\prime_c(x) \leq \phi,\\
& (\bar l, 0), & V^\prime_c(x) \leq 1.
\end{aligned}\right.
\end{equation*}
As $V_c(x)$ is concave by Proposition \ref{prop:Vc_property}, we conjecture that there exist $0 < x_r \leq x_l$ such that $V^\prime_c(x) > \phi$ for $x \in (0, x_r)$, $1 < V^\prime_c(x) \leq \phi$ for $x \in [x_r, x_l)$, and $V^\prime_c(x) \leq 1$ for $x \in [x_l, \infty)$. Then the optimal strategy is a feedback control, i.e.,
\begin{equation}\label{eq:optimal_strategy}
u^*(x) = \left\{\begin{aligned}
& (0, \bar r), & x \in (0, x_r),\\
& (0, 0), \qquad & x \in [x_r, x_l),\\
& (\bar l, 0), & x \in [x_l, \infty).
\end{aligned}\right.
\end{equation}
Substituting the optimal strategy \eqref{eq:optimal_strategy} into HJB equation \eqref{eq:Vc_HJB} leads to the following second-order ODEs:
\begin{equation}\label{eq:Vc_ODE}
\begin{cases}
\frac{1}{2} \sigma^2 V^{\prime\prime}_c(x) + (\mu + \bar r) V^\prime_c(x) - \rho V_c(x) - \phi \bar r = 0, & x \in (0, x_r),\\
\frac{1}{2} \sigma^2 V^{\prime\prime}_c(x) + \mu V^\prime_c(x) - \rho V_c(x) = 0, \qquad & x \in [x_r, x_l),\\
\frac{1}{2} \sigma^2 V^{\prime\prime}_c(x) + (\mu - \bar l) V^\prime_c(x) - \rho V_c(x) + \bar l = 0, & x \in [x_l, \infty).
\end{cases}
\end{equation}
Together with the boundary condition $V_c(0) = 0$ and \eqref{eq:Vc_limit}, the solution to the ODEs is
\begin{equation}\label{eq:Vc_analytic}
V_c(x) = \begin{cases}
A_1 \left(e^{\theta_1 x} - e^{\theta_2 x}\right) + \frac{\phi \bar r}{\rho} \left(e^{\theta_2 x} - 1\right),\qquad & x \in (0, x_r),\\
A_2 e^{\theta_3 x} + A_3 e^{\theta_4 x}, & x \in [x_r, x_l),\\
\frac{\bar l}{\rho} + A_4 e^{\theta_5 x}, & x \in [x_l, \infty),
\end{cases}
\end{equation}
where $\theta_i$ and $A_i$ are constants to be determined.
\begin{equation}\label{eq:thetas}
\begin{aligned}
&\theta_1 = \frac{-(\mu+\bar r) + \sqrt{(\mu+\bar r)^2 + 2 \rho \sigma^2}}{\sigma^2}, &&\theta_2 = \frac{-(\mu+\bar r) - \sqrt{(\mu+\bar r)^2 + 2 \rho \sigma^2}}{\sigma^2},\\
&\theta_3 = \frac{-\mu + \sqrt{\mu^2 + 2 \rho \sigma^2}}{\sigma^2}, &&\theta_4 = \frac{-\mu - \sqrt{\mu^2 + 2 \rho \sigma^2}}{\sigma^2},\\
&\theta_5 = \frac{-(\mu-\bar l) - \sqrt{(\mu-\bar l)^2 + 2 \rho \sigma^2}}{\sigma^2}. &&
\end{aligned}
\end{equation}
Using smooth fit principle, i.e., $V_c(x_r-) = V_c(x_r+)$, $V_c(x_l-) = V_c(x_l+)$, $V^\prime_c(x_r-) = V^\prime_c(x_r+)$, $V^\prime_c(x_l-) = V^\prime_c(x_l+)$,
\begin{equation}\label{eq:Vc_As}
\begin{aligned}
&A_1  = -\frac{B_1 \bar l + C_1 \phi \bar r}{\rho D}, && A_2  = -\frac{B_2 \bar l  + C_2 \phi \bar r}{\rho D}, \\
&A_3  = \frac{B_3 \bar l + C_3 \phi \bar r}{\rho D},  && A_4 = - e^{-\theta_5 x_l} \frac{B_4 \bar l + C_4 \phi \bar r}{\rho D},
\end{aligned}
\end{equation}
where
\begin{equation}\label{eq:Vc_BCD}
\begin{aligned}
& B_1 = \theta_5 (\theta_3 - \theta_4) e^{(\theta_3 + \theta_4) x_r}, \\
& B_2 = \theta_5 \left((\theta_1 - \theta_4) e^{(\theta_1+\theta_4) x_r} + (\theta_4 - \theta_2) e^{(\theta_2+\theta_4) x_r}\right), \\
& B_3 = \theta_5 \left((\theta_1 - \theta_3) e^{(\theta_1+\theta_3) x_r} + (\theta_3 - \theta_2) e^{(\theta_2+\theta_3) x_r}\right), \\
& B_4 = \theta_3 (\theta_1 - \theta_4) e^{(\theta_1+\theta_4) x_r + \theta_3 x_l} + \theta_4 (\theta_3 - \theta_1) e^{(\theta_1+\theta_3) x_r + \theta_4 x_l} \\
&\hspace{6ex} + \theta_3 (\theta_4 - \theta_2) e^{(\theta_2+\theta_4) x_r + \theta_3 x_l} + \theta_4 (\theta_2 - \theta_3) e^{(\theta_2+\theta_3) x_r + \theta_4 x_l}, \\
& C_1 = \theta_3(\theta_5 - \theta_4) e^{\theta_3 x_r + \theta_4 x_l} + \theta_4 (\theta_3  - \theta_5) e^{\theta_4 x_r + \theta_3 x_l} + (\theta_3 - \theta_2) (\theta_4 - \theta_5) e^{(\theta_2 + \theta_3) x_r+ \theta_4 x_l}\\
&\hspace{6ex} + (\theta_2 - \theta_4) (\theta_3 - \theta_5) e^{(\theta_2 + \theta_4) x_r + \theta_3 x_l}, \\
& C_2 = \theta_1(\theta_5 - \theta_4) e^{\theta_1 x_r + \theta_4 x_l} + \theta_2 (\theta_4 - \theta_5) e^{\theta_2 x_r + \theta_4 x_l} + (\theta_1 - \theta_2) (\theta_4 - \theta_5) e^{(\theta_1+\theta_2) x_r + \theta_4 x_l}, \\
& C_3 =  \theta_1(\theta_5 - \theta_3) e^{\theta_1 x_r + \theta_3 x_l} + \theta_2(\theta_3 - \theta_5) e^{\theta_2 x_r + \theta_3 x_l} + (\theta_1 - \theta_2) (\theta_3 - \theta_5) e^{(\theta_1+\theta_2) x_r + \theta_3 x_l}, \\
& C_4 = \theta_1 (\theta_3 - \theta_4) e^{\theta_1 x_r + (\theta_3+\theta_4) x_l} + \theta_2 (\theta_4 - \theta_3) e^{\theta_2 x_r + (\theta_3+\theta_4) x_l} \\
&\hspace{6ex} + (\theta_1 -\theta_2) (\theta_4 - \theta_3) e^{(\theta_1+\theta_2) x_r + (\theta_3+\theta_4) x_l}, \\
& D  = (\theta_3 - \theta_1) (\theta_4 - \theta_5) e^{(\theta_1+\theta_3) x_r + \theta_4 x_l}  + (\theta_1 - \theta_4) (\theta_3 - \theta_5) e^{(\theta_1 + \theta_4) x_r + \theta_3 x_l} \\
&\hspace{6ex} + (\theta_2 - \theta_3) (\theta_4 - \theta_5) e^{(\theta_2+\theta_3) x_r + \theta_4 x_l} +(\theta_4 - \theta_2) (\theta_3 - \theta_5) e^{(\theta_2+\theta_4) x_r + \theta_3 x_l}.
\end{aligned}
\end{equation}
It then remains to solve for $x_r$ and $x_l$, which can be derived from $V_c^\prime(x_r) = \phi$ and $V_c^\prime(x_l) = 1$, that is,
\begin{equation}\label{eq:xrxl}
\left\{\begin{aligned}
& A_2 e^{\theta_3 x_r} + A_3 e^{\theta_4 x_r} = \phi,\\
& A_2 e^{\theta_3 x_l} + A_3 e^{\theta_4 x_l} = 1.
\end{aligned}\right.
\end{equation}
We provide the conditions for the existence and uniqueness of $x_r$ and $x_l$ in the following.
\begin{proposition}\label{thm:xrxl}
\begin{enumerate}
\item[(i)] If $\ln\phi < \frac{1}{\theta4-\theta3} \left( \theta_4 \ln\frac{\theta_5-\theta_4}{\theta_3} + \theta_3 \ln\frac{-\theta_4}{\theta_3-\theta_5} \right)$, there exists a unique pair $(x_r, x_l) \in (0, \infty)\times[x_r, \infty)$ that solves \eqref{eq:xrxl}.
\item[(ii)] If $\ln\phi \geq \frac{1}{\theta4-\theta3} \left( \theta_4 \ln\frac{\theta_5-\theta_4}{\theta_3} + \theta_3 \ln\frac{-\theta_4}{\theta_3-\theta_5} \right)$ and $\frac{\bar l}{\rho} + \frac{1}{\theta_5} > 0$, there doesn't exist positive $x_r$ that satisfies \eqref{eq:xrxl}. Instead, let $x_r = 0$, there exists a unique positive $x_l$ that satisfies \eqref{eq:xrxl}, that is, $x_l = \frac{1}{\theta_3 - \theta_4} \ln\frac{\theta_4(\theta_5-\theta_4)}{\theta_3(\theta_5-\theta_3)}$.
\item[(iii)] If $\frac{\bar l}{\rho} + \frac{1}{\theta_5} \leq 0$, there doesn't exist $(x_r, x_l) \in (0, \infty)\times[x_r, \infty)$ that solves \eqref{eq:xrxl}.
\end{enumerate}
\end{proposition}

Based on the condition in Proposition \ref{thm:xrxl}, we formally establish the solution to HJB equation \eqref{eq:Vc_HJB} under three cases accordingly.
\begin{theorem}\label{thm:Vc_classification}
Let $V_c(x)$ denote the value function of the dividend payment and capital injection problem \eqref{eq:classical_objective} and ${\bf u^*} = \{u_t^*\}_{t\ge0}$ be the optimal strategy. Recall the definitions in \eqref{eq:thetas}-\eqref{eq:Vc_BCD}.
\begin{enumerate}
\item[(i)] If $\ln\phi < \frac{1}{\theta_4-\theta_3} \left( \theta_4 \ln\frac{\theta_5-\theta_4}{\theta_3} + \theta_3 \ln\frac{-\theta_4}{\theta_3-\theta_5} \right)$, then $V_c(x)$ is given by \eqref{eq:Vc_analytic} and
$$
u_t^* = (l^*(X_t^{\bf u^*}), r^*(X_t^{\bf u^*})) = \left\{\begin{aligned}
& (0, \bar r), & X_t^{\bf u^*} \in (0, x_r),\\
& (0, 0), \qquad & X_t^{\bf u^*} \in [x_r, x_l),\\
& (\bar l, 0), & X_t^{\bf u^*} \in [x_l, \infty),
\end{aligned}\right.
$$
where $(x_r, x_l)$ is uniquely determined by \eqref{eq:xrxl}.

\item[(ii)] If $\ln\phi \geq \frac{1}{\theta_4-\theta_3} \left( \theta_4 \ln\frac{\theta_5-\theta_4}{\theta_3} + \theta_3 \ln\frac{-\theta_4}{\theta_3-\theta_5} \right)$ and $\frac{\bar l}{\rho} + \frac{1}{\theta_5} > 0$, then
$$
V_c(x) = \begin{cases}
A_2 \left(e^{\theta_3 x} - e^{\theta_4 x}\right), & x \in (0, x_l),\\
\frac{\bar l}{\rho} + A_4 e^{\theta_5 x}, & x \in [x_l, \infty),
\end{cases}
$$
and
$$
u_t^* = (l^*(X_t^{\bf u^*}), r^*(X_t^{\bf u^*})) = \left\{\begin{aligned}
& (0, 0), \qquad & X_t^{\bf u^*} \in (0, x_l),\\
& (\bar l, 0), & X_t^{\bf u^*} \in [x_l, \infty),
\end{aligned}\right.
$$
where $x_r = 0$ and $x_l = \frac{1}{\theta_3 - \theta_4} \ln\frac{\theta_4(\theta_5-\theta_4)}{\theta_3(\theta_5-\theta_3)}$.

\item[(iii)] If $\frac{\bar l}{\rho} + \frac{1}{\theta_5} \leq 0$, then
$$
V_c(x) = \frac{\bar l}{\rho} + A_4 e^{\theta_5 x}, \qquad x \in (0, \infty),
$$
and
$$
u_t^* = (\bar l, 0),
$$
where $x_r = x_l =0$.
\end{enumerate}
\end{theorem}

\subsection{Equilibrium strategy and value function}\label{subsec:time-inconsistent}
Now, we solve the general time-inconsistent problem \eqref{eq:objective} under a non-exponential discount function $\delta(s,\cdot)$.
Denote $V$ as the equilibrium value function and ${\bf \hat u}$ the equilibrium strategy, if exists, that is,
\begin{equation}\label{eq:equilibrium V}
V(s, x) := J(s, x, {\bf\hat u}) = \expect_{s,x}\left[\int_s^{\tau_{s,x}^{\bf \hat u}} \delta(s,t) \left( \hat l_t - \phi \hat r_t \right) dt \right].
\end{equation}

\begin{proposition}\label{prop:V_bounded}
The equilibrium value function $V(s, x)$ is bounded:
\begin{align*}
 V(s, x) \leq \int_s^\infty \delta(s,t) \bar l dt, \quad \forall s \ge 0, x \ge 0.
\end{align*}
\end{proposition}

Suppose the equilibrium value function $V(s,x) \in {\cal C}^{1,2}$. The following extended HJB equation system characterizes the equilibrium value function $V$. For any $s \ge k \ge 0$ and $x > 0$,
\begin{align}
&  \sup_{l \in [0,\bar l], r \in [0,\bar r]} \Big\{ V_s(s,x) + (\mu - l + r) V_x(s,x) + \frac{1}{2} \sigma^2 V_{xx}(s,x) + (l - \phi r) \notag \\
&\quad - f_s(s,x) - (\mu - l + r) f_x(s,x) - \frac{1}{2} \sigma^2 f_{xx}(s,x) \notag \\
&\quad + g_s(s,x;s) + (\mu - l + r) g_x(s,x;s) + \frac{1}{2} \sigma^2 g_{xx}(s,x;s) \Big\} = 0, \label{extended_HJB:1} \\
& g_s(s,x;k) + (\mu - \hat l_s + \hat r_s) g_x(s,x;k) + \frac{1}{2} \sigma^2 g_{xx}(s,x;k) + \delta(k,s) (\hat l_s - \phi \hat r_s) = 0, \label{extended_HJB:2} \\
& f(s,x) = g(s,x;s), \label{extended_HJB:3} \\
& V(s,0) = 0,\label{extended_HJB:bound1}\\
& g(s,0;k) = 0, \label{extended_HJB:bound2}
\end{align}
where $\hat u_s = (\hat l_s,\hat r_s)$ solves the supremum in \eqref{extended_HJB:1}.
We present the verification theorem below to show that $V(s, x)$ solving \eqref{extended_HJB:1}-\eqref{extended_HJB:bound2} is indeed an equilibrium value function.

\begin{theorem}\label{thm:verification_equilibrium}
Suppose that there exists $V, g, f \in {\cal C}^{1,2}$ and $\hat u_s$ solving the extended HJB equations \eqref{extended_HJB:1}-\eqref{extended_HJB:bound2}, and $\left\vert g \right\vert$, $\left\vert f \right\vert$, $\left\vert V \right\vert$, $\left\vert g_x \right\vert$, $\left\vert f_x \right\vert$, $\left\vert V_x \right\vert$ are bounded. Then $V$ is the equilibrium value function and $\hat u_s$ is the equilibrium strategy at time $s$ and surplus level $x$.
\end{theorem}

Note that the boundedness of $V$ is shown in Proposition \ref{prop:V_bounded}. It is also straightforward to verify that $g$ and $f$ are bounded via probabilistic interpretations shown in the proof of Theorem \ref{thm:verification_equilibrium} in Appendix \ref{proof}. In other words, the assumptions in Theorem \ref{thm:verification_equilibrium} are satisfied. We then simplify the extended HJB equation system \eqref{extended_HJB:1}-\eqref{extended_HJB:bound2} into
\begin{align}\label{simplified_extended_HJB}
& \sup_{l \in [0,\bar l], r \in [0,\bar r]} \left\{ V_s(s,x) + (\mu - l + r) V_x(s,x) + \frac{1}{2} \sigma^2 V_{xx}(s,x) + (l - \phi r)\right\}\notag \\
&\quad = -\expect_{s, x} \left[ \int_s^{\tau_{s, x}^{\bf \hat u}}  \delta_s(s,t) \left(\hat l_t - \phi \hat r_t \right) dt \right],
\end{align}
with boundary condition
\begin{equation}\label{simplified_extended_HJB:bound}
V(s,0) = 0,
\end{equation}
where $\hat u_s = (\hat l_s,\hat r_s)$ solves the supremum on the left hand side of \eqref{simplified_extended_HJB}.

\begin{proposition}\label{prop:simplified_HJB}
Suppose that there exists $V(s,x) \in {\cal C}^{1,2}$ and $\hat u_s$ solving \eqref{simplified_extended_HJB} with boundary condition \eqref{simplified_extended_HJB:bound}, and $\left\vert V_x \right\vert$ is bounded. Then $V$ is the equilibrium value function and $\hat u_s$ is the equilibrium strategy at time $s$ and surplus level $x$.
\end{proposition}

Solving the supremum on the left-hand side of the simplified extended HJB equation \eqref{simplified_extended_HJB}, we obtain
\begin{equation}\label{eq:equilibrium_strategy}
\hat u_s = (\hat l_s, \hat r_s) = \left\{\begin{aligned}
& (0, \bar r), & V_x(s, x) > \phi,\\
& (0, 0), \qquad & 1 < V_x(s, x) \leq \phi,\\
& (\bar l, 0), & V_x(s, x) \leq 1.
\end{aligned}\right.
\end{equation}

\section{An explicit solution}\label{se:pseudo}
In this section, we give an explicit solution to \eqref{simplified_extended_HJB} under a pseudo-exponential discount function, which is
\begin{equation}\label{eq:pseudo_exp_discount}
\delta(s, t) = \omega e^{-\rho_1 (t-s)} + (1-\omega) e^{-\rho_2 (t-s)},
\end{equation}
where $0 \leq \omega \leq 1$. Without loss of generality, let $0 < \rho_1 \le \rho_2$.
The pseudo-exponential discount function \eqref{eq:pseudo_exp_discount} was first considered in \cite{ekeland2006being} and \cite{ekeland2008investment} who explained the rationality of such a discount function. Recently, \cite{ebert2020weighted} stated that there could be inter-personal disagreement about the discount rate in a group. Thus,
\eqref{eq:pseudo_exp_discount} can also be viewed as a special case of weighted discounting where the insurance company consists of two groups of shareholders with different discount rates.

Since the problem is time-homogeneous, the equilibrium value function independent of time $s$.
Thus, the equilibrium value function is written as $V(x)$ thereafter.
Inspired by the time-consistent benchmark that the value function is concave, we assume that $V(x)$ is concave, which will be proved a posteriori, and there exists $0 < x_1 \leq x_2$ such that $V'(x) \leq 1$ when $x \geq x_2$, $1 < V'(x) \leq \phi$ when $x_1 \leq x < x_2$, $V'(x) > \phi$ when $0 < x < x_1$. Then the equilibrium strategy \eqref{eq:equilibrium_strategy} is the following feedback control:
\begin{equation}\label{eq:equilibrium_feedback_strategy}
\hat u_s = (\hat l(x), \hat r(x)) = \left\{\begin{aligned}
& (0, \bar r), & 0 < x < x_1,\\
& (0, 0), \qquad & x_1 \leq x < x_2,\\
& (\bar l, 0), & x \geq x_2,
\end{aligned}\right.
\end{equation}
where $\hat l: (0, \infty) \mapsto [0, \bar l]$ and $\hat r: (0, \infty) \mapsto [0, \bar r]$ are two mappings.

By definition \eqref{eq:equilibrium V}, the equilibrium value function $V(x)$ under pseudo-exponential discounting can be split into two:
\begin{align*}
V(x) & = \expect_{0,x}\left[\int_0^{\tau_{0,x}^{\bf \hat u}} (\omega e^{-\rho_1 t} + (1-\omega)e^{-\rho_2 t}) \left( \hat l(X^{\bf \hat u}_{t}) - \phi \hat r(X^{\bf \hat u}_{t}) \right) dt \right]\\
& = \omega \expect_{0,x}\left[\int_0^{\tau_{0,x}^{\bf \hat u}} e^{-\rho_1 t} \left( \hat l(X^{\bf \hat u}_{t}) - \phi \hat r(X^{\bf \hat u}_{t}) \right) dt \right]\\
& \quad + (1-\omega) \expect_{0,x}\left[\int_0^{\tau_{0,x}^{\bf \hat u}} e^{-\rho_2 t} \left( \hat l(X^{\bf \hat u}_{t}) - \phi \hat r(X^{\bf \hat u}_{t}) \right) dt \right]\\
& = \omega V_1(x) + (1 - \omega) V_2(x),
\end{align*}
where
\begin{equation}\label{eq:Vi}
V_i(x) := \expect_{0,x}\left[\int_0^{\tau_{0,x}^{\bf \hat u}} e^{-\rho_i t} \left( \hat l(X^{\bf \hat u}_{t}) - \phi \hat r(X^{\bf \hat u}_{t}) \right) dt \right], \quad i = 1, 2.
\end{equation}
The proposition below follows from Proposition \ref{prop:V_bounded}, which helps to prove Proposition \ref{prop:Vi}.

\begin{proposition}\label{corollary:Vi_bounded}
Given $V_i$ defined in \eqref{eq:Vi}, $i = 1, 2$. Then $V_i$, $i = 1, 2$, are bounded:
\begin{align*}
V_i(x) \leq \frac{\bar l}{\rho_i}, \quad \forall x \in [0, \infty).
\end{align*}
\end{proposition}

For fixed $0 < x_1 \le x_2$, consider the following ODEs:
\begin{equation}\label{eq:Vi_ODE}
\begin{cases}
\frac{1}{2} \sigma^2 v^{\prime\prime}_i(x) + (\mu + \bar r) v^\prime_i(x) - \rho_i v_i(x) - \phi \bar r = 0, & x \in (0, x_1),\\
\frac{1}{2} \sigma^2 v^{\prime\prime}_i(x) + \mu v^\prime_i(x) - \rho_i v_i(x) = 0, \qquad & x \in [x_1, x_2),\\
\frac{1}{2} \sigma^2 v^{\prime\prime}_i(x) + (\mu - \bar l) v^\prime_i(x) - \rho_i v_i(x) + \bar l = 0, & x \in [x_2, \infty).
\end{cases}
\end{equation}
with boundary condition
\begin{equation}\label{Vi_boundary}
v_i(0) = 0, \quad i = 1, 2.
\end{equation}
The following proposition asserts that solving $v_i$ from \eqref{eq:Vi_ODE} indeed solves the equilibrium value function.
\begin{proposition}\label{prop:Vi}
Suppose that there exist $0 < x_1 \le x_2$ such that the equilibrium strategy is given by \eqref{eq:equilibrium_feedback_strategy} and there exist functions $v_i(x)\in {\cal C}^2$ with bounded first-order derivatives that solve \eqref{eq:Vi_ODE}-\eqref{Vi_boundary}. Then $v_i = V_i$ which satisfies the probabilistic interpretation \eqref{eq:Vi} and $V(x) = \omega v_1(x) + (1-\omega) v_2(x)$ is the equilibrium value function.
\end{proposition}

Solving the ODEs \eqref{eq:Vi_ODE}, we have
\begin{equation*}
V_i(x) = v_i(x) = \left\{\begin{aligned}
& -\frac{\phi \bar r}{\rho_i} + A_{i1} e^{\theta_i^+(\mu+\bar r) x} + B_{i1} e^{\theta_i^-(\mu+\bar r) x},\qquad & x \in (0, x_1),\\
& A_{i2} e^{\theta_i^+(\mu) x} + B_{i2} e^{\theta_i^-(\mu) x}, & x \in [x_1, x_2),\\
& \frac{\bar l}{\rho_i} + A_{i3} e^{\theta_i^+(\mu-\bar l) x} + B_{i3} e^{\theta_i^-(\mu-\bar l) x}, & x \in [x_2, \infty),
\end{aligned}\right.
\end{equation*}
where
$$
\theta_i^+(z) := \frac{-z + \sqrt{z^2 + 2 \rho_i \sigma^2}}{\sigma^2} > 0,\quad \theta_i^-(z) := \frac{-z - \sqrt{z^2 + 2 \rho_i \sigma^2}}{\sigma^2} < 0,
$$
where $A_{ij}$ and $B_{ij}$, $i = 1,2$, $j = 1,2,3$, are the constants to be determined.
Due to Proposition \ref{corollary:Vi_bounded}, $V_i(x)$ is bounded, which then implies that $A_{i3} = 0$.
To determine the value of $A_{i1}$, $A_{i2}$, $B_{i1}$, $B_{i2}$, $B_{i3}$, we use the boundary condition \eqref{Vi_boundary} and smooth fit principle, that is,
$$
V_i(0) = 0,\ V_i(x_1 -) = V_i(x_1), \ V^\prime_i(x_1 -) = V^\prime_i(x_1),\ V_i(x_2 -) = V_i(x_2),\ V^\prime_i(x_2 -) = V^\prime_i(x_2).
$$
Then $A_{ij}$ and $B_{ij}$ satisfy the following linear equation system:
\begin{equation*}
\left\{\begin{aligned}
& -\frac{\phi \bar r}{\rho_i} + A_{i1} + B_{i1} = 0, \\
& -\frac{\phi \bar r}{\rho_i} + A_{i1} e^{\theta_{i1} x_1} + B_{i1} e^{\theta_{i2} x_1} = A_{i2} e^{\theta_{i3} x_1} + B_{i2} e^{\theta_{i4} x_1}, \\
& A_{i1} \theta_{i1} e^{\theta_{i1} x_1} + B_{i1} \theta_{i2} e^{\theta_{i2} x_1} = A_{i2} \theta_{i3} e^{\theta_{i3} x_1} + B_{i2} \theta_{i4} e^{\theta_{i4} x_1}, \\
& A_{i2} e^{\theta_{i3} x_2} + B_{i2} e^{\theta_{i4} x_2} = \frac{\bar l}{\rho_i} + B_{i3} e^{\theta_{i5} x_2},\\
& A_{i2} \theta_{i3} e^{\theta_{i3} x_2} + B_{i2} \theta_{i4} e^{\theta_{i4} x_2} = B_{i3} \theta_{i5}  e^{\theta_{i5} x_2},
\end{aligned}\right.
\end{equation*}
where
\begin{equation}\label{eq:theta_i}
\theta_{i1} = \theta_i^+(\mu+\bar r),\ \theta_{i2} = \theta_i^-(\mu+\bar r),\ \theta_{i3} = \theta_i^+(\mu),\ \theta_{i4} = \theta_i^-(\mu),\ \theta_{i5} = \theta_i^-(\mu-\bar l).
\end{equation}
Solving the above system of equations,
\begin{equation}\label{eq:ABs}
\begin{aligned}
&  A_{i1}  = -\frac{C_{i1} \bar l + D_{i1} \phi \bar r}{\rho_i E_i},  && A_{i2}  = -\frac{C_{i2} \bar l  + D_{i2} \phi \bar r}{\rho_i E_i},&&\\
&  B_{i1}  = \frac{C_{i3} \bar l + D_{i3} \phi \bar r}{\rho_i E_i}, && B_{i2}  = \frac{C_{i4} \bar l + D_{i4} \phi \bar r}{\rho_i E_i}, &&  B_{i3}  = - e^{-\theta_{i5} x_2} \frac{C_{i5} \bar l + D_{i5} \phi \bar r}{\rho_i E_i},
\end{aligned}
\end{equation}
where
\begin{align}\label{eq:C}
C_{i1} &= C_{i3} = \theta_{i5} (\theta_{i3} - \theta_{i4}) e^{(\theta_{i3} + \theta_{i4}) x_1}, \notag \\
C_{i2} &= \theta_{i5} \left( (\theta_{i1} - \theta_{i4}) e^{(\theta_{i1} + \theta_{i4}) x_1} + (\theta_{i4} - \theta_{i2}) e^{(\theta_{i2} + \theta_{i4}) x_1} \right), \notag\\
C_{i4} &= \theta_{i5} \left((\theta_{i1} - \theta_{i3}) e^{(\theta_{i1}+\theta_{i3}) x_1} + (\theta_{i3} - \theta_{i2}) e^{(\theta_{i2}+\theta_{i3}) x_1}\right), \notag \\
C_{i5} &= \theta_{i3} (\theta_{i1} - \theta_{i4}) e^{(\theta_{i1}+\theta_{i4}) x_1 + \theta_{i3} x_2} + \theta_{i4} (\theta_{i3} - \theta_{i1}) e^{(\theta_{i1}+\theta_{i3}) x_1 + \theta_{i4} x_2} \notag \\
&\quad + \theta_{i3} (\theta_{i4} - \theta_{i2}) e^{(\theta_{i2}+\theta_{i4}) x_1 + \theta_{i3} x_2} + \theta_{i4} (\theta_{i2} - \theta_{i3}) e^{(\theta_{i2}+\theta_{i3}) x_1 + \theta_{i4} x_2},
\end{align}
\begin{align}\label{eq:D}
D_{i1} &= \theta_{i3}(\theta_{i5} - \theta_{i4}) e^{\theta_{i3} x_1 + \theta_{i4} x_2} + \theta_{i4} (\theta_{i3}  - \theta_{i5}) e^{\theta_{i4} x_1 + \theta_{i3} x_2} \notag\\
&\quad + (\theta_{i3} - \theta_{i2}) (\theta_{i4} - \theta_{i5}) e^{(\theta_{i2} + \theta_{i3}) x_1+ \theta_{i4} x_2} + (\theta_{i2} - \theta_{i4}) (\theta_{i3} - \theta_{i5}) e^{(\theta_{i2} + \theta_{i4}) x_1 + \theta_{i3} x_2}, \notag \\
D_{i2} &= \theta_{i1} (\theta_{i5} - \theta_{i4}) e^{\theta_{i1} x_1 + \theta_{i4} x_2} + \theta_{i2} (\theta_{i4} - \theta_{i5}) e^{\theta_{i2} x_1 + \theta_{i4} x_2} \notag \\
&\quad + (\theta_{i1} - \theta_{i2}) (\theta_{i4} - \theta_{i5}) e^{(\theta_{i1} + \theta_{i2}) x_1 + \theta_{i4} x_2}, \notag \\
D_{i3} &= \theta_{i3}(\theta_{i5} - \theta_{i4}) e^{\theta_{i3} x_1 + \theta_{i4} x_2} + \theta_{i4} (\theta_{i3} - \theta_{i5}) e^{\theta_{i4} x_1 + \theta_{i3} x_2} \notag \\
&\quad + ( \theta_{i3} - \theta_{i1}) (\theta_{i4} - \theta_{i5}) e^{(\theta_{i1}+\theta_{i3}) x_1 + \theta_{i4} x_2} + (\theta_{i1} - \theta_{i4}) (\theta_{i3} - \theta_{i5}) e^{(\theta_{i1}+\theta_{i4}) x_1 + \theta_{i3} x_2}, \notag \\
D_{i4} &=  \theta_{i1}(\theta_{i5} - \theta_{i3}) e^{\theta_{i1} x_1 + \theta_{i3} x_2} + \theta_{i2}(\theta_{i3} - \theta_{i5}) e^{\theta_{i2} x_1 + \theta_{i3} x_2} \notag \\
&\quad  + (\theta_{i1} - \theta_{i2}) (\theta_{i3} - \theta_{i5}) e^{(\theta_{i1}+\theta_{i2}) x_1 + \theta_{i3} x_2}, \notag \\
D_{i5} &= \theta_{i1} (\theta_{i3} - \theta_{i4}) e^{\theta_{i1} x_1 + (\theta_{i3}+\theta_{i4}) x_2} + \theta_{i2} (\theta_{i4} - \theta_{i3}) e^{\theta_{i2} x_1 + (\theta_{i3}+\theta_{i4}) x_2} \notag \\
&\quad + (\theta_{i1} -\theta_{i2}) (\theta_{i4} - \theta_{i3}) e^{(\theta_{i1}+\theta_{i2}) x_1 + (\theta_{i3}+\theta_{i4}) x_2},
\end{align}
and
\begin{align}\label{eq:CDEs}
E_i  &= (\theta_{i3} - \theta_{i1}) (\theta_{i4} - \theta_{i5}) e^{(\theta_{i1}+\theta_{i3}) x_1 + \theta_{i4} x_2}  + (\theta_{i1} - \theta_{i4}) (\theta_{i3} - \theta_{i5}) e^{(\theta_{i1} + \theta_{i4}) x_1 + \theta_{i3} x_2} \notag \\
&\quad + (\theta_{i2} - \theta_{i3}) (\theta_{i4} - \theta_{i5}) e^{(\theta_{i2}+\theta_{i3} ) x_1 + \theta_{i4} x_2} +(\theta_{i4} - \theta_{i2}) (\theta_{i3} - \theta_{i5}) e^{(\theta_{i2}+\theta_{i4}) x_1 + \theta_{i3} x_2}.
\end{align}

\begin{theorem}\label{thm:V_concave}
Suppose that there exist $0 < x_1 \le x_2$ such that the equilibrium strategy is given by \eqref{eq:equilibrium_feedback_strategy}. Then
\begin{equation}\label{eq:Vi_solution}
V_i(x) = \left\{\begin{aligned}
& -\frac{\phi \bar r}{\rho_i} + A_{i1} e^{\theta_{i1} x} + B_{i1} e^{\theta_{i2} x}, \qquad & x \in (0, x_1),\\
& A_{i2} e^{\theta_{i3} x} + B_{i2} e^{\theta_{i4} x}, & x \in [x_1, x_2),\\
& \frac{\bar l}{\rho_i} + B_{i3} e^{\theta_{i5} x}, & x \in [x_2, \infty),
\end{aligned}\right.
\end{equation}
$i = 1,2$, with $A_{ij}$, $B_{ij}$ and $\theta_{ij}$ defined in \eqref{eq:theta_i}-\eqref{eq:CDEs}. In addition, $V(x) = \omega V_1(x) + (1-\omega) V_2(x)$ is increasing and concave.
\end{theorem}

It remains to determine the values of $x_1$ and $x_2$, which are solved from $V^\prime(x_1) = \phi$ and $V^\prime(x_2) = 1$, that is,
\begin{equation}\label{eq:x1x2}
\left\{\begin{aligned}
& \omega \left(A_{12} \theta_{13} e^{\theta_{13} x_1} + B_{12} \theta_{14} e^{\theta_{14} x_1}\right) + (1-\omega) \left(A_{22} \theta_{23} e^{\theta_{23} x_1} + B_{22} \theta_{24} e^{\theta_{24} x_1}\right) = \phi,\\
& \omega \left(A_{12} \theta_{13} e^{\theta_{13} x_2} + B_{12} \theta_{14} e^{\theta_{14} x_2}\right) + (1-\omega) \left(A_{22} \theta_{23} e^{\theta_{23} x_2} + B_{22} \theta_{24} e^{\theta_{24} x_2}\right) = 1,
\end{aligned}\right.
\end{equation}
where $A_{ij}$, $B_{ij}$, and $\theta{ij}$ are given in \eqref{eq:theta_i}-\eqref{eq:CDEs}.
With the concavity of $V$ shown in Theorem \ref{thm:V_concave}, we classify the results into three cases: (i) $0<x_1\le x_2$; (ii) $x_1=0$, $x_2 > 0$; (iii) $x_1 = x_2 = 0$. To be more precise, if there exists $0 < x_1 \le x_2$ that solve \eqref{eq:x1x2}, then the equilibrium value function $V = \omega V_1 + (1-\omega)V_2$, where $V_1$ and $V_2$ given by \eqref{eq:Vi_solution} are three-stage functions. Otherwise if there does not exist such $x_1$, $x_2$ to solve \eqref{eq:x1x2}, let $x_1 = 0$ and if there exists $x_2 > 0$ that solves
\begin{equation}\label{eq:x2}
\omega \left(A_{12} \theta_{13} e^{\theta_{13} x_2} + B_{12} \theta_{14} e^{\theta_{14} x_2}\right) + (1-\omega) \left(A_{22} \theta_{23} e^{\theta_{23} x_2} + B_{22} \theta_{24} e^{\theta_{24} x_2}\right) = 1,
\end{equation}
then $V_1$ and $V_2$ degenerate to two-stage functions. Otherwise if there does not exists such $x_2$ to solve \eqref{eq:x2}, then let $x_2 = 0$ and $V_i(x) = \frac{\bar l}{\rho_i} + B_{i3} e^{\theta_{i5} x}$ for $x\in[0, \infty)$. We summarize the equilibrium value functions and policies in the theorem below.

\begin{theorem}\label{thm:V_classification}
\begin{enumerate}
\item[(i)] Suppose there exist $0 < x_1 \le x_2 < \infty$ that solve \eqref{eq:x1x2}, with $A_{ij}$, $B_{ij}$ and $\theta_{ij}$ defined in \eqref{eq:theta_i}-\eqref{eq:CDEs}.
Then the value function $V(x) = \omega V_1(x) + (1-\omega)V_2(x)$, where $V_i(x)$ is given by \eqref{eq:Vi_solution}, $i = 1,2$. The equilibrium strategy is
\begin{align*}
\hat u_t = (\hat l(X_t^{\bf \hat u}), \hat r(X_t^{\bf \hat u})) =
\begin{cases}
(0, \bar r), & X_t^{\bf \hat u} \in (0, x_1),\\
(0, 0), \qquad & X_t^{\bf \hat u} \in [x_1, x_2),\\
(\bar l, 0), & X_t^{\bf \hat u} \in [x_2,\infty).
\end{cases}
\end{align*}
\item[(ii)] Suppose there do not exist $0 < x_1 \le x_2 < \infty$ that solve \eqref{eq:x1x2}, and there exists $0 < x_2 < \infty$ that solves \eqref{eq:x2}, with $A_{ij}$, $B_{ij}$ and $\theta_{ij}$ defined in \eqref{eq:theta_i}-\eqref{eq:CDEs} and $x_1 = 0$.
Then the value function $V(x) = \omega V_1(x) + (1-\omega)V_2(x)$, where
$$
V_i(x) = \begin{cases}
A_{i2} \left(e^{\theta_{i3} x} - e^{\theta_{i4} x}\right), & x \in (0, x_2),\\
\frac{\bar l}{\rho_i} + B_{i3} e^{\theta_{i5} x}, & x \in [x_2, \infty),
\end{cases}
$$
$i = 1,2$. The equilibrium strategy is
$$
\hat u_t = (\hat{l}(X_t^{\bf \hat u}), \hat{r}(X_t^{\bf \hat u})) = \left\{\begin{aligned}
& (0, 0), \qquad & X_t^{\bf \hat u} \in (0, x_2),\\
& (\bar l, 0), & X_t^{\bf \hat u} \in [x_2, \infty).
\end{aligned}\right.
$$
\item[(iii)] Suppose there do not exist $0 < x_1 \le x_2 < \infty$ solve \eqref{eq:x1x2} nor exist $0 < x_2 < \infty$ that solves \eqref{eq:x2}.
Then $V(x) = \omega V_1(x) + (1-\omega)V_2(x)$, where
$$
V_i(x) = \frac{\bar l}{\rho_i} + B_{i3} e^{\theta_{i5} x},
$$
$i = 1,2$, with $A_{ij}$, $B_{ij}$ and $\theta_{ij}$ defined in \eqref{eq:theta_i}-\eqref{eq:CDEs} and $x_1 = x_2 = 0$. The equilibrium strategy is
$$
\hat u_t = (\hat{l}(X_t^{\bf \hat u}), \hat{r}(X_t^{\bf \hat u})) = (\bar l, 0).
$$
\end{enumerate}
\end{theorem}

Theorem \ref{thm:V_classification} case $(i)$ shows the equilibrium strategy is to inject capital at the maximal rate once the surplus is below the threshold $x_1$ and to pay the dividend at the maximal rate once the surplus exceeds the threshold $x_2$. For case $(ii)$ and $(iii)$, the equilibrium strategy does not require capital injection even if the insurance company is at the ruin time and thus, the problem is the same as the pure dividend problem without capital injection in \cite{zhao2014dividend}. It is straightforward to verify that the existence and uniqueness of a positive solution to \eqref{eq:x2} coincides with the condition in \citet[Lemma 4.1]{zhao2014dividend}: If $\omega \theta_{15} \bar l/\rho_1 + (1-\omega) \theta_{25} \bar l/\rho_2 + 1 < 0$, there exists a unique $x_2 > 0$ solving \eqref{eq:x2} and the equilibrium strategy belongs to case $(ii)$; otherwise, there does not exist a positive solution to \eqref{eq:x2} and the equilibrium strategy belongs to case $(iii)$.

\section{Numerical examples}\label{se:numericalexample}

In this section, we illustrate the equilibrium dividend and capital injection policy and value function with numerical examples.

Figure \ref{fig:equilibrium_V} plots equilibrium strategies and corresponding value functions according to Theorem \ref{thm:V_classification}. Let $\mu = 1$, $\sigma = 1$, $\omega = 0.3$, $\rho_1 = 0.6$, $\rho_2 = 1$, which are the parameter values commonly used in the literature.
Note that all the equilibrium value functions are increasing and concave in the surplus level and the equilibrium dividend and capital injection strategies are of threshold type.
First, for $\bar l = \bar r = 1$ and $\phi = 1.2$, one obtains that $x_1 = 0.1916$ and $x_2 = 0.317$, which implies that the equilibrium strategy is to inject capital when surplus is less than 0.1916 and pay dividends when surplus is greater than 0.317; see the left panels of Figure \ref{fig:equilibrium_V}.
Second, for $\bar l = \bar r = 1$ and $\phi = 2$, one obtains that $x_1 = 0$ and $x_2 = 0.3204$, which implies that there is no capital injection while dividend is paid when surplus is greater than 0.3204; see the middle panels of Figure \ref{fig:equilibrium_V}. This is because $\phi$ is relatively high, it is even better to let the insurance company go bankrupt rather than to inject capital.
Third, for $\bar l = 0.4$, $\bar r = 1$, and $\phi = 1.2$, one obtains that $x_1 =  x_2 = 0$, which implies that the equilibrium strategy is always to pay dividends at any positive surplus level; see the right panels of Figure \ref{fig:equilibrium_V}. This is because the maximal rate $\bar l$ compared with the discount rate is small and it is not economical to distribute the reserves afterward.

\begin{figure}[!htbp]
\begin{minipage}[t]{0.329\textwidth}
\centering
\includegraphics[width=\textwidth]{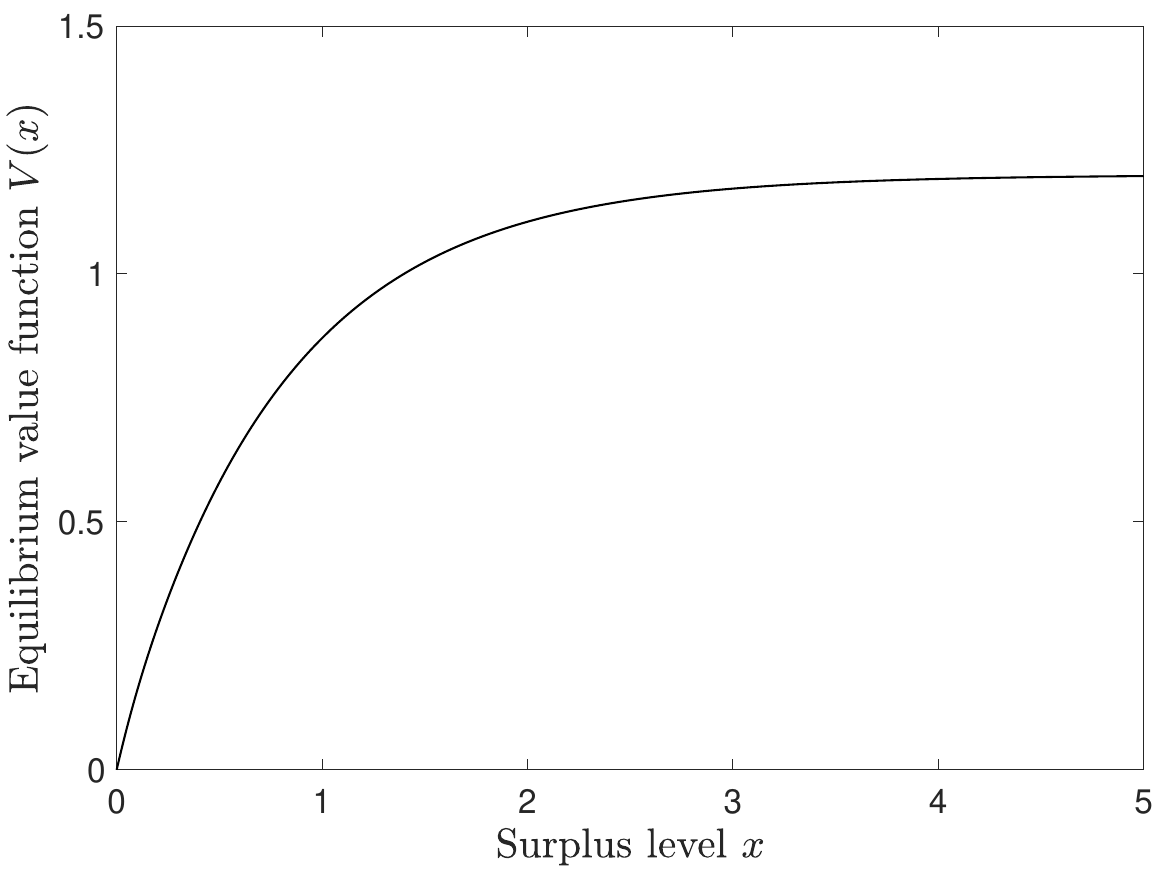}
\end{minipage}
\begin{minipage}[t]{0.329\textwidth}
\centering
\includegraphics[width=\textwidth]{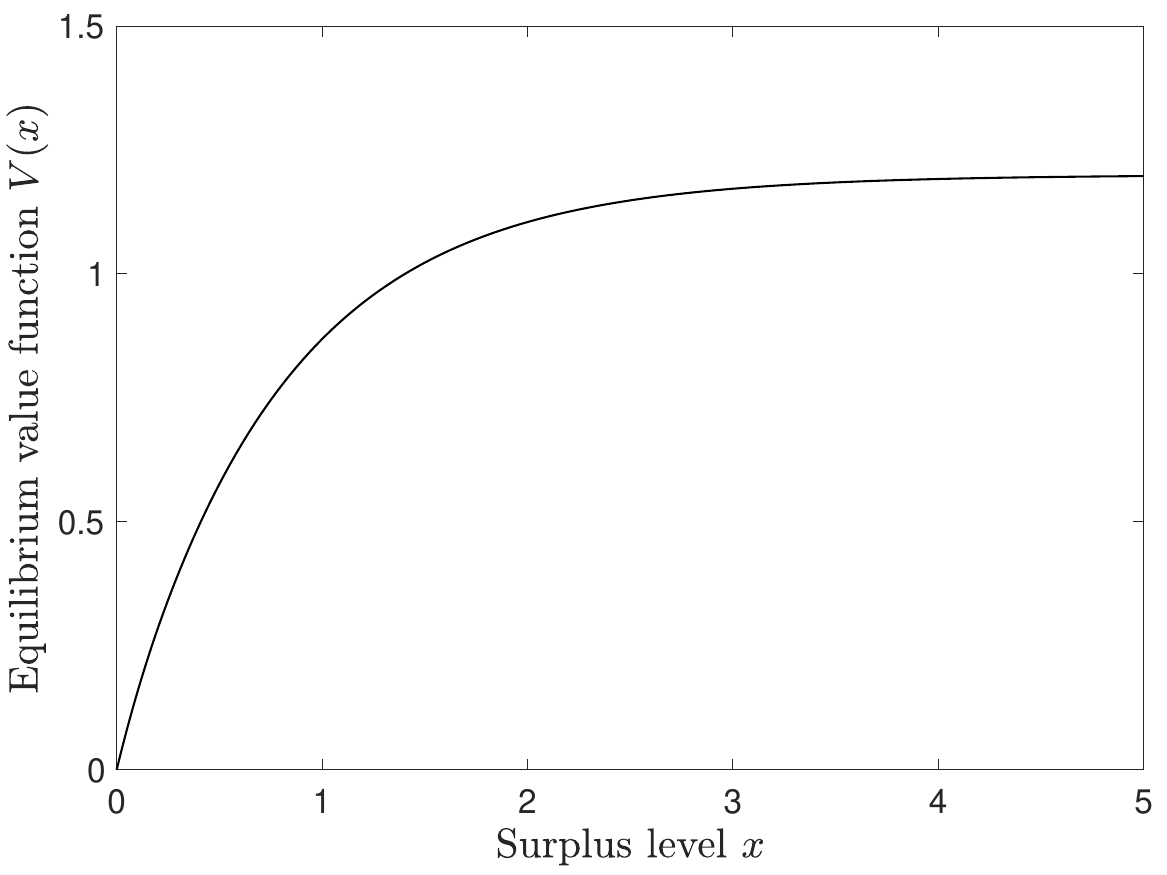}
\end{minipage}
\begin{minipage}[t]{0.329\textwidth}
\centering
\includegraphics[width=\textwidth]{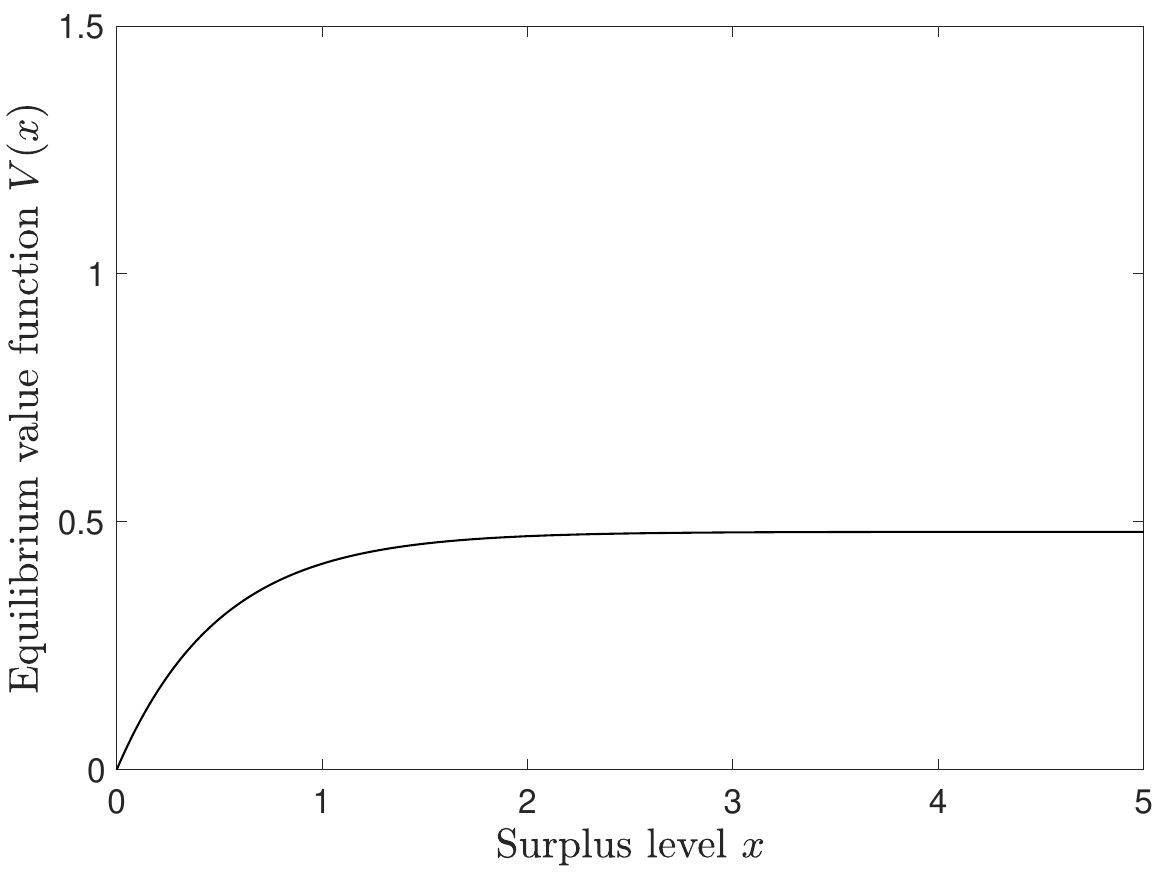}
\end{minipage}
\begin{minipage}[t]{0.329\textwidth}
\centering
\includegraphics[width=\textwidth]{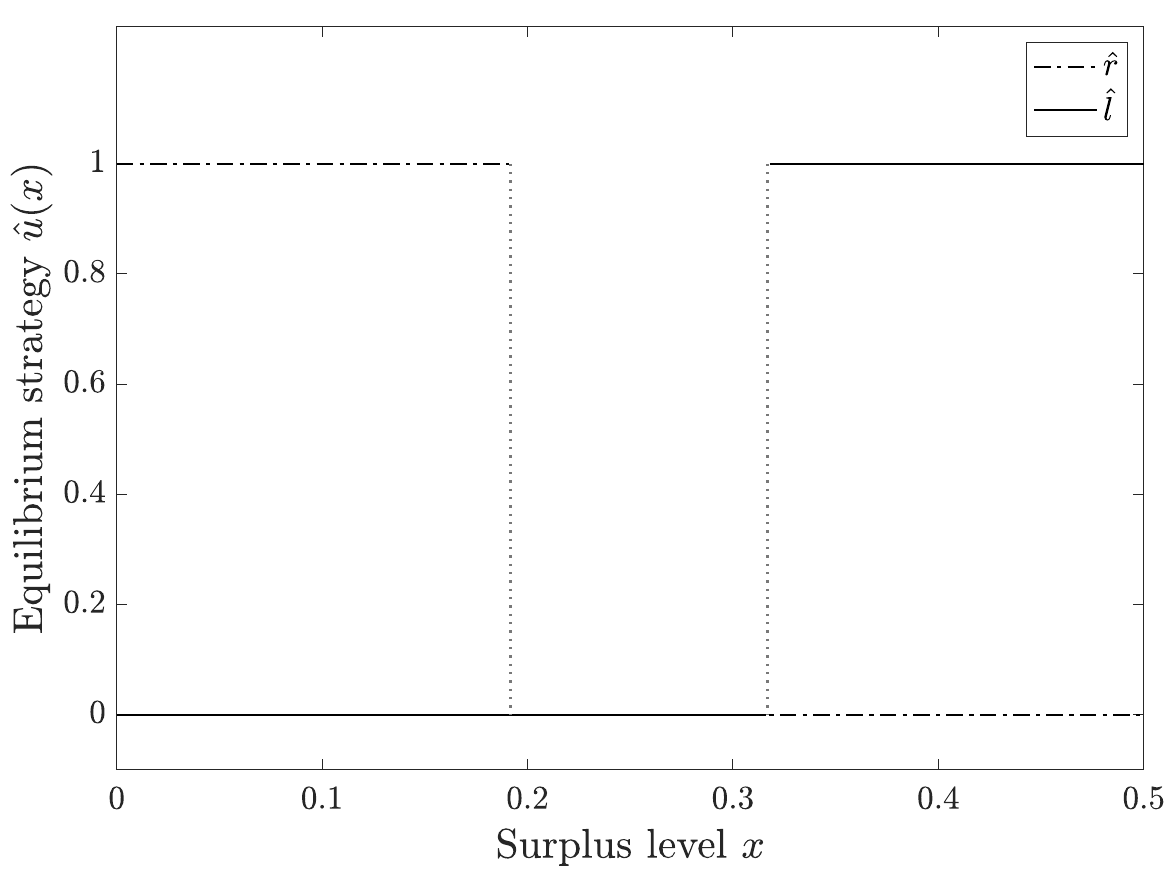}
\end{minipage}
\begin{minipage}[t]{0.329\textwidth}
\centering
\includegraphics[width=\textwidth]{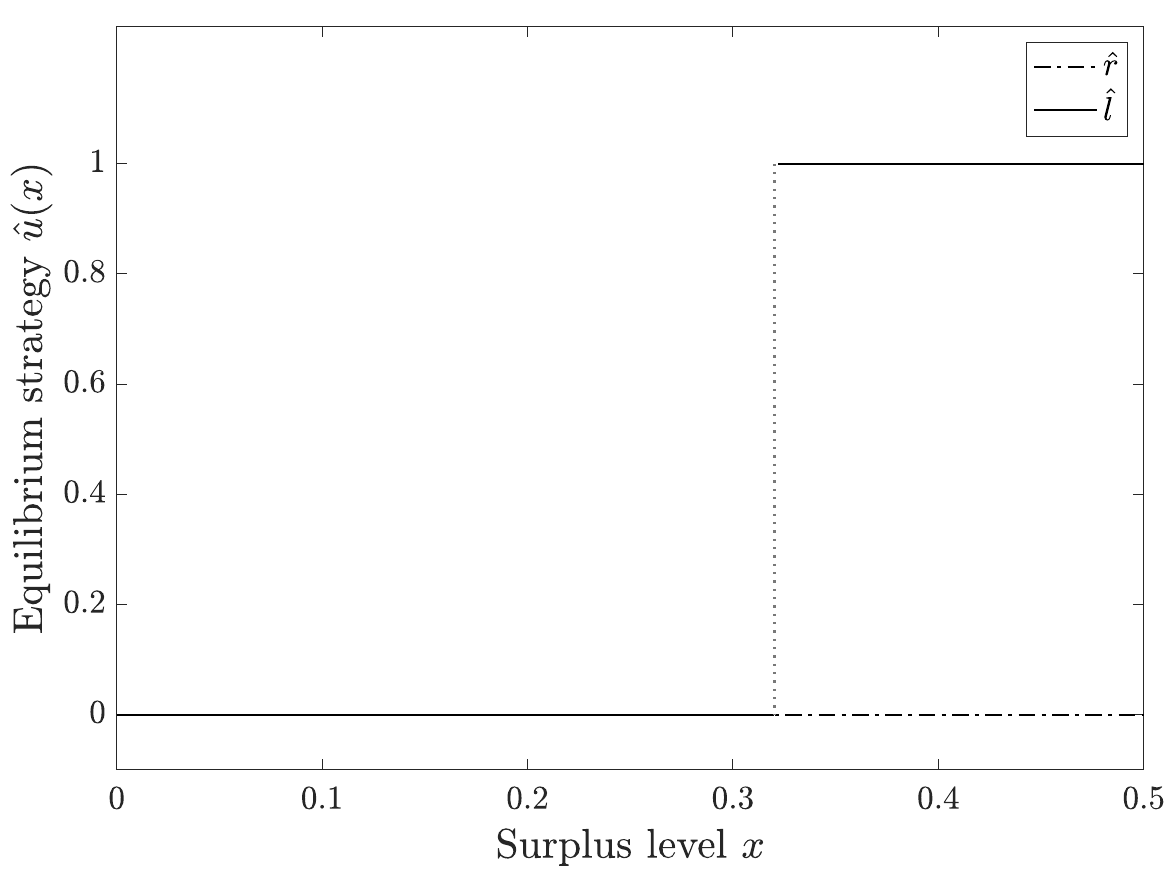}
\end{minipage}
\begin{minipage}[t]{0.329\textwidth}
\centering
\includegraphics[width=\textwidth]{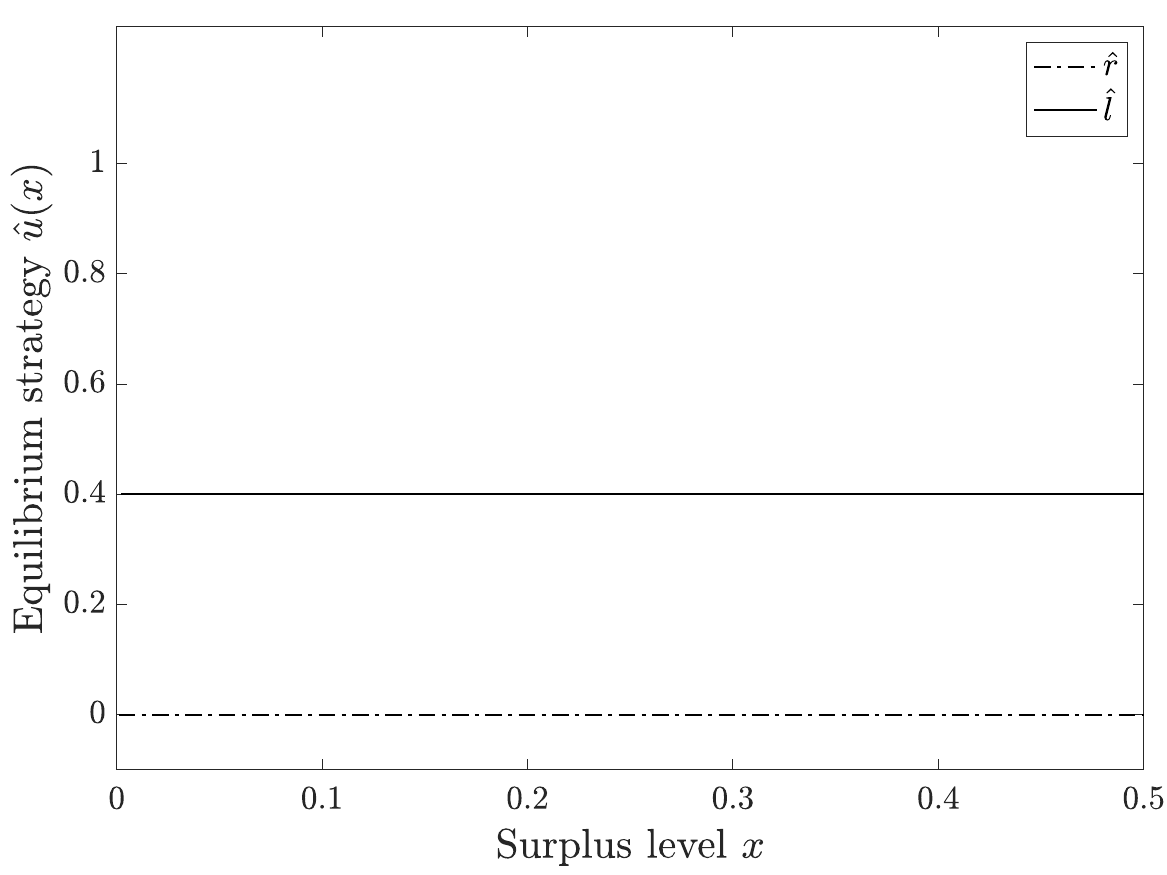}
\end{minipage}
\caption{The equilibrium value function $V$ (top) and the equilibrium strategy $\hat{u}(x)$ (bottom) for $\mu = \sigma = 1$, $\omega = 0.3$, $\rho_1 = 0.6$, $\rho_2 = 1$, $\bar l = \bar r = 1$, $\phi = 1.2$ (left panels), $\bar l = \bar r = 1$, $\phi = 2$ (middle panels), and $\bar l = 0.4$, $\bar r = 1$, $\phi = 1.2$ (right panels), respectively.}
\label{fig:equilibrium_V}
\end{figure}

Figure \ref{fig:different_l} left panel shows that for different combinations of maximal dividend rate $\bar l$ and capital injection cost $\phi$, there are three regions of the thresholds $x_1$ and $x_2$ which correspond to the three cases in Theorem \ref{thm:V_classification}. Let $\mu = \sigma = 1$, $\omega = 0.3$, $\rho_1 = 0.6$, $\rho_2 = 1$, $\bar r = 1$. First, if the maximal dividend rate $\bar l$ is less than 0.4, then the equilibrium strategy belongs to case $(iii)$, regardless of the capital injection cost, under which the dividend is always paid at the maximal rate due to the relatively low $\bar l$. Second, if the maximal dividend rate $\bar l$ is larger than 0.4 and the capital injection cost is relatively high, then the equilibrium strategy belongs to case $(ii)$, which means that there is no capital injection due to the high cost, while the dividend is paid as long as the surplus exceeds threshold $x_2$. Third, if the maximal rate $\bar l$ is larger than 0.4 and the capital injection cost is relatively low, then the equilibrium strategy belongs to case $(i)$, which means that the capital is injected once the surplus is below threshold $x_1$, dividend is paid if the surplus exceeds threshold $x_2$, and nothing did in between. The right panel of Figure \ref{fig:different_l} shows that how thresholds $x_1$ and $x_2$ are varying with respect to $\bar l$ when fixing $\phi = 1.2$. In particular, both $x_1$ and $x_2$ are increasing in $\bar l$ with the difference between two thresholds getting larger as well.
\begin{figure}[!htbp]
\centering
\begin{minipage}{0.45\linewidth}
\includegraphics[width=\linewidth]{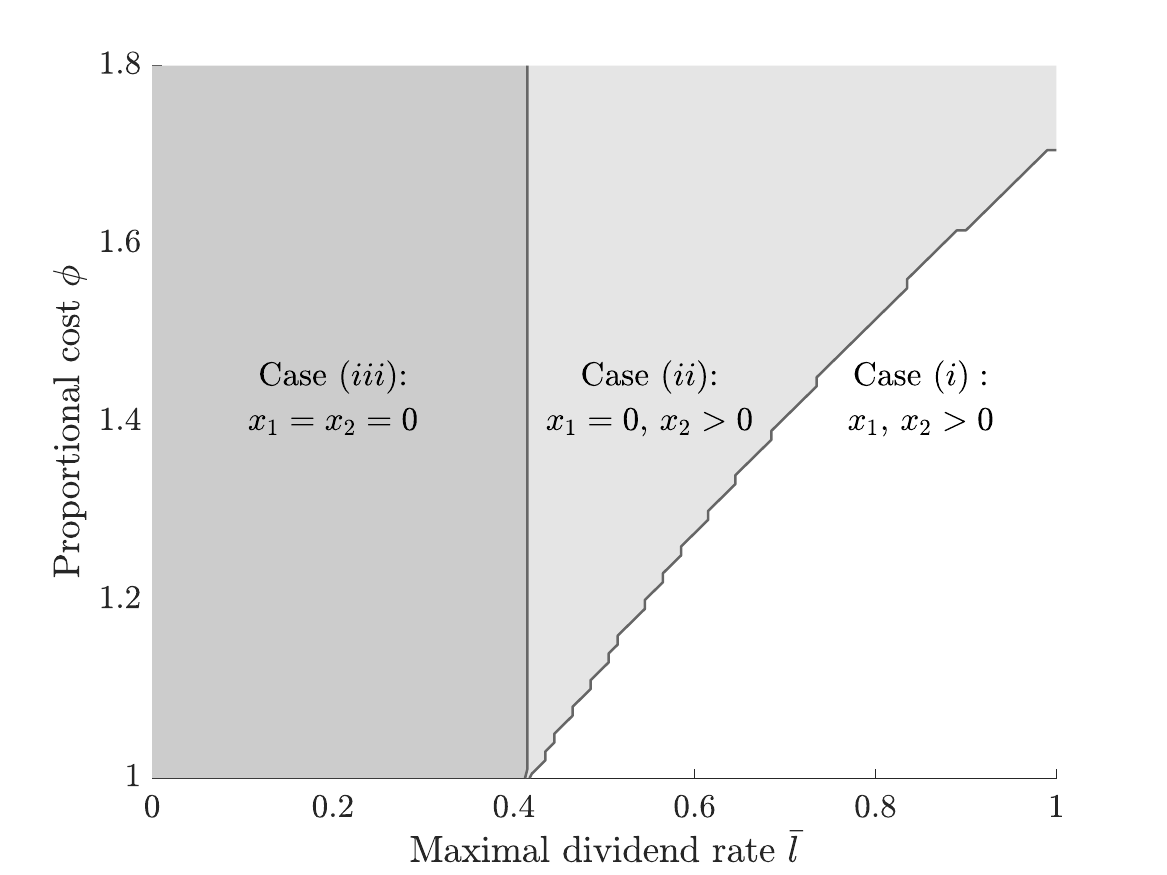}
\end{minipage}
\begin{minipage}{0.45\linewidth}
\includegraphics[width=\linewidth]{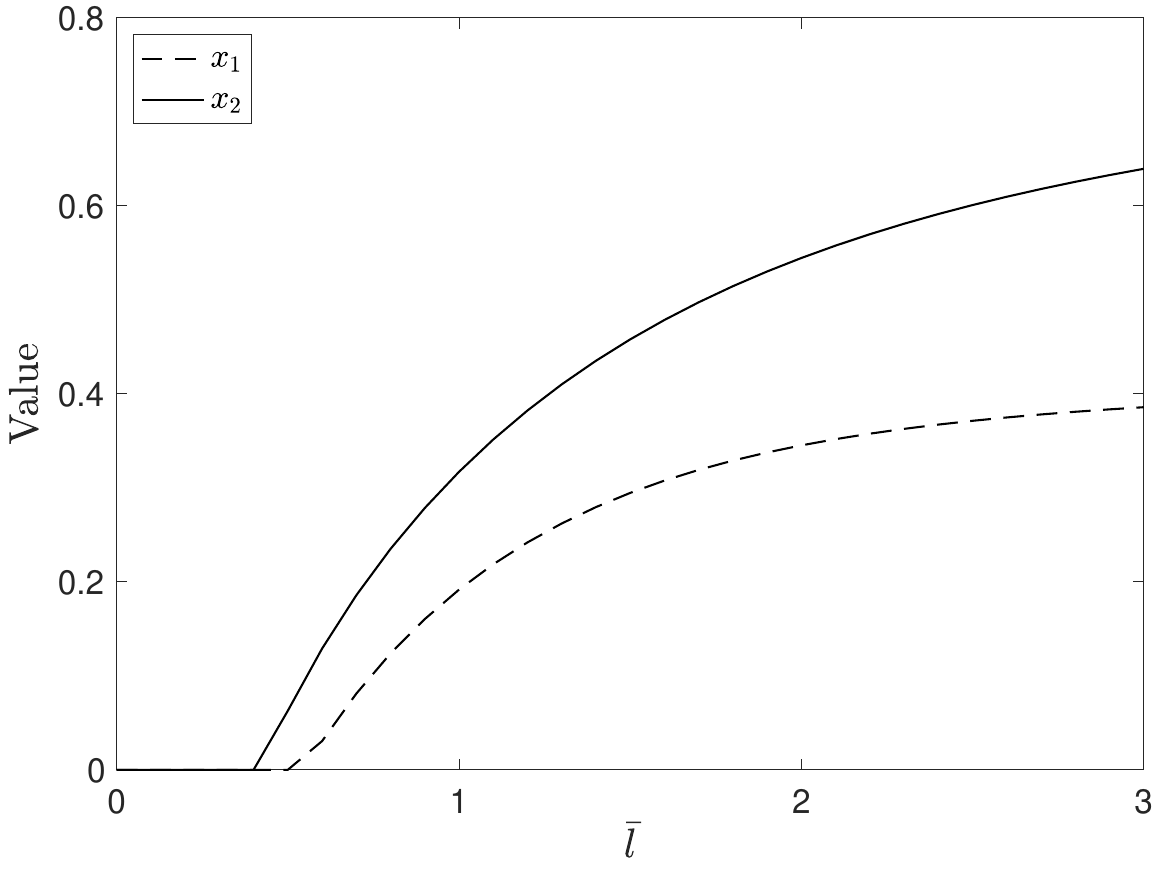}
\end{minipage}
\caption{The region of the three cases (left panel) and thresholds $x_1$ and $x_2$ (right panel) for $\mu = \sigma = 1$, $\omega = 0.3$, $\rho_1 = 0.6$, $\rho_2 = 1$, $\bar r = 1$, $\phi = 1.2$ (when varying $\bar l$ only) given different values of $\bar l$ and $\phi$.}
\label{fig:different_l}
\end{figure}

Figure \ref{fig:different_phi} shows the impact from the capital injection cost on the equilibrium value function $V$, capital injection threshold $x_1$, and dividend threshold $x_2$.
Let $\mu = \sigma = 1$, $\omega = 0.3$, $\rho_1 = 0.6$, $\rho_2 = 1$, $\bar l = \bar r = 2$.
The equilibrium value function $V$ is decreasing in the proportional cost $\phi$, as shown in the left panel in Figure \ref{fig:different_phi}.
The right panel shows that $x_1$ is decreasing in $\phi$, and $x_2$ is increasing in $\phi$.
Moreover, when $\phi = 1$, which means that there is no cost for capital injection, it can be inferred from \eqref{eq:equilibrium_strategy} that $x_1$ equal to $x_2$, and hence, the equilibrium strategy is either injecting capital or paying dividends based on whether the surplus is less or greater than the threshold 0.475.
For $\phi \ge 2.1577$, one obtains that $x_1 = 0$, meaning that the equilibrium strategy never has capita injection, and the problem degenerates to a pure dividend problem; e.g., when $\phi = 2.5$, the equilibrium value function shown in the left panel of Figure \ref{fig:different_phi} is equivalent to the one in \cite{zhao2014dividend}. For $\phi < 2.1577$, it is straightforward to see that the equilibrium value function with dividend and capital injection outperforms the one with dividend strategy only.

\begin{figure}[!htbp]
\centering
\begin{minipage}{0.45\linewidth}
\includegraphics[width=\linewidth]{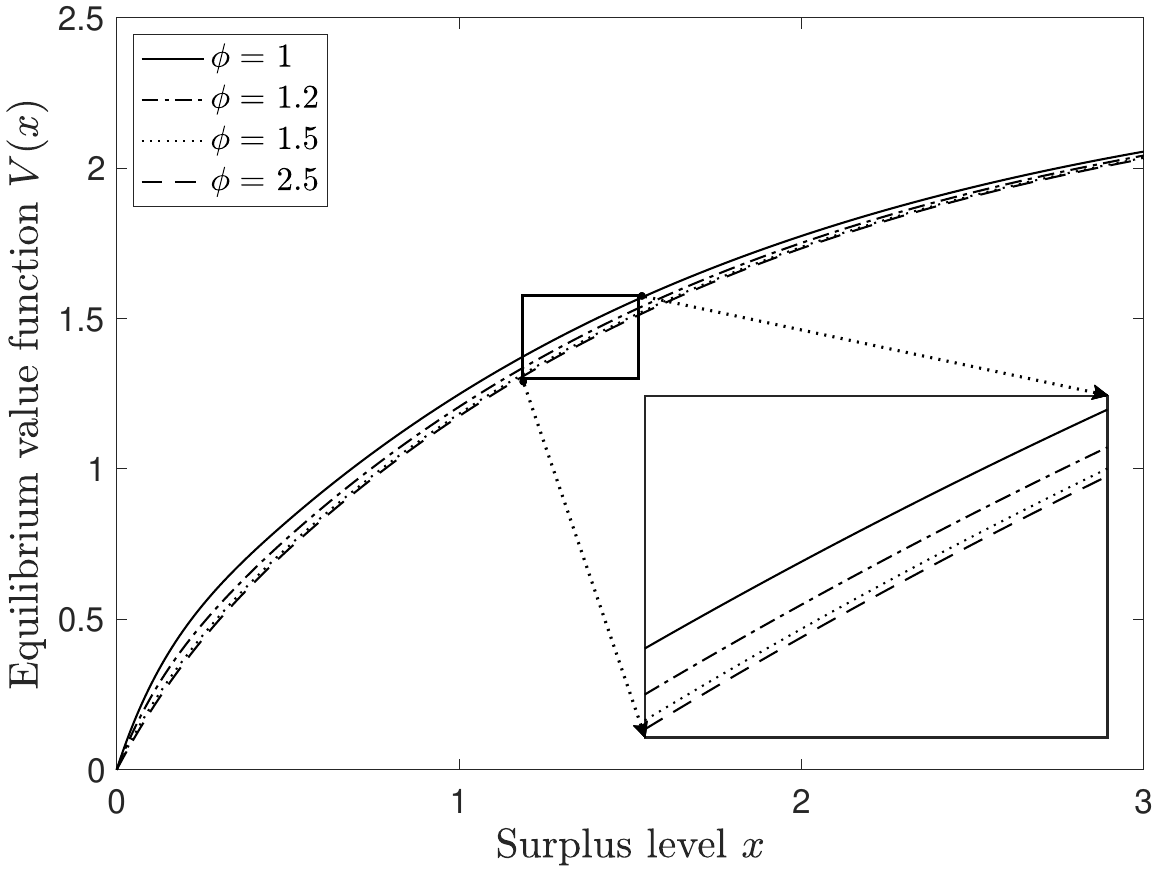}
\end{minipage}
\begin{minipage}{0.45\linewidth}
\includegraphics[width=\linewidth]{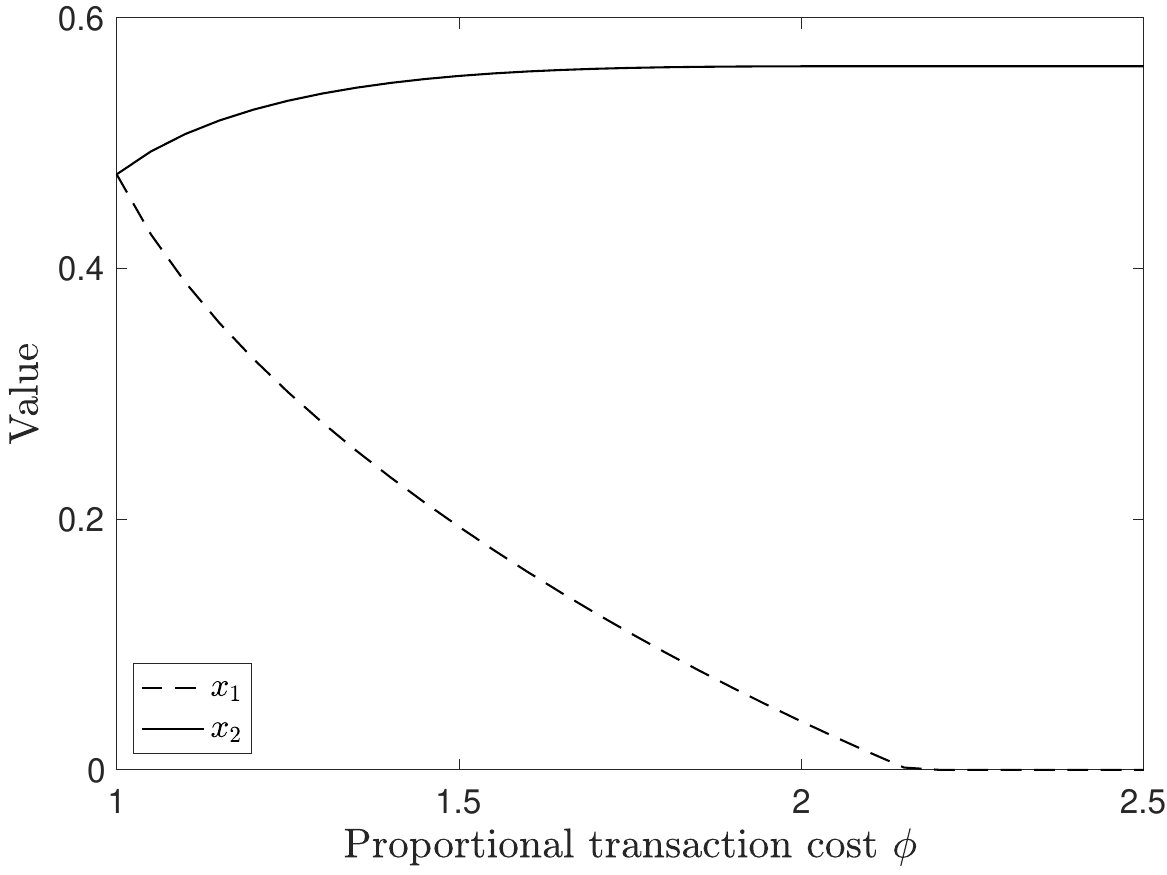}
\end{minipage}
\caption{The equilibrium value function $V$ (left panel) and thresholds $x_1$ and $x_2$ (right panel) for $\mu = \sigma = 1$, $\omega = 0.3$, $\rho_1 = 0.6$, $\rho_2 = 1$, $\bar l = \bar r = 2$ given different values of $\phi$.}
\label{fig:different_phi}
\end{figure}

Figure \ref{fig:different_musigma} shows the impact from the drift rate $\mu$ and volatility $\sigma$ on the equilibrium value function $V$, capital injection threshold $x_1$, and dividend threshold $x_2$.
Let $\mu = 1$ when varying $\sigma$ and $\sigma = 1$ when varying $\mu$, $\omega = 0.3$, $\rho_1 = 0.6$, $\rho_2 = 1$, $\bar l = \bar r = 1$, $\phi = 1.2$.
The equilibrium value function $V$ is increasing in  $\mu$ and decreasing in $\sigma$, as shown in the left panel in Figure \ref{fig:different_musigma}.
The thresholds $x_1$ and $x_2$ are first increasing then decreasing in both $\mu$ and $\sigma$, as shown in the right panel in Figure \ref{fig:different_musigma}.

\begin{figure}[!htbp]
\centering
\begin{minipage}{0.45\linewidth}
\includegraphics[width=\linewidth]{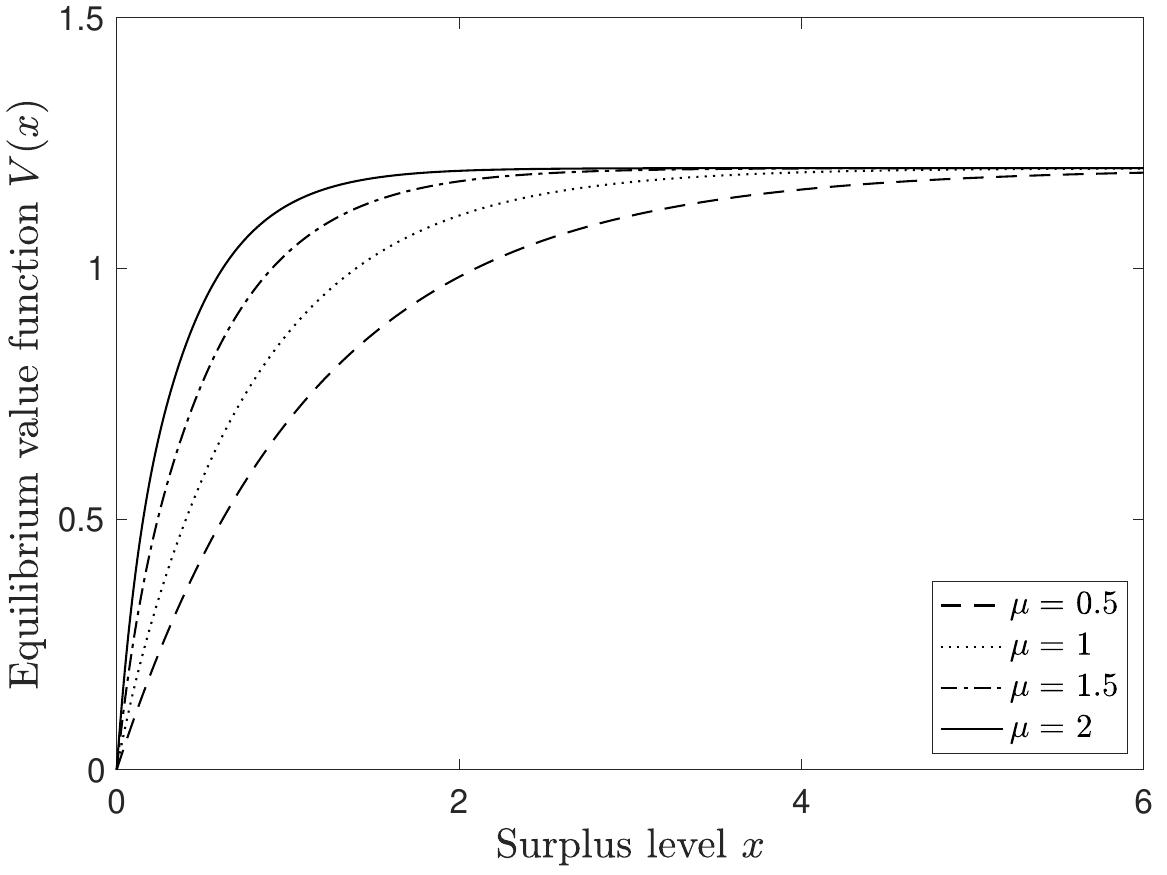}
\end{minipage}
\begin{minipage}{0.45\linewidth}
\includegraphics[width=\linewidth]{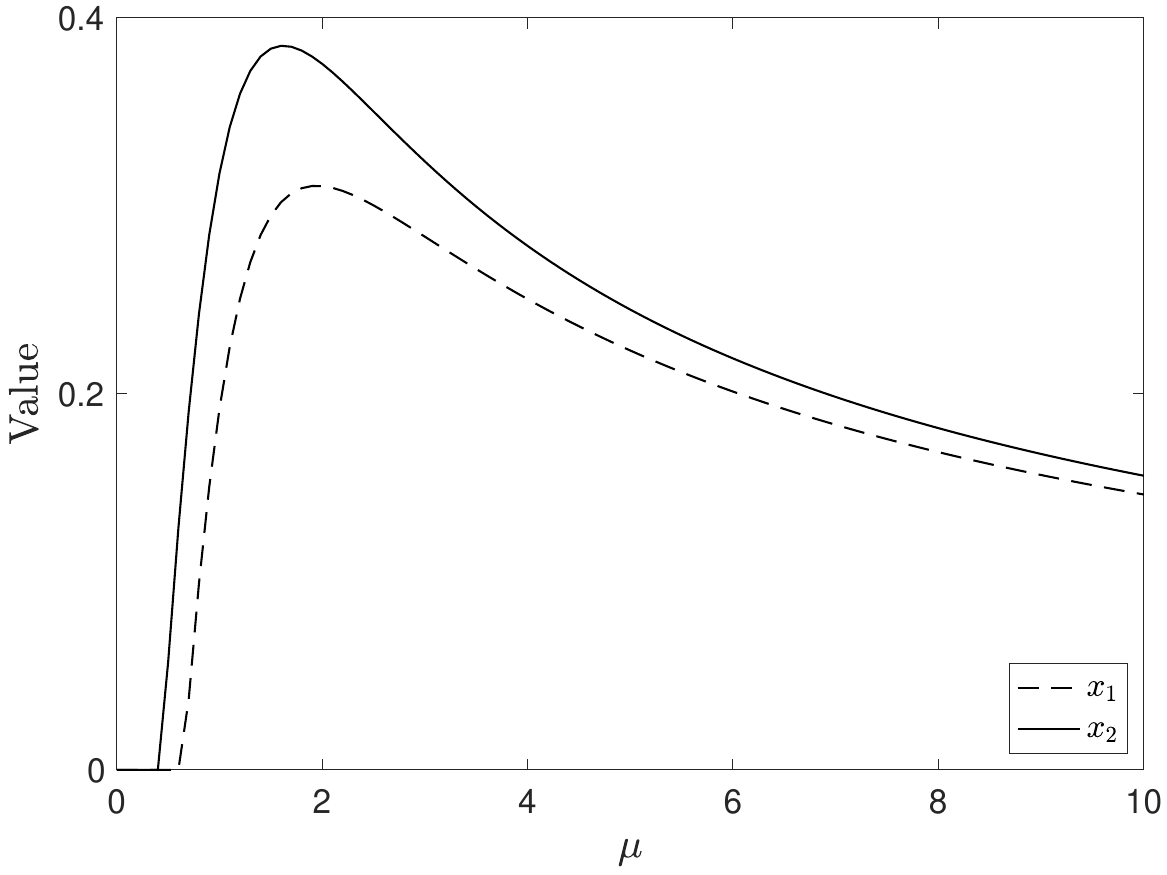}
\end{minipage}
\begin{minipage}{0.45\linewidth}
\includegraphics[width=\linewidth]{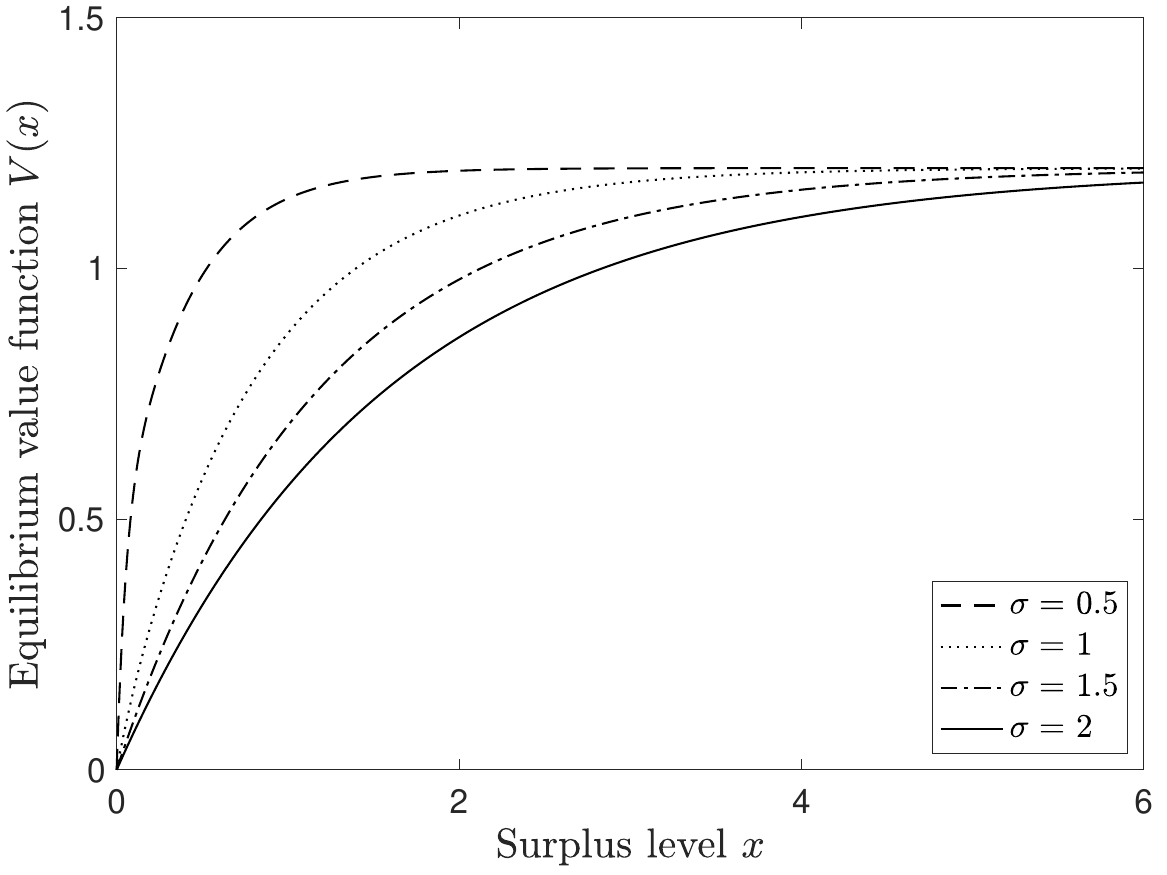}
\end{minipage}
\begin{minipage}{0.45\linewidth}
\includegraphics[width=\linewidth]{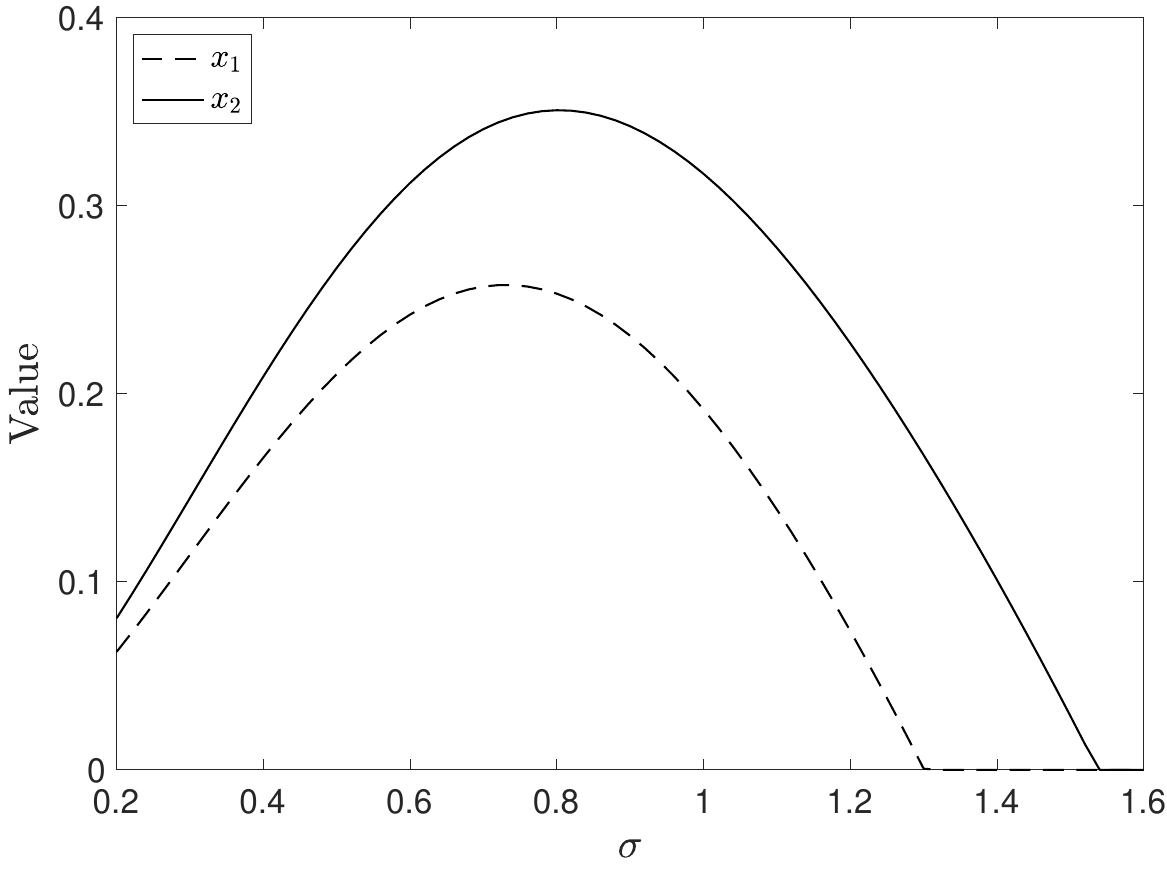}
\end{minipage}
\caption{The equilibrium value function $V$ (left panel) and thresholds $x_1$ and $x_2$ (right panel) for $\mu = 1$ (when varying $\sigma$), $\sigma = 1$ (when varying $\mu$), $\omega = 0.3$, $\rho_1 = 0.6$, $\rho_2 = 1$, $\bar l = \bar r = 2$,  $\phi = 1.2$ given different values of $\mu$ (top) and $\sigma$ (bottom).}
\label{fig:different_musigma}
\end{figure}

Figure \ref{fig:different_w} shows the impact from $\omega$ on the equilibrium value function $V$ (left panel) and thresholds $x_1$ and $x_2$ (right panel).
Let $\mu = \sigma = 1$, $\rho_1 = 0.6$, $\rho_2 = 1$, $\bar l = \bar r = 1$, $\phi = 1.2$.
The equilibrium value $V$ increases with $\omega$, where $\omega = 0, 1$ degenerate to the time-consistent problem.
Moreover, the right panel of Figure \ref{fig:different_w} shows that both $x_1$ and $x_2$ increase with the weight $\omega$.
This is because when $\omega$ increases, more weight will be put on the lower discount rate $\rho_1$, and hence, it is better to raise the thresholds for capital injection and dividend in order to enjoy more potential dividends.

\begin{figure}[!htbp]
\centering
\begin{minipage}{0.45\linewidth}
\includegraphics[width=\linewidth]{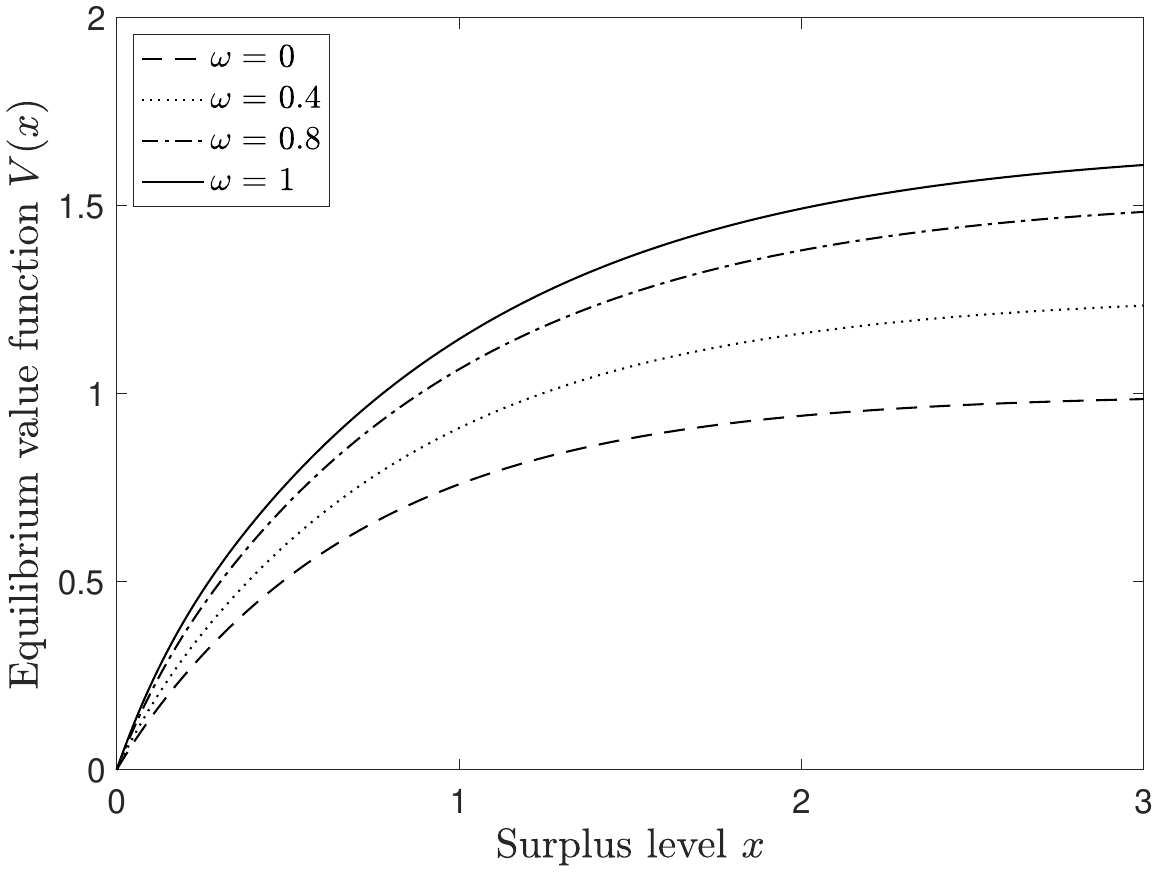}
\end{minipage}
\begin{minipage}{0.45\linewidth}
\includegraphics[width=\linewidth]{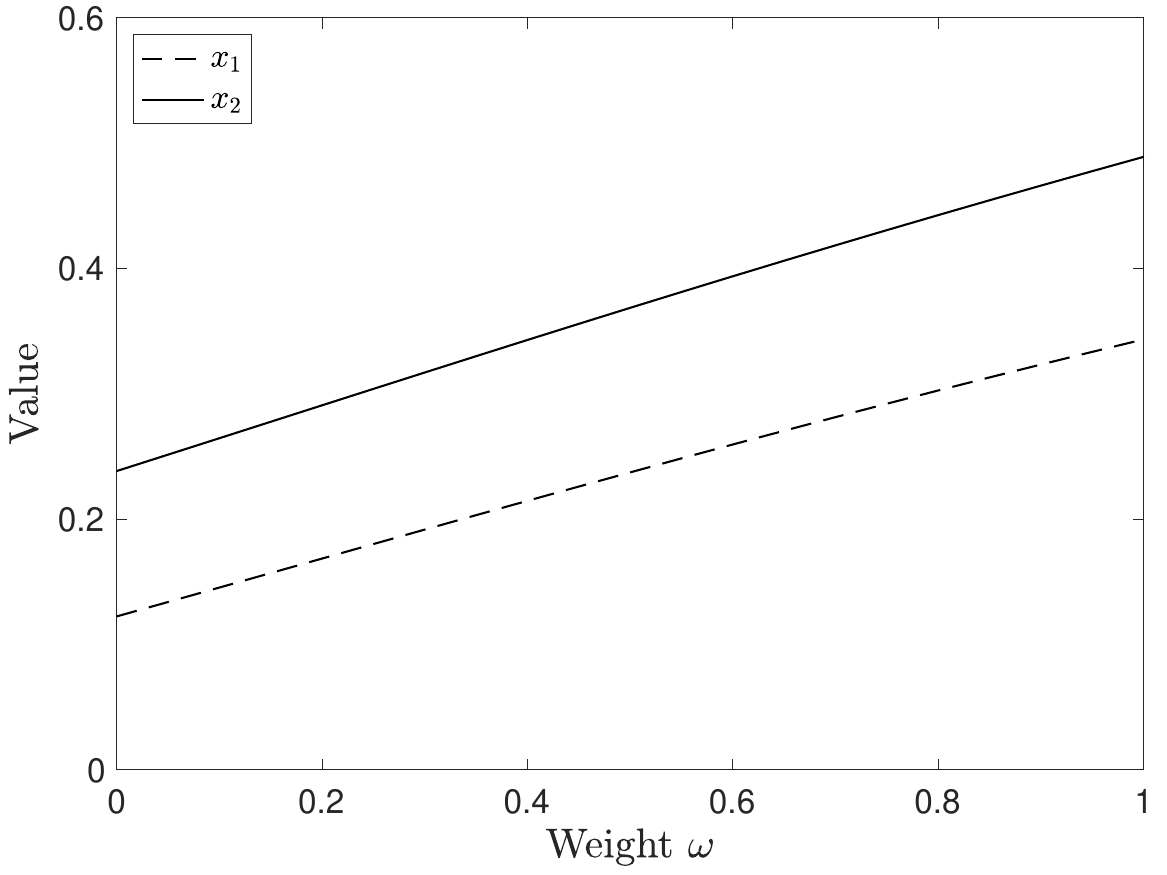}
\end{minipage}
\caption{The equilibrium value function $V$ (left panel) and thresholds $x_1$ and $x_2$ (right panel) for $\mu = \sigma = 1$, $\rho_1 = 0.6$, $\rho_2 = 1$, $\bar l = \bar r = 1$, $\phi = 1.2$ given different values of $\omega$.}
\label{fig:different_w}
\end{figure}

Finally, we vary the values of $\rho_1$ and $\rho_2$ while keeping the weighted average discount rate unchanged, i.e., $\omega \rho_1 + (1-\omega) \rho_2 = 0.8$.
Let $\mu = \sigma = 1$, $\omega = 0.5$, $\bar l = \bar r = 1$, $\phi = 1.2$. The left panel of Figure \ref{fig:different_rho} exhibits that the larger the difference between $\rho_1$ and $\rho_2$, the greater the equilibrium value function, with $\rho_1 = \rho_2 = 0.8$ being the time-consistent counterpart.
This implies that a larger difference in the shareholders' discount rates can increase the equilibrium value.
The right panel of Figure \ref{fig:different_rho} shows that the increase in the difference between $\rho_1$ and $\rho_2$ will raise the thresholds $x_1$ and $x_2$. The reason is that the larger the difference between $\rho_1$ and $\rho_2$, the less the discount is applied to future cashflows, and hence, the thresholds for both capital injections and dividends is raised to reduce the ruin probability to seek for more potential dividends in the future.

\begin{figure}[!htbp]
\centering
\begin{minipage}{0.45\linewidth}
\includegraphics[width=\linewidth]{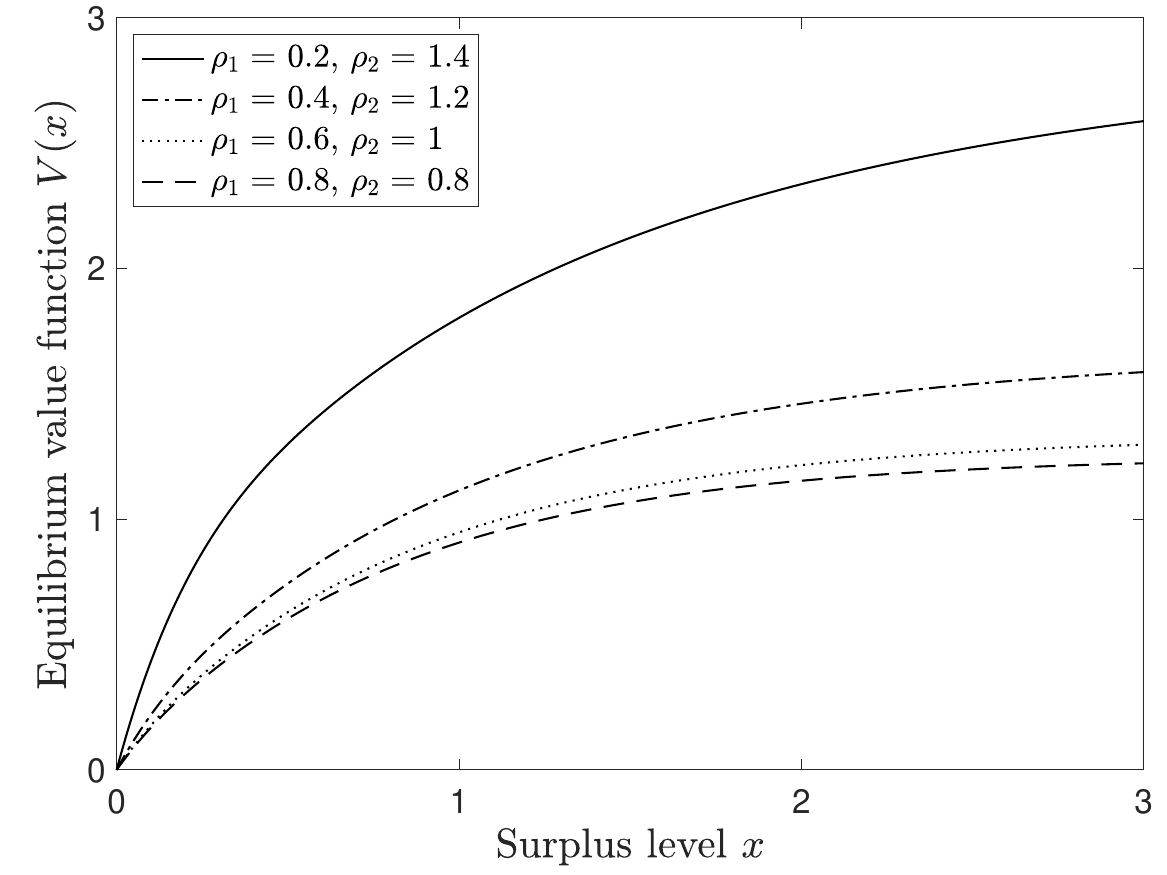}
\end{minipage}
\begin{minipage}{0.45\linewidth}
\includegraphics[width=\linewidth]{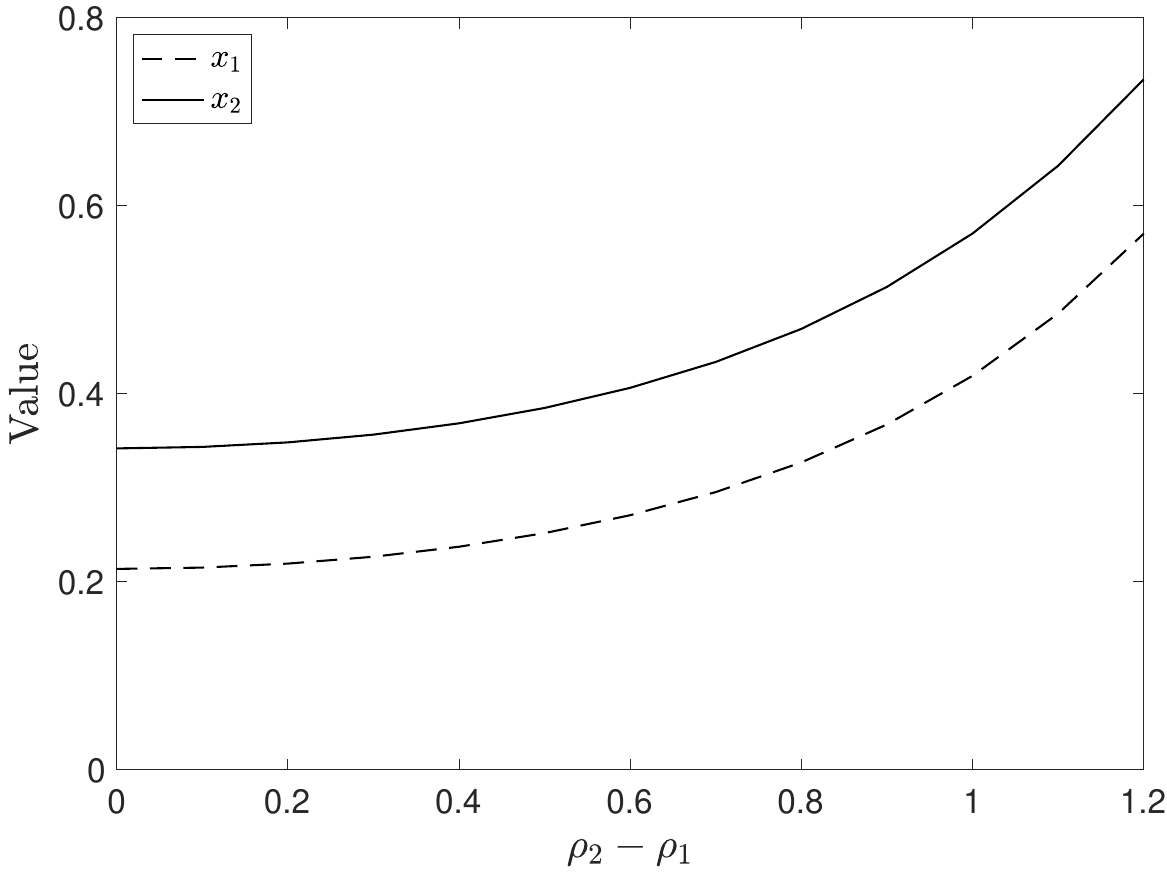}
\end{minipage}
\caption{The equilibrium value function $V$ (left panel) and thresholds $x_1$ and $x_2$ (right panel) for $\mu = \sigma = 1$, $\omega = 0.5$, $\bar l = \bar r = 1$, $\phi = 1.2$ given different combinations of $\rho_1$ and $\rho_2$ subjected to $\omega \rho_1 + (1-\omega) \rho_2 = 0.8$.}
\label{fig:different_rho}
\end{figure}

\section{Conclusion}\label{sec:conclusion}

This paper studies the dividend and capital injection problem in a diffusion risk model under time-inconsistent preferences. The shareholders of the insurance company aim to maximize the total expected discounted dividends with penalized capital injection deducted. For time-consistent benchmark where exponential discounting is used, we obtain the value function explicitly. In the general time-inconsistent framework, the definition of weak equilibrium in stochastic control is adopted to deal with the inconsistency. An extended HJB equation system is derived with verification theorem to prove the solution is indeed an equilibrium.
We then elaborate the results under the pseudo-exponential discount function, where the equilibrium strategy solved from the extended HJB is shown to be in threshold form. The equilibrium value function can be categorized into three cases according to the dividend paying threshold and capital injection threshold. If injection cost is too high or maximal dividend paying rate is too restrictive, the problem degenerates to the pure dividend part where no capital is injected.

\bibliography{references}

\appendix

\begin{appendices}

\setcounter{equation}{0}

\renewcommand\theequation{A.\arabic{equation}}

\section{Proof}\label{proof}

\pfof{Proposition \ref{prop:Vc_property}}
The boundedness of $V_c(x)$ is straightforward from the boundedness of $l_t$ and $r_t$, that is,
\begin{align*}
V_c(x) \leq \expect\left[\int_0^\infty e^{-\rho t} \bar l dt \right] = \frac{\bar l}{\rho}, \;
V_c(x) \geq \expect\left[\int_0^\infty e^{-\rho t} (-\phi \bar r) dt \right] = - \frac{\phi \bar r}{\rho}.
\end{align*}

Next, we show the concavity of $V_c(x)$. Let ${\bf u}^{(1)}$ and ${\bf u}^{(2)}$ be optimal strategies given initial surplus level $x_1$ and $x_2$ respectively. For $\lambda \in [0, 1]$, we define a strategy as $u^{\xi}_t := \lambda u^{(1)}_t + (1-\lambda)u^{(2)}_t$. Given initial surplus $X^{{\bf u}^{\xi}}_0 = \lambda x_1 + (1-\lambda) x_2$, according to \eqref{eq:surplus_process}, the controlled surplus process should evolve like
\begin{align*}
X_t^{{\bf u}^{\xi}} &= (\mu - \lambda l^{(1)}_t - (1-\lambda)l^{(2)}_t + \lambda r^{(1)}_t + (1-\lambda)r^{(2)}_t )dt+ \sigma dW_t\\
&=  \lambda \left[(\mu - l^{(1)}_t + r^{(1)}_t )dt+ \sigma dW_t\right] + (1-\lambda) \left[(\mu - l^{(2)}_t + r^{(2)}_t )dt+ \sigma dW_t\right]\\
&= \lambda X_t^{{\bf u}^{(1)}} + (1-\lambda) X_t^{{\bf u}^{(2)}}.
\end{align*}
Thus, $\tau_{0, \lambda x_1 + (1-\lambda) x_2}^{{\bf u}^{\xi}} = \tau_{0, x_1}^{{\bf u}^{(1)}} \vee \tau_{0, x_2}^{{\bf u}^{(2)}}$. Then the objective function corresponding to strategy ${\bf u}^{\xi}$ could be written as
\begin{align*}
& J_c(\lambda x_1 + (1-\lambda) x_2, {\bf u}^{\xi})\\
=& \expect\left[\int_0^{\tau_{0, x_1}^{{\bf u}^{(1)}} \vee \tau_{0, x_2}^{{\bf u}^{(2)}}} e^{-\rho t} \left(\lambda l^{(1)}_t + (1-\lambda)l^{(2)}_t - \phi \lambda r^{(1)}_t - \phi (1-\lambda)r^{(2)}_t \right) dt \right]\\
=&  \lambda \expect\left[\int_0^{\tau_{0, x_1}^{{\bf u}^{(1)}}} e^{-\rho t} \left(l^{(1)}_t - \phi r^{(1)}_t \right) dt \right] + (1-\lambda) \expect\left[\int_0^{\tau_{0, x_2}^{{\bf u}^{(2)}}} e^{-\rho t} \left(l^{(2)}_t - \phi r^{(2)}_t \right) dt \right]\\
=& \lambda V_c(x_1) + (1-\lambda) V_c(x_2).
\end{align*}
Therefore, we conclude that
\begin{align*}
V_c(\lambda x_1 + (1-\lambda) x_2) \geq J_c(\lambda x_1 + (1-\lambda) x_2, {\bf u}^{\xi}) = \lambda V_c(x_1) + (1-\lambda) V_c(x_2).
\end{align*}

Finally, to show the limit of $V_c$, we consider a constant strategy $\tilde u_t = (\tilde l_t, \tilde r_t) := (\bar l, 0)$, and the corresponding objective function is
\begin{align*}
J_c(x,{\bf \tilde u}) = \expect\left[\int_0^{\tau_{0, x}^{\bf \tilde u}} e^{-\rho t} \left(\tilde l_t - \phi \tilde r_t \right) dt \right] = \expect\left[\int_0^{\tau_{0, x}^{\bf \tilde u}} e^{-\rho t} \bar l dt \right].
\end{align*}
Therefore, by definition of value function:
\begin{align*}
V_c(x) \geq J_c(x,{\bf \tilde u}) = \expect\left[\int_0^{\tau_{0, x}^{\bf \tilde u}} e^{-\rho t} \bar l dt \right].
\end{align*}
Taking limit on both sides of the above inequality and by bounded convergence theorem,
$$
\lim_{x\to\infty} V_c(x) \geq \lim_{x\to\infty} \expect\left[\int_0^{\tau_{0, x}^{\bf \tilde u}} e^{-\rho t} \bar l dt \right] = \int_0^{\infty} e^{-\rho t} \bar l dt = \frac{\bar l}{\rho}.
$$
Combined with the fact we proved above, which states that $V_c$ is upper bounded by $\frac{\bar l}{\rho}$, we then reach the conclusion that $\lim_{x\to\infty} V_c(x) = \frac{\bar l}{\rho}$.
\qed

\pfof{Theorem \ref{thm:verification_exp}}
We omit the proof since it is standard, see, e.g., proof of \citet[Proposition 2.1]{jgaard1999controlling}.
\qed

\begin{lemma}\label{lemma:thetas_equivalent}
Given $\theta_3$, $\theta_4$, $\theta_5$ defined in \eqref{eq:thetas}, the following statements are equivalent:
$$
(i)\ \frac{\bar l}{\rho} + \frac{1}{\theta_5} > 0;\quad (ii)\ \theta_5 > \theta_3 + \theta_4; \quad (iii)\ \bar l > \frac{1}{2} \frac{\rho \sigma^2}{\mu}.
$$
\end{lemma}

\begin{proof}
First, we show that $(i)$ is equivalent to $(iii)$. Note that $\frac{\bar l}{\rho} + \frac{1}{\theta_5} > 0$ indicates $\theta_5 < -\rho/\bar l$, where $\theta_5$ is given in \eqref{eq:thetas}. By the definition of $\theta_5$, the inequality is equivalent to $-\mu + \bar l + \rho \sigma^2 /\bar l < \sqrt{(\mu-\bar l)^2 + 2 \rho \sigma^2}$. Taking square on both sides of the inequality, we obtain that $2\mu \bar l > \rho \sigma^2$.

If $(ii)$ holds, then $\theta_3+\theta_4 = -2\mu/\sigma^2 < \theta_5$. By the definition of $\theta_5$, $\mu+\bar l > \sqrt{(\mu-\bar l)^2 + 2 \rho \sigma^2}$. Taking square on both sides leads to $(iii)$, which completes our proof.
\end{proof}

\begin{lemma}\label{lemma:thetas_imply}
Given $\theta_i$ defined in \eqref{eq:thetas}, if $\ln\phi < \frac{1}{\theta4-\theta3} \left( \theta_4 \ln\frac{\theta_5-\theta_4}{\theta_3} + \theta_3 \ln\frac{-\theta_4}{\theta_3-\theta_5} \right)$, then $\frac{\bar l}{\rho} + \frac{1}{\theta_5} > 0$.
\end{lemma}

\begin{proof}
Since $\phi$ is defined to be greater than $1$, we have $\frac{1}{\theta4-\theta3} \left(\theta_4 \ln\frac{\theta_5-\theta_4}{\theta_3} + \theta_3 \ln\frac{-\theta_4}{\theta_3-\theta_5} \right) > \ln\phi > 0$, that is,
\begin{equation}\label{ineq:thetas_imply}
\theta_4 \ln\frac{\theta_5-\theta_4}{\theta_3} + \theta_3 \ln\frac{-\theta_4}{\theta_3-\theta_5} < 0.
\end{equation}
Suppose $\theta_5 < \theta_3 + \theta_4$, then $0 < (\theta_5-\theta_4)(\theta_3-\theta_5) = \theta_5 (\theta_3 + \theta_4 - \theta_5) - \theta_3 \theta_4 < - \theta_3 \theta_4$. Dividing $\theta_3 (\theta_3 - \theta_5)$ on both sides of the above inequality leads to $0 < \frac{\theta_5-\theta_4}{\theta_3} < \frac{-\theta_4}{\theta_3 - \theta_5}$, and $\frac{-\theta_4}{\theta_3 - \theta_5} < 1$ is directly obtained from $\theta_5 < \theta_3 + \theta_4$, hence $\ln\frac{\theta_5-\theta_4}{\theta_3} < \ln\frac{-\theta_4}{\theta_3-\theta_5} < 0$. Combined with $\theta_4 < -\theta_3$, it can be derived that $\theta_4 \ln\frac{\theta_5-\theta_4}{\theta_3} > \theta_4 \ln\frac{-\theta_4}{\theta_3-\theta_5} > -\theta_3 \ln\frac{-\theta_4}{\theta_3-\theta_5}$, which is a contradiction with \eqref{ineq:thetas_imply}. Therefore, it requires $\theta_5 > \theta_3 + \theta_4$, and according to Lemma \ref{lemma:thetas_equivalent}, it is equivalent to have $\frac{\bar l}{\rho} + \frac{1}{\theta_5} > 0$. So the inequality $\ln\phi < \frac{1}{\theta4-\theta3} \left( \theta_4 \ln\frac{\theta_5-\theta_4}{\theta_3} + \theta_3 \ln\frac{-\theta_4}{\theta_3-\theta_5} \right)$ implies that $\frac{\bar l}{\rho} + \frac{1}{\theta_5} > 0$.
\end{proof}

Define an auxiliary function
\begin{equation}\label{eq:f_c}
f_c(x) := \frac{\theta_3 \theta_4 \phi}{\theta_3 - \theta_5} \left(\frac{e^{\theta_1 x} - e^{\theta_2 x} + \frac{\bar r}{\rho} (\theta_1-\theta_2) e^{(\theta_1+\theta_2) x}}{\theta_1 e^{\theta_1 x} - \theta_2 e^{\theta_2 x}} -\frac{\bar r}{\rho} - \frac{1}{\theta_3}\right), \qquad x\in[0, \infty),
\end{equation}
where $\theta_i$ is given in \eqref{eq:thetas}.

\begin{lemma}\label{lemma:f_c}
Given function $f_c$ defined in \eqref{eq:f_c}, it satisfies the following properties:
\begin{enumerate}
\item[(i)] $f_c(0) = \frac{-\theta_4}{\theta_3 - \theta_5} \phi > \lim_{x\to\infty} f_c(x) = \frac{\theta_3 - \theta_1}{\theta_3 - \theta_5} \phi$;
\item[(ii)] there exists an unique $x^* \in (0, \infty)$ such that $f_c(x)$ decreases on $x \in [0, x^*]$ and then increases on $x\in[x^*, \infty)$.
\end{enumerate}
\end{lemma}

\begin{proof}
We first show $(i)$. It is obvious that $f_c(0) = \frac{-\theta_4}{\theta_3 - \theta_5} \phi$, and
\begin{align*}
\lim_{x\to\infty} f_c(x) &= \lim_{x\to\infty} \frac{\theta_3 \theta_4 \phi}{\theta_3 - \theta_5} \left(\frac{e^{\theta_1 x} - e^{\theta_2 x} + \frac{\bar r}{\rho} (\theta_1-\theta_2) e^{(\theta_1+\theta_2) x}}{\theta_1 e^{\theta_1 x} - \theta_2 e^{\theta_2 x}} -\frac{\bar r}{\rho} - \frac{1}{\theta_3}\right)\\
&= \frac{\theta_3 \theta_4 \phi}{\theta_3 - \theta_5} \left(\frac{1}{\theta_1} - \frac{1}{\theta_3} - \frac{\bar r}{\rho}\right)\\
&= \frac{\phi}{\theta_3 - \theta_5} \left(\theta_2 - \theta_4 + \frac{2 \bar r}{\sigma^2}\right)\\
&= \frac{\phi}{\theta_3 - \theta_5} \left(\theta_2 - \theta_4 + \theta_3 + \theta_4 - \theta_1 - \theta_2\right) = \frac{\theta_3 - \theta_1}{\theta_3 - \theta_5} \phi > 0,
\end{align*}
where the third equality holds because $\theta_1\theta_2 = \theta_3\theta_4 = -\frac{2\rho}{\sigma^2}$, and the fourth due to $\theta_3 + \theta_4 - (\theta_1 + \theta_2) = \frac{2 \bar r}{\sigma^2}$. $f_c(0) > \lim_{x\to\infty} f_c(x)$ follows from $\theta_3 + \theta_4 - \theta_1 < 0$.

To show the monotonicity of $f_c$, taking the derivative leads to
\begin{align*}
f_c^\prime(x) &= \frac{\theta_3 \theta_4 \phi}{\theta_3 - \theta_5} \frac{\left(\theta_1 e^{\theta_1 x} - \theta_2 e^{\theta_2 x}\right)^2 + \frac{\bar r}{\rho} (\theta_1^2-\theta_2^2) e^{(\theta_1+\theta_2) x} \left(\theta_1 e^{\theta_1 x} - \theta_2 e^{\theta_2 x}\right)}{\left(\theta_1 e^{\theta_1 x} - \theta_2 e^{\theta_2 x}\right)^2}\\
& \quad - \frac{\theta_3 \theta_4 \phi}{\theta_3 - \theta_5} \frac{\left(e^{\theta_1 x} - e^{\theta_2 x}\right) \left(\theta_1^2 e^{\theta_1 x} - \theta_2^2 e^{\theta_2 x}\right) + \frac{\bar r}{\rho} (\theta_1-\theta_2) e^{(\theta_1+\theta_2) x} \left(\theta_1^2 e^{\theta_1 x} - \theta_2^2 e^{\theta_2 x}\right)}{\left(\theta_1 e^{\theta_1 x} - \theta_2 e^{\theta_2 x}\right)^2}\\
&= \frac{\theta_3 \theta_4 \phi}{\theta_3 - \theta_5} \frac{(\theta_1 - \theta_2)^2 e^{(\theta_1+\theta_2) x} + \frac{\bar r}{\rho} \theta_1\theta_2 (\theta_1 - \theta_2) \left(e^{\theta_1 x} - e^{\theta_2 x}\right) e^{(\theta_1+\theta_2) x}}{\left(\theta_1 e^{\theta_1 x} - \theta_2 e^{\theta_2 x}\right)^2}\\
&= \frac{\theta_3 \theta_4 \phi}{\theta_3 - \theta_5} \frac{(\theta_1 - \theta_2) e^{(\theta_1+\theta_2) x} \left(\theta_1 - \theta_2 + \frac{\bar r}{\rho} \theta_1\theta_2 \left(e^{\theta_1 x} - e^{\theta_2 x}\right) \right)}{(\theta_1 e^{\theta_1 x} - \theta_2 e^{\theta_2 x})^2}.
\end{align*}
It is obvious that $f_c^\prime(x) < 0$ for $x \in [0, x^*)$ and $f_c^\prime(x) > 0$ for $x \in (x^*, \infty)$, with $x^*$ satisfying $\theta_1 - \theta_2 + \frac{\bar r}{\rho} \theta_1\theta_2 \left(e^{\theta_1 x^*} - e^{\theta_2 x^*}\right) = 0$.
\end{proof}

Define
\begin{equation}\label{eq:g_c}
g_c(x) := (\theta_4 - \theta_5) \left(f_c(x)\right)^\frac{\theta_3}{\theta_4} - (\theta_3 - \theta_5) f_c(x) + (\theta_3 - \theta_4) \phi, \qquad x\in\mathcal{O},
\end{equation}
where $f_c$ is defined in \eqref{eq:f_c} and $\mathcal{O} := \{x\in[0, \infty): f_c(x) \geq 1\}$.

\begin{lemma}\label{lemma:g_c}
$g_c(x)$ is defined in \eqref{eq:g_c} with domain $\mathcal{O}$. The monotonicity of $g_c$ is classified into the following cases based on the properties of $f_c$ given in Lemma \ref{lemma:f_c}:
\begin{enumerate}
\item[(i)] $\mathcal{O} = \emptyset$;
\item[(ii)] $\mathcal{O} = [0, \tilde{x}]$ for some $\tilde{x}\in[0, x^*)$: $g_c(x)$ is monotonically increasing in $[0, \tilde{x}]$, with $g_c(\tilde{x}) > 0$;
\item[(iii)] $\mathcal{O} = [0, \underline{x}] \cup [\bar x, \infty)$ for some $(\underline{x}, \bar x) \in [0,  x^*) \times (x^*, \infty)$:  $g_c(x)$ is increasing in $ [0, \underline{x}]$ and decreasing in $[\bar x, \infty)$, with $g_c(\underline{x}) = g_c(\bar{x}) > 0$ and $\lim_{x\to\infty} g_c(x) > 0$;
\item[(iv)] $\mathcal{O} = [0, \infty)$: $g_c(x)$ is increasing in $ [0, x^*]$ and decreasing in $[x^*, \infty)$, with $\lim_{x\to\infty} g_c(x) > 0$,
\end{enumerate}
where $x^*$ is solved from $f^\prime(x^*) = 0$.
\end{lemma}

\begin{proof}
First, if $f_c(x) = 1$, it is straightforward that $g_c(x) = (\theta_3 - \theta_4)(\phi - 1) > 0$ because of $\phi > 1$. Considering the monotonicity of $f_c(x)$ and the definition of $\mathcal{O}$, we must have $f(\tilde{x}) = f(\underline{x}) = f(\bar x) = 1$, thus $g_c(\tilde{x}) = g_c(\underline{x}) = g_c(\bar{x}) > 0$. To obtain the monotonicity of $g_c$, taking the first order derivative leads to
$$
g_c^\prime(x) = \left(\frac{\theta_3(\theta_4 - \theta_5)}{\theta_4} \left(f_c(x)\right)^{\frac{\theta_3}{\theta_4} -1} - (\theta_3 - \theta_5)\right) f^\prime_c(x), \qquad x \in  \mathcal{O}.
$$
Since $f_c(x) \geq 1$ for $x \in  \mathcal{O}$ and $\frac{\theta_3}{\theta_4} - 1 < 0$, we find that $\left(f_c(x)\right)^{\frac{\theta_3}{\theta_4} -1} < 1$, and
$$
\frac{\theta_3(\theta_4 - \theta_5)}{\theta_4} \left(f_c(x)\right)^{\frac{\theta_3}{\theta_4} -1} - (\theta_3 - \theta_5) < \frac{\theta_3(\theta_4 - \theta_5) - \theta_4 (\theta_3 - \theta_5)}{\theta_4} = \frac{(\theta_4-\theta_3) \theta_5}{\theta_4} < 0.
$$
For case $(ii)$, we have $f^\prime(x) < 0$ for $x \in [0, \tilde{x})$, thus $g_c^\prime(x) > 0$ for $x \in  \mathcal{O}$, which indicates that $g_c$ is monotonically increasing in $x$. Similarly for case $(iii)$, since $f^\prime(x) < 0$ for $x \in [0, \underline{x})$ and $f^\prime(x) > 0$ for $x \in (\bar x, \infty)$, we obtain that $g_c^\prime(x) > 0$ for $x \in [0, \underline{x})$ and $g^\prime(x) < 0$ for $x \in (\bar x, \infty)$. For case $(iv)$, as Lemma \ref{lemma:f_c} shows that $f^\prime(x) < 0$ for $x \in [0, x^*)$ and $f^\prime(x) > 0$ for $x \in (x^*, \infty)$, which implies that $g^\prime(x) > 0$ for $x \in [0, x^*)$ and $g^\prime(x) < 0$ for $x \in (x^*, \infty)$.

Note that in both case $(iii)$ and $(iv)$, $\lim_{x\to\infty} f_c(x) = \frac{\theta_3 - \theta_1}{\theta_3 - \theta_5} \phi > 1$, which implies that
\begin{equation}\label{ineq:lnphi_lb}
\ln \phi > \ln \frac{\theta_3 - \theta_5}{\theta_3 - \theta_1}.
\end{equation}

Suppose that $\lim_{x\to\infty} g_c(x) = (\theta_4 - \theta_5) \left(\frac{\theta_3 - \theta_1}{\theta_3 - \theta_5} \phi\right)^\frac{\theta_3}{\theta_4} + (\theta_1 - \theta_4) \phi < 0$. Rearranging the inequality leads to
\begin{equation}\label{ineq:lnphi_ub}
\ln\phi< \frac{\theta_4}{\theta_4 - \theta_3}  \ln\frac{\theta_5 - \theta_4}{\theta_1 - \theta_4} + \frac{\theta_3}{\theta_4 - \theta_3} \ln\frac{\theta_3 - \theta_1}{\theta_3 - \theta_5}.
\end{equation}
Combining the upper bound \eqref{ineq:lnphi_ub} with the lower bound \eqref{ineq:lnphi_lb},
\begin{equation}\label{ineq:lnphi_diff}
\ln \frac{\theta_3 - \theta_5}{\theta_3 - \theta_1} < \frac{\theta_4}{\theta_4 - \theta_3}  \ln\frac{\theta_5 - \theta_4}{\theta_1 - \theta_4} + \frac{\theta_3}{\theta_4 - \theta_3} \ln\frac{\theta_3 - \theta_1}{\theta_3 - \theta_5}.
\end{equation}
However,
\begin{align}\label{eq:lnphi_diff}
&\ln \frac{\theta_3 - \theta_5}{\theta_3 - \theta_1} - \frac{\theta_4}{\theta_4 - \theta_3}  \ln\frac{\theta_5 - \theta_4}{\theta_1 - \theta_4} - \frac{\theta_3}{\theta_4 - \theta_3} \ln\frac{\theta_3 - \theta_1}{\theta_3 - \theta_5}\notag\\
=& \frac{\theta_4}{\theta_4 - \theta_3} \left(\ln \frac{\theta_3 - \theta_5}{\theta_3 - \theta_1} - \ln\frac{\theta_5 - \theta_4}{\theta_1 - \theta_4}\right)\notag\\
=& \frac{\theta_4}{\theta_4 - \theta_3} \ln \frac{(\theta_3 - \theta_5)(\theta_1 - \theta_4)}{(\theta_3 - \theta_1)(\theta_5 - \theta_4)},
\end{align}
where
\begin{align*}
\frac{(\theta_3 - \theta_5)(\theta_1 - \theta_4)}{(\theta_3 - \theta_1)(\theta_5 - \theta_4)} - 1&= \frac{(\theta_3 - \theta_5)(\theta_1 - \theta_4) - (\theta_3 - \theta_1)(\theta_5 - \theta_4)}{(\theta_3 - \theta_1)(\theta_5 - \theta_4)}\\
&= \frac{(\theta_1-\theta_5)(\theta_3 - \theta_4)}{(\theta_3 - \theta_1)(\theta_5 - \theta_4)} > 0.
\end{align*}
Substituting the above inequality into \eqref{eq:lnphi_diff}, we have
$$
\ln \frac{\theta_3 - \theta_5}{\theta_3 - \theta_1} - \frac{\theta_4}{\theta_4 - \theta_3}  \ln\frac{\theta_5 - \theta_4}{\theta_1 - \theta_4} - \frac{\theta_3}{\theta_4 - \theta_3} \ln\frac{\theta_3 - \theta_1}{\theta_3 - \theta_5} > 0,
$$
which contradicts with \eqref{ineq:lnphi_diff}. Thus we conclude that when $\lim_{x\to\infty} f_c(x) > 1$, $\lim_{x\to\infty} g_c(x) > 0$.
\end{proof}

\pfof{Proposition \ref{thm:xrxl}}
First, using expression $V_c(x_l) = \frac{\bar l}{\rho} + A_4 e^{\theta_5 x_l}$ and $V_c^\prime(x_l) = 1$, it can be obtained that $A_4 = \frac{1}{\theta_5} e^{-\theta_5 x_l}$, thus $V_c(x_l) = \frac{\bar l}{\rho} + \frac{1}{\theta_5}$ and $V^{\prime\prime}_c(x_l) = A_4 \theta_5^2 e^{\theta_5 x_l} = \theta_5$. If $V_c(x_l) = \frac{\bar l}{\rho} + \frac{1}{\theta_5} \leq 0$, it contradicts with Corollary \ref{corollary:Vc_bound} which shows that $V_c(x)$ is an increasing function with $V_c(0) = 0$, thus in this case, there doesn't exist $x\in(0, \infty)$ that satisfies $V_c^\prime(x) = 1$, nor $V_c^\prime(x) = \phi > 1$, which proves $(iii)$.

In the following, we shall assume $\frac{\bar l}{\rho} + \frac{1}{\theta_5} > 0$. Using the smooth fit principle and combining with the ODE system \eqref{eq:Vc_ODE}, we find that
\begin{equation}\label{eq:smooth1}
V^{\prime}_c(x_l-) = A_2 \theta_3 e^{\theta_3 x_l} + A_3 \theta_4 e^{\theta_4 x_l} = 1,
\end{equation}
and
\begin{align}\label{eq:smooth2}
V^{\prime\prime}_c(x_l-) = A_2 \theta_3^2 e^{\theta_3 x_l} + A_3 \theta_4^2 e^{\theta_4 x_l} &= - \frac{2 \mu}{\sigma^2} V^\prime_c(x_l-) + \frac{2 \rho}{\sigma^2} V_c(x_l-) \notag\\
&= - \frac{2 (\mu - \bar l)}{\sigma^2} V^\prime_c(x_l) + \frac{2 \rho}{\sigma^2} V_c(x_l) - \frac{2 \bar l}{\sigma^2}\notag\\
& = V^{\prime\prime}_c(x_l) = \theta_5.
\end{align}
Multiplying $\theta_4$ on both sides of equation \eqref{eq:smooth1} and subtracting \eqref{eq:smooth2} leads to
\begin{equation}\label{eq:A2_xl}
A_2 = \frac{\theta_4 - \theta_5}{(\theta_4 - \theta_3) \theta_3} e^{-\theta_3 x_l}.
\end{equation}
Similarly, we obtain
\begin{equation}\label{eq:A3_xl}
A_3 = \frac{\theta_3 - \theta_5}{(\theta_3 - \theta_4) \theta_4} e^{-\theta_4 x_l}.
\end{equation}
Then according to $V_c(x_l-) = A_2 e^{\theta_3 x_l} + A_3 e^{\theta_4 x_l} = V_c(x_l) = \frac{\bar l}{\rho} + \frac{1}{\theta_5}$ and combining with expression \eqref{eq:A2_xl} and \eqref{eq:A3_xl}, we have
\begin{equation}\label{eq:Vc_xl}
\frac{\theta_4 - \theta_5}{(\theta_4 - \theta_3) \theta_3} + \frac{\theta_3 - \theta_5}{(\theta_3 - \theta_4) \theta_4} = \frac{\theta_3 + \theta_4 - \theta_5}{\theta_3\theta_4} = \frac{\bar l}{\rho} + \frac{1}{\theta_5}.
\end{equation}

Next, suppose there exits $x_r \in (0, \infty)$ such that $V^{\prime}_c(x_r) = \phi$. We establish the relationship between point $x_r$ and $x_l$ using expression \eqref{eq:A2_xl}-\eqref{eq:A3_xl} and smooth fit condition at point $x_r$. Since $V^{\prime}_c(x_r) = A_2 \theta_3 e^{\theta_3 x_r} + A_3 \theta_4 e^{\theta_4 x_r} = \phi$, we have $A_2 e^{\theta_3 x_r} = \frac{1}{\theta_3} \left(\phi - A_3 \theta_4 e^{\theta_4 x_r}\right)$ and substituting it into $V_c(x_r)$ leads to
$$
V_c(x_r) = \frac{1}{\theta_3} \left(\phi - A_3 \theta_4 e^{\theta_4 x_r}\right) + A_3 e^{\theta_4 x_r} = \frac{\phi}{\theta_3} + \frac{\theta_3 - \theta_4}{\theta_3} A_3 e^{\theta_4 x_r} = \frac{\phi}{\theta_3} + \frac{\theta_3 - \theta_5}{\theta_3 \theta_4} e^{-\theta_4 (x_l-x_r)}.
$$
On the other hand, $V^{\prime}_c(x_r-) = A_1 \left(\theta_1 e^{\theta_1 x_r} - \theta_2 e^{\theta_2 x_r}\right) + \frac{\phi \bar r}{\rho}\theta_2 e^{\theta_2 x_r} = \phi$ implies that $A_1 = \frac{\rho - \bar r \theta_2 e^{\theta_2 x_r}}{\theta_1 e^{\theta_1 x_r} - \theta_2 e^{\theta_2 x_r}} \frac{\phi}{\rho}$. Thus,
\begin{align*}
V_c(x_r-) &= A_1 \left(e^{\theta_1 x_r} - e^{\theta_2 x_r}\right) + \frac{\phi \bar r}{\rho} \left(e^{\theta_2 x_r} - 1\right)\\
&=\frac{\phi - \frac{\phi \bar r}{\rho} \theta_2 e^{\theta_2 x_r}}{\theta_1 e^{\theta_1 x_r} - \theta_2 e^{\theta_2 x_r}} \left(e^{\theta_1 x_r} - e^{\theta_2 x_r}\right) + \frac{\phi \bar r}{\rho} \left(e^{\theta_2 x_r} - 1\right)\\
&= -\frac{\phi \bar r}{\rho} + \phi \frac{e^{\theta_1 x_r} - e^{\theta_2 x_r} + \frac{\bar r}{\rho}(\theta_1-\theta_2) e^{(\theta_1+\theta_2) x_r}}{\theta_1 e^{\theta_1 x_r} - \theta_2 e^{\theta_2 x_r}}.
\end{align*}
Then by smooth fit principle, i.e., $V_c(x_r) = V_c(x_r-)$, we establish that
\begin{equation}\label{eq:xl-xr}
e^{-\theta_4 (x_l-x_r)} = \frac{\theta_3 \theta_4 \phi}{\theta_3 - \theta_5} \left(\frac{e^{\theta_1 x_r} - e^{\theta_2 x_r} + \frac{\bar r}{\rho} (\theta_1-\theta_2) e^{(\theta_1+\theta_2) x_r}}{\theta_1 e^{\theta_1 x_r} - \theta_2 e^{\theta_2 x_r}} -\frac{\bar r}{\rho} - \frac{1}{\theta_3}\right) = f_c(x_r),
\end{equation}
where $f_c(x)$ is defined in \eqref{eq:f_c}. Hence $x_l = x_r + \frac{1}{-\theta_4} \ln f_c(x_r)$, and $e^{-\theta_3 (x_l-x_r)} = \left(f_c(x_r)\right)^\frac{\theta_3}{\theta_4}$.

Recall that
\begin{align*}
V^{\prime}_c(x_r) &= A_2 \theta_3 e^{\theta_3 x_r} + A_3 \theta_4 e^{\theta_4 x_r}\\
&= \frac{\theta_4 - \theta_5}{\theta_4 - \theta_3} e^{-\theta_3 (x_l - x_r)} + \frac{\theta_3 - \theta_5}{\theta_3 - \theta_4} e^{-\theta_4 (x_l - x_r)}\\
&= \frac{\theta_4 - \theta_5}{\theta_4 - \theta_3} \left(f_c(x_r)\right)^\frac{\theta_3}{\theta_4} + \frac{\theta_3 - \theta_5}{\theta_3 - \theta_4} f_c(x_r) = \phi,
\end{align*}
where we use the relationship between $x_l$ and $x_r$ in \eqref{eq:xl-xr} and expression \eqref{eq:A2_xl}-\eqref{eq:A3_xl}. We construct a nonlinear equation that $x_r$ should satisfy, that is, $g_c(x_r) = 0$, where $g_c$ is defined in \eqref{eq:g_c}. Once $x_r$ is obtained, $x_l$ can be solved from \eqref{eq:xl-xr}. Besides, as Proposition \ref{prop:Vc_property} shows that $V_c$ is concave, it is required that $x_l \geq x_r$, that is, $f_c(x_r) \geq 1$ according to \eqref{eq:xl-xr}, thus the domain of function $g_c(x)$ should be restricted to $\mathcal{O} := \{x\in[0, \infty): f_c(x) \geq 1\}$.

Thanks to  Lemma \ref{lemma:f_c} and Lemma \ref{lemma:g_c}, we can classify the solution to $g_c(x_r) = 0$ into the following four cases. First, $\mathcal{O} = \emptyset$ if $f_c(0) = \frac{-\theta_4}{\theta_3 - \theta_5} \phi < 1$. As required above, $\frac{\bar l}{\rho} + \frac{1}{\theta_5} > 0$, which is equivalent to $\theta_5 > \theta_3 + \theta_4$ according to Lemma \ref{lemma:thetas_equivalent}, thus we have $\frac{-\theta_4}{\theta_3 - \theta_5} \phi > 1$, which implies that $\mathcal{O} = \emptyset$ cannot happen. In the second case, $\mathcal{O} = [0, \tilde{x}]$ for some $\tilde{x}\in[0, x^*)$, where $x^*$ is solved from $f^\prime(x^*) = 0$, since $g_c(x)$ is monotonically increasing in $[0, \tilde{x}]$ with $g_c(\tilde{x}) > 0$, there exists an unique $x_r \in (0, \infty)$ such that $g(x_r) = 0$ if and only if $g_c(0) = (\theta_4 - \theta_5) \left(\frac{-\theta_4}{\theta_3 - \theta_5} \phi\right)^\frac{\theta_3}{\theta_4} + \theta_3 \phi < 0$, which is equivalent to $\ln\phi < \frac{1}{\theta4-\theta3} \left( \theta_4 \ln\frac{\theta_5-\theta_4}{\theta_3} + \theta_3 \ln\frac{-\theta_4}{\theta_3-\theta_5} \right)$, otherwise, if $\ln\phi \geq \frac{1}{\theta4-\theta3} \left( \theta_4 \ln\frac{\theta_5-\theta_4}{\theta_3} + \theta_3 \ln\frac{-\theta_4}{\theta_3-\theta_5} \right)$,
there is no such $x_r\in(0,\infty)$ that satisfies $g_c(x_r) = 0$. Similarly for Case $(iii)$ and $(iv)$ in Lemma \ref{lemma:g_c}, since $g_c(x) > 0$ for $x \in [\bar x, \infty)$ and $x \in [x^*, \infty)$ respectively, and $g_c(\underline{x})>0$, there exists a unique $x_r \in (0, \infty)$ if and only if $g_c(0) < 0$, otherwise, if $g_c(0) \geq 0$, such $x_r \in (0, \infty)$ doesn't exits.

Then it remains to check the existence of $x_l$. If a unique $x_r \in (0, \infty)$ exists, i.e., $g_c(0) < 0$, a unique $x_l \in[x_r, \infty)$ follows from \eqref{eq:xl-xr}. If there doesn't exists $x_r \in \mathcal{O}$ such that $g_c(x_r) = 0$, then let $x_r = 0$, and solving $x_l$ from $V_c^\prime(x_l) = A_4 \theta_5 e^{\theta_5 x_l}= 1$, that is,
$$
\left(1-\left(\frac{\bar l}{\rho} + \frac{1}{\theta_5}\right)\theta_3\right) e^{\theta_3 x_l} = \left(1-\left(\frac{\bar l}{\rho} + \frac{1}{\theta_5}\right)\theta_4\right) e^{\theta_4 x_l},
$$
since $\frac{\bar l}{\rho} + \frac{1}{\theta_5} = \frac{\theta_3 + \theta_4 - \theta_5}{\theta_3\theta_4}$ as derived in \eqref{eq:Vc_xl}, we can see that $1-\left(\frac{\bar l}{\rho} + \frac{1}{\theta_5}\right)\theta_3 = 1-\frac{\theta_3 + \theta_4 - \theta_5}{\theta_4} = \frac{\theta_5 - \theta_3}{\theta_4} > 0$, and obviously $1-\left(\frac{\bar l}{\rho} + \frac{1}{\theta_5}\right)\theta_4 = \frac{\theta_5 - \theta_4}{\theta_3} > 0$. Thus taking the logarithm on both sides of the above equation and rearranging it leads to
$$
x_l = \frac{1}{\theta_3 - \theta_4} \ln\frac{\theta_4(\theta_5-\theta_4)}{\theta_3(\theta_5-\theta_3)}.
$$
Thus, there always exists a unique $x_l \in (0, \infty)$ such that $V_c^\prime(x_l) = 1$ if $\frac{\bar l}{\rho} + \frac{1}{\theta_5} > 0$, which completes our proof.
\qed

\pfof{Theorem \ref{thm:verification_equilibrium}}
First, we show that $V(s, x) = J(s, x, {\bf \hat u})$, where ${\bf \hat u}$ solves the supremum. Applying It\^{o}'s Lemma on $g(t, X_{t}^{\bf \hat u}; k)$, one obtains that for any $T > s$,
\begin{align*}
g(s, X_s^{\bf \hat u}; k) &= g(T \wedge \tau_{s, x}^{\bf \hat u}, X_{T \wedge \tau_{s, x}^{\bf \hat u}}^{\bf \hat u}; k)\\
& \quad -  \int_s^{T \wedge \tau_{s, x}^{\bf \hat u}} \left(g_s(t, X_t^{\bf \hat u}; k) + \left(\mu - \hat l_t + \hat r_t \right) g_x(t, X_t^{\bf \hat u}; k) + \frac{1}{2}\sigma^2 g_{xx}(t, X_t^{\bf \hat u}; k) \right) dt\\
& \quad - \int_s^{T \wedge \tau_{s, x}^{\bf \hat u}} \sigma g_x(t, X_t^{\bf \hat u}; k) dW_t\\
&= g(T \wedge \tau_{s, x}^{\bf \hat u}, X_{T \wedge \tau_{s, x}^{\bf \hat u}}^{\bf \hat u}; k) +  \int_s^{T \wedge \tau_{s, x}^{\bf \hat u}} \delta(k,t) \left( \hat l_t - \phi \hat r_t \right) dt \\
& \quad - \int_s^{T \wedge \tau_{s, x}^{\bf \hat u}} \sigma g_x(t, X_t^{\bf \hat u}; k) dW_t,
\end{align*}
where the last equality is due to \eqref{extended_HJB:2}. Taking expectation conditional on $X_s^{\bf \hat u} = x$ leads to
\begin{align}\label{eq:g_expect}
g(s,x; k) &= \mathbb{E}_{s, x}\left[ g(T \wedge \tau_{s, x}^{\bf \hat u}, X_{T \wedge \tau_{s, x}^{\bf \hat u}}^{\bf \hat u}; k) \right] + \mathbb{E}_{s, x}\left[ \int_s^{T \wedge \tau_{s, x}^{\bf \hat u}} \delta(k,t) \left( \hat l_t - \phi \hat r_t \right) dt \right]\notag\\
&\quad - \mathbb{E}_{s, x}\left[ \int_s^{T \wedge \tau_{s, x}^{\bf \hat u}} \sigma g_x(t, X_t^{\bf \hat u}; k) dW_t\right].
\end{align}
For the first term on the right-hand side of \eqref{eq:g_expect}, since $\left\vert g \right\vert$ is bounded, then by bounded convergence theorem and the boundary condition \eqref{extended_HJB:bound2},
\begin{align*}
\lim_{T\to\infty} \mathbb{E}_{s, x}\left[  g(T \wedge \tau_{s, x}^{\bf \hat u}, X_{T \wedge \tau_{s, x}^{\bf \hat u}}^{\bf \hat u}; k) \right] = \mathbb{E}_{s, x}\left[ g(\tau_{s, x}^{\bf \hat u}, 0; k) \right] = 0.
\end{align*}
For the second term on the right hand side of \eqref{eq:g_expect}, since ${\bf \hat u}$ is admissible and satisfies Definition \ref{def:admissible_control}, by monotone convergence theorem,
\begin{align*}
 \lim_{T\to\infty} \mathbb{E}_{s, x}\left[ \int_s^{T \wedge \tau_{s, x}^{\bf \hat u}} \delta(k,t) \left(\hat l_t - \phi \hat r_t \right) dt \right] =  \mathbb{E}_{s, x}\left[ \int_s^{\tau_{s, x}^{\bf \hat u}} \delta(k,t) \left(\hat l_t - \phi \hat r_t \right) dt \right].
\end{align*}
For the third term on the right hand side of \eqref{eq:g_expect}, since $\left\vert g_x \right\vert$ is bounded, the stochastic integral $\left\{\int_s^{\tilde s} \sigma g_x(t, X_t^{\bf \hat u}; k) dW_t \right\}_{{\tilde s}\ge s}$ is a martingale, then by optional sampling theorem,
\begin{align*}
\mathbb{E}_{s, x}\left[ \int_s^{T \wedge \tau_{s, x}^{\bf \hat u}} \sigma g_x(t, X_t^{\bf \hat u}; k) dW_t \right] = 0.
\end{align*}
Thus, letting $T \to \infty$ on both sides of \eqref{eq:g_expect} leads to
\begin{equation}
\label{eq:g}
g(s, x; k)  =  \mathbb{E}_{s, x}\left[ \int_s^{\tau_{s, x}^{\bf \hat u}} \delta(k,t) \left(\hat l_t - \phi \hat r_t \right) dt \right].
\end{equation}
Then by relation \eqref{extended_HJB:3},
\begin{equation}\label{eq:f}
f(s, x) = g(s, x; s) = \mathbb{E}_{s, x} \left[\int_s^{\tau_{s, x}^{\bf \hat u}} \delta(s,t) \left(\hat l_t - \phi \hat r_t \right) dt\right].
\end{equation}
Because $\hat u_s = (\hat l_s, \hat r_s)$ reaches the supremum in \eqref{extended_HJB:1}, combining with \eqref{extended_HJB:2}, we obtain
\begin{align}\label{eq:V_f}
&V_s(s,x) + (\mu - \hat l_s + \hat r_s) V_x(s,x) + \frac{1}{2} \sigma^2 V_{xx}(s,x)\notag\\
=& f_s(s,x) + (\mu - \hat l_s + \hat r_s) f_x(s,x) + \frac{1}{2} \sigma^2 f_{xx}(s,x).
\end{align}
Applying It\^{o}'s Lemma on $f(t, X_{t}^{\bf \hat u})$ and $V(t, X_{t}^{\bf \hat u})$, respectively, and taking conditional expectations on both sides,
\begin{align*}
f(s, x) = - \mathbb{E}_{s, x}\left[\int_s^{\tau_{s, x}^{\bf \hat u}} \left( f_s(t,X_t^{\bf \hat u}) + (\mu - \hat l_t + \hat r_t) f_x(t,X_t^{\bf \hat u}) + \frac{1}{2} \sigma^2 f_{xx}(t,X_t^{\bf \hat u}) \right) dt \right].
\end{align*}
By \eqref{eq:V_f} and \eqref{eq:f},
\begin{align}
V(s, x) & = - \mathbb{E}_{s, x}\left[\int_s^{\tau_{s, x}^{\bf \hat u}} \left( V_s(t,X_t^{\bf \hat u}) + (\mu - \hat l_t + \hat r_t) V_x(t,X_t^{\bf \hat u}) + \frac{1}{2} \sigma^2 V_{xx}(t,X_t^{\bf \hat u}) \right) dt \right] \notag\\
& = f(s, x). \label{eq:Vf}
\end{align}
It is also shown that
\begin{align*}
 f(s, x)  = \mathbb{E}_{s, x} \left[ \int_s^{\tau_{s, x}^{\bf \hat u}} \delta(s,t) \left(\hat l_t - \phi \hat r_t \right) dt \right] = J(s, x, {\bf \hat u}).
\end{align*}

Next, we show that ${\bf \hat u}$ is an equilibrium strategy. The $\epsilon$-perturbed strategy ${\bf u}^{\epsilon,l,r}$ is defined as in Definition \ref{def:eq}, where $(l, r)\in[0,\bar l]\times[0,\bar r]$ are arbitrary constants. In the following, we write ${\bf u}^{\epsilon}$ for simplicity. Note that
\begin{align*}
&J(s, x, {\bf u^\epsilon})\\
=& \mathbb{E}_{s, x} \left[\int_{s}^{(s+\epsilon) \wedge \tau_{s, x}^{\bf u^\epsilon}} \delta(s,t) \left(l - \phi r \right) dt \right] + \mathbb{E}_{s, x} \left[\int_{(s+\epsilon) \wedge \tau_{s, x}^{\bf u^\epsilon}}^{\tau_{s, x}^{\bf u^\epsilon}} \delta(s,t) \left(\hat l_t - \phi \hat r_t \right) dt \right] \\
=& \mathbb{E}_{s, x} \left[\int_{s}^{(s+\epsilon) \wedge \tau_{s, x}^{\bf u^\epsilon}} \delta(s,t) dt \right] \left(l - \phi r \right) \\
\quad & + \mathbb{E}_{s, x} \left[\mathbb{E}_{ (s+\epsilon)\wedge \tau_{s, x}^{\bf u^\epsilon}, X_{ (s+\epsilon) \wedge \tau_{s, x}^{\bf u^\epsilon}}^{\bf u^\epsilon}} \left[\int_{(s+\epsilon) \wedge \tau_{s, x}^{\bf u^\epsilon}}^{\tau_{(s+\epsilon)\wedge \tau_{s,x}^{\bf u^\epsilon},  X_{ (s+\epsilon) \wedge \tau_{s, x}^{\bf u^\epsilon}}^{\bf u^\epsilon} }^{\bf \hat u}} \delta(s,t) \left(\hat l_t - \phi \hat r_t \right) dt \right]\right]\\
=& \mathbb{E}_{s, x} \left[\int_{s}^{(s+\epsilon) \wedge \tau_{s, x}^{\bf u^\epsilon}} \delta(s,t) dt \right] \left(l - \phi r \right)  + \mathbb{E}_{s, x} \left[g\left( (s+\epsilon) \wedge\tau_{s, x}^{\bf u^\epsilon} , X_{(s+\epsilon) \wedge \tau_{s, x}^{\bf u^\epsilon}}^{\bf u^\epsilon}; s \right) \right],
\end{align*}
where the last equality comes from \eqref{eq:g}.
Applying It\^{o}'s Lemma on $V(t, X_{t}^{\bf u^{\epsilon}})$ and taking conditional expectation on both sides,
\begin{align*}
V(s, x) &= \mathbb{E}_{s, x}\left[V( (s+\epsilon) \wedge \tau_{s, x}^{\bf u^\epsilon}, X_{ (s+\epsilon) \wedge \tau_{s, x}^{\bf u^\epsilon}}^{\bf u^\epsilon})\right] \\
 & \quad - \mathbb{E}_{s, x}\left[\int_s^{(s+\epsilon) \wedge \tau_{s, x}^{\bf u^\epsilon}} \left( V_s(t, X_{t}^{\bf u^{\epsilon}}) + (\mu - l + r) V_x(t, X_{t}^{\bf u^{\epsilon}}) + \frac{1}{2} \sigma^2 V_{xx}(t, X_{t}^{\bf u^{\epsilon}}) \right) dt \right] \\
&= \mathbb{E}_{s, x}\left[f( (s+\epsilon) \wedge \tau_{s, x}^{\bf u^\epsilon}, X_{ (s+\epsilon) \wedge \tau_{s, x}^{\bf u^\epsilon}}^{\bf u^\epsilon})\right] \\
& \quad - \mathbb{E}_{s, x}\left[\int_s^{(s+\epsilon) \wedge \tau_{s, x}^{\bf u^\epsilon}} \left( V_s(t, X_{t}^{\bf u^{\epsilon}}) + (\mu - l + r) V_x(t, X_{t}^{\bf u^{\epsilon}}) + \frac{1}{2} \sigma^2 V_{xx}(t, X_{t}^{\bf u^{\epsilon}}) \right) dt \right],
\end{align*}
where the last equality comes from \eqref{eq:Vf}. Noting that $V(s, x) = J(s, x, {\bf \hat u})$, we obtain
\begin{equation}\label{eq:J_diff}
\begin{aligned}
& J(s, x, {\bf \hat u}) - J(s, x, {\bf u^\epsilon})\\
=& \mathbb{E}_{s, x}\left[f((s+\epsilon) \wedge \tau_{s, x}^{\bf u^\epsilon}, X_{(s+\epsilon) \wedge \tau_{s, x}^{\bf u^\epsilon}}^{\bf u^\epsilon}) - g\left( (s+\epsilon) \wedge\tau_{s, x}^{\bf u^\epsilon} , X_{(s+\epsilon) \wedge \tau_{s, x}^{\bf u^\epsilon}}^{\bf u^\epsilon}; s \right) \right]\\
\quad & - \mathbb{E}_{s, x}\left[\int_s^{(s+\epsilon) \wedge \tau_{s, x}^{\bf u^\epsilon}} \left( V_s(t, X_{t}^{\bf u^{\epsilon}}) + (\mu - l + r) V_x(t, X_{t}^{\bf u^{\epsilon}}) + \frac{1}{2} \sigma^2 V_{xx}(t, X_{t}^{\bf u^{\epsilon}}) \right) dt \right] \\
\quad & - \mathbb{E}_{s, x} \left[\int_{s}^{(s+\epsilon) \wedge \tau_{s, x}^{\bf u^\epsilon}} \delta(s,t) dt \right] \left(l - \phi r \right).
\end{aligned}
\end{equation}
For the first term on the right-hand side of \eqref{eq:J_diff},
\begin{equation}\label{eq:f-g}
\begin{aligned}
& \quad \lim_{\epsilon \to 0}\frac{\mathbb{E}_{s, x}\left[f((s+\epsilon) \wedge \tau_{s, x}^{\bf u^\epsilon}, X_{(s+\epsilon) \wedge \tau_{s, x}^{\bf u^\epsilon}}^{\bf u^\epsilon}) - g\left( (s+\epsilon) \wedge\tau_{s, x}^{\bf u^\epsilon} , X_{(s+\epsilon) \wedge \tau_{s, x}^{\bf u^\epsilon}}^{\bf u^\epsilon}; s \right) \right]}{\epsilon}\\
& = \lim_{\epsilon \to 0}\frac{\mathbb{E}_{s, x}\left[f((s+\epsilon) \wedge \tau_{s, x}^{\bf u^\epsilon}, X_{(s+\epsilon) \wedge \tau_{s, x}^{\bf u^\epsilon}}^{\bf u^\epsilon})\right]   - f(s, x)}{\epsilon}\\
&\quad - \lim_{\epsilon \to 0} \frac{ \mathbb{E}_{s, x}\left[g((s+\epsilon) \wedge\tau_{s, x}^{\bf u^\epsilon} , X_{(s+\epsilon) \wedge \tau_{s, x}^{\bf u^\epsilon}}^{\bf u^\epsilon}; s)\right] - g(s, x; s)}{\epsilon}\\
& = f_s(s,x) + (\mu - l + r) f_x(s,x) + \frac{1}{2} \sigma^2 f_{xx}(s,x) \\
& \quad - \left(  g_s(s,x;s) + (\mu - l + r) g_x(s,x;s) + \frac{1}{2} \sigma^2 g_{xx}(s,x; s)\right).
\end{aligned}
\end{equation}
For the second term on the right-hand side of \eqref{eq:J_diff},
\begin{align*}
& \quad \lim_{\epsilon \to 0} -\frac{\mathbb{E}_{s, x}\left[\int_s^{(s+\epsilon) \wedge \tau_{s, x}^{\bf u^\epsilon}} \left( V_s(t, X_{t}^{\bf u^{\epsilon}}) + (\mu - l + r) V_x(t, X_{t}^{\bf u^{\epsilon}}) + \frac{1}{2} \sigma^2 V_{xx}(t, X_{t}^{\bf u^{\epsilon}}) \right) dt\right]}{\epsilon}\\
&= - \left( V_s(s,x) + (\mu - l + r) V_x(s,x) + \frac{1}{2} \sigma^2 V_{xx}(s,x) \right).
\end{align*}
For the third term on the right-hand side of \eqref{eq:J_diff},
\begin{align*}
\lim_{\epsilon \to 0} -\frac{\mathbb{E}_{s, x} \left[\int_{s}^{(s+\epsilon) \wedge \tau_{s, x}^{\bf u^\epsilon}} \delta(s,t) dt \right] \left(l - \phi r \right) }{\epsilon}
= -\left( l - \phi r \right).
\end{align*}
Therefore, dividing $\epsilon$ and letting $\epsilon\to 0$ on both sides of \eqref{eq:J_diff} leads to
\begin{align*}
& \quad \liminf_{\epsilon \downarrow 0} \frac{ J(s, x, {\bf \hat u}) - J(s, x, {\bf u^\epsilon})}{\epsilon} \\
& = f_s(s,x) + (\mu - l + r) f_x(s,x) + \frac{1}{2} \sigma^2 f_{xx}(s,x) \\
& \quad - \left(  g_s(s,x;s) + (\mu - l + r) g_x(s,x;s) + \frac{1}{2} \sigma^2 g_{xx}(s,x; s)\right) \\
& \quad -\left( V_s(s,x) + (\mu - l + r) V_x(s,x) + \frac{1}{2} \sigma^2 V_{xx}(s,x) \right) - \left( l - \phi r \right) \\
& \geq - \sup_{l \in [0, \bar l], r \in [0, \bar r]} \Big\{ V_s(s,x) + (\mu - l + r) V_x(s,x) + \frac{1}{2} \sigma^2 V_{xx}(s,x) +  \left( l - \phi r \right) \\
& \quad - \left( f_s(s,x) + (\mu - l + r) f_x(s,x) + \frac{1}{2} \sigma^2 f_{xx}(s,x) \right ) \\
& \quad + g_s(s,x;s) + (\mu - l + r) g_x(s,x;s) + \frac{1}{2} \sigma^2 g_{xx}(s,x; s) \Big\}\\
& = 0,
\end{align*}
where the last equality is due to \eqref{extended_HJB:1}.
\qed

\pfof{Proposition \ref{prop:simplified_HJB}}
Note that $f(s, x) = g(s, x; s)$ as given by \eqref{extended_HJB:3}, we have
$$
f_x(s,x) = \lim_{\epsilon\to0} \frac{f(s, x+\epsilon) - f(s, x)}{\epsilon} = \lim_{\epsilon\to0} \frac{g(s, x+\epsilon; s) - g(s, x; s)}{\epsilon} = g_x(s,x;s),
$$
and similarly, $f_{xx}(s,x) = g_{xx}(s,x;s)$. Thus
\begin{align*}
& f_s(s,x) + (\mu - l + r) f_x(s,x) + \frac{1}{2} \sigma^2 f_{xx}(s,x) \\
& \quad - \left(g_s(s,x;s) + (\mu - l + r) g_x(s,x;s) + \frac{1}{2} \sigma^2 g_{xx}(s,x; s)\right) = f_s(s,x) - g_s(s,x;s).
\end{align*}

Note the fact that $\hat l_{\tau_{s, x}^{\bf \hat u}} = \hat r_{\tau_{s, x}^{\bf \hat u}} = 0$ given initial pair $(s, x)$. Recall the probabilistic interpretation \eqref{eq:g}-\eqref{eq:f}, by using Leibniz formula, we derive that
$$
g_s(s,x;k) = \mathbb{E}_{s, x} \left[\frac{\partial \int_s^{\tau_{s, x}^{\bf \hat u}} \delta(k, t) \left(\hat l_t - \phi \hat r_t \right) dt}{\partial s}\right]= - \delta(k, s) \left(\hat l_s - \phi \hat r_s \right),
$$
and
\begin{align*}
f_s(s,x) &= \mathbb{E}_{s, x} \left[\frac{\partial \int_s^{\tau_{s, x}^{\bf \hat u}} \delta(s, t) \left(\hat l_t - \phi \hat r_t \right) dt}{\partial s}\right]\\
&= - \delta(s, s) \left(\hat l_s - \phi \hat r_s \right) + \mathbb{E}_{s, x} \left[\int_s^{\tau_{s, x}^{\bf \hat u}} \delta_s(s, t) \left(\hat l_t - \phi \hat r_t \right) dt\right]\\
&= - \left(\hat l_s - \phi \hat r_s \right) + \mathbb{E}_{s, x} \left[\int_s^{\tau_{s, x}^{\bf \hat u}} \delta_s(s, t) \left(\hat l_t - \phi \hat r_t \right) dt\right],
\end{align*}
where $\delta_s(s, t)$ is the partial derivative with respect to variable $s$. Therefore,
\begin{align*}
f_s(s,x) - g_s(s,x;s) &= - \left(\hat l_s - \phi \hat r_s \right) + \mathbb{E}_{s, x} \left[\int_s^{\tau_{s, x}^{\bf \hat u}} \delta_s(s, t) \left(\hat l_t - \phi \hat r_t \right) dt\right] + \delta(s, s) \left(\hat l_s - \phi \hat r_s \right)\\
&= \mathbb{E}_{s, x} \left[\int_s^{\tau_{s, x}^{\bf \hat u}} \delta_s(s, t) \left(\hat l_t - \phi \hat r_t \right) dt\right].
\end{align*}
It is obvious that the above term has nothing to do with strategy $u = (l, r)$, thus we take it out of the supremum in \eqref{extended_HJB:1}, which leads to the conclusion.
\qed

\pfof{Proposition \ref{prop:Vi}}
Note that given the equilibrium strategy \eqref{eq:equilibrium_feedback_strategy}, the ODEs \eqref{eq:Vi_ODE} is equivalent to
\begin{equation}\label{eq:Vi_HJB}
\frac{1}{2} \sigma^2 v^{\prime\prime}_i(x) + (\mu - \hat l(x) + \hat r(x)) v^\prime_i(x) - \rho_i v_i(x) + (\hat l(x) - \phi \hat r(x)) = 0.
\end{equation}
Then we show that $v_i\in {\cal C}^2$ solves \eqref{eq:Vi_HJB} has the probabilistic interpretation \eqref{eq:Vi}. Applying It\^{o}'s Lemma on $e^{-\rho_i t} v_i(X_{t}^{\bf \hat u})$, for any $T > 0$,
\begin{align*}
v_i(x) &= e^{-\rho_i (T \wedge \tau_{0, x}^{\bf \hat u})} v_i(X_{T \wedge \tau_{0, x}^{\bf \hat u})}^{\bf \hat u})\\
& \quad - \int_0^{T \wedge \tau_{0, x}^{\bf \hat u}} e^{-\rho_i t} \left(-\rho_i v_i(X_t^{\bf \hat u}) + \left(\mu - \hat l(X_t^{\bf \hat u}) + \hat r(X_t^{\bf \hat u})\right)v^\prime_i(X_t^{\bf \hat u}) + \frac{1}{2} \sigma^2 v^{\prime\prime}_i(X_t^{\bf \hat u}) \right) dt\\
&\quad - \int_0^{T \wedge \tau_{0, x}^{\bf \hat u}} \sigma e^{-\rho_i t} v_i^\prime(x) dW_t\\
&= e^{-\rho_i (T \wedge \tau_{0, x}^{\bf \hat u})} v_i(X_{T \wedge \tau_{0, x}^{\bf \hat u})}^{\bf \hat u}) + \int_0^{T \wedge \tau_{0, x}^{\bf \hat u}} e^{-\rho_i t} \left(\hat l(X_t^{\bf \hat u}) - \phi \hat r(X_t^{\bf \hat u})\right) dt \\
& \quad - \int_0^{T \wedge \tau_{0, x}^{\bf \hat u}} \sigma e^{-\rho_i t} v_i^\prime(x) dW_t,
\end{align*}
where the last equality is due to \eqref{eq:Vi_HJB}. Then taking expectations on both sides,
\begin{align*}
v_i(x) &= \expect_{0,x}\left[e^{-\rho_i (T \wedge \tau_{0, x}^{\bf \hat u})} v_i(X_{T \wedge \tau_{0, x}^{\bf \hat u})}^{\bf \hat u})\right] + \expect_{0,x}\left[\int_0^{T \wedge \tau_{0, x}^{\bf \hat u}} e^{-\rho_i t} \left(\hat l(X_t^{\bf \hat u}) - \phi \hat r(X_t^{\bf \hat u})\right) dt \right] \\
& \quad - \expect_{0,x}\left[\int_0^{T \wedge \tau_{0, x}^{\bf \hat u}} \sigma e^{-\rho_i t} v_i^\prime(x) dW_t \right].
\end{align*}
Letting $T \to \infty$ and similar to the arguments in the proof of Theorem \ref{thm:verification_equilibrium}, we have
$$
v_i(x) = \expect_{0,x}\left[\int_0^{\tau_{0, x}^{\bf \hat u}} e^{-\rho_i t} \left(\hat l(X_t^{\bf \hat u}) - \phi \hat r(X_t^{\bf \hat u})\right) dt \right].
$$
Thus, $V(x) = \omega v_1(x) + (1-\omega) v_2(x)$ is the equilibrium value function.
\qed

\begin{proposition}\label{prop:parameter_sign}
Let $A_{i1}$, $B_{i3}$ be given by \eqref{eq:ABs} with $C_{ij}$, $D_{ij}$, $E_i$ defined in \eqref{eq:CDEs}, then $A_{i1} > 0$ and $B_{i3} < 0$, $\forall i = 1, 2$.
\end{proposition}

\pfof{Proposition \ref{prop:parameter_sign}}
First, to show the sign of $A_{i1}$, notice that $C_{i1}  < 0$. Recall the definition of $\theta_{ij}$ given in \eqref{eq:theta_i}, we derive that
$$
\theta_{i3} - \theta_{i1} = \frac{\bar r + \sqrt{\mu^2 + 2 \rho_i \sigma^2} - \sqrt{(\mu + \bar r)^2 + 2 \rho_i \sigma^2}}{\sigma^2} > 0,
$$
where we use the fact that
$$
\left(\bar r + \sqrt{\mu^2 + 2 \rho_i \sigma^2}\right)^2 - \left[(\mu + \bar r)^2 + 2 \rho_i \sigma^2\right] = 2 \bar r \left(\sqrt{\mu^2 + 2 \rho_i \sigma^2} - \mu\right) > 0,
$$
and we can show $\theta_{i2} < \theta_{i4} < \theta_{i5}$ using the same argument. Recall that
\begin{align*}
E_i  &= (\theta_{i3} - \theta_{i1}) (\theta_{i4} - \theta_{i5}) e^{(\theta_{i1}+\theta_{i3}) x_1 + \theta_{i4} x_2}  + (\theta_{i1} - \theta_{i4}) (\theta_{i3} - \theta_{i5}) e^{(\theta_{i1} + \theta_{i4}) x_1 + \theta_{i3} x_2} \\
&\quad + (\theta_{i2} - \theta_{i3}) (\theta_{i4} - \theta_{i5}) e^{(\theta_{i2}+\theta_{i3} ) x_1 + \theta_{i4} x_2} +(\theta_{i4} - \theta_{i2}) (\theta_{i3} - \theta_{i5}) e^{(\theta_{i2}+\theta_{i4}) x_1 + \theta_{i3} x_2}.
\end{align*}
It is obvious that the last two terms of $E_i$ are positive, and the first two terms of $E_i$ can be rearranged as follows:
\begin{align*}
& (\theta_{i3} - \theta_{i1}) (\theta_{i4} - \theta_{i5}) e^{(\theta_{i1}+\theta_{i3}) x_1 + \theta_{i4} x_2}  + (\theta_{i1} - \theta_{i4}) (\theta_{i3} - \theta_{i5}) e^{(\theta_{i1} + \theta_{i4}) x_1 + \theta_{i3} x_2} \\
= & \theta_{i4} (\theta_{i3} - \theta_{i1}) e^{(\theta_{i1}+\theta_{i3}) x_1 + \theta_{i4} x_2} - \theta_{i5}  (\theta_{i3} - \theta_{i1}) e^{(\theta_{i1}+\theta_{i3}) x_1 + \theta_{i4} x_2} \\
\quad& + \theta_{i1} (\theta_{i3} - \theta_{i5}) e^{(\theta_{i1} + \theta_{i4}) x_1 + \theta_{i3} x_2} - \theta_{i4} (\theta_{i3} - \theta_{i5}) e^{(\theta_{i1} + \theta_{i4}) x_1 + \theta_{i3} x_2}\\
= & \theta_{i4} (\theta_{i3} - \theta_{i1}) e^{(\theta_{i1}+\theta_{i3}) x_1 + \theta_{i4} x_2} - \theta_{i4} (\theta_{i3} - \theta_{i5}) e^{(\theta_{i1} + \theta_{i4}) x_1 + \theta_{i3} x_2} \\
\quad & - \theta_{i5}  (\theta_{i3} - \theta_{i1}) e^{(\theta_{i1}+\theta_{i3}) x_1 + \theta_{i4} x_2} + \theta_{i1} (\theta_{i3} - \theta_{i5}) e^{(\theta_{i1} + \theta_{i4}) x_1 + \theta_{i3} x_2}.
\end{align*}
The last two terms above are positive by definition, and it remains to check the first two terms. Note that $0 < x_1 \le x_2$, thus $(\theta_{i1}+\theta_{i3}) x_1 + \theta_{i4} x_2 \le (\theta_{i1} + \theta_{i4}) x_1 + \theta_{i3} x_2$ and $e^{(\theta_{i1}+\theta_{i3}) x_1 + \theta_{i4} x_2} \le e^{(\theta_{i1} + \theta_{i4}) x_1 + \theta_{i3} x_2}$. Then we obtain that
\begin{align*}
& \theta_{i4} (\theta_{i3} - \theta_{i1}) e^{(\theta_{i1}+\theta_{i3}) x_1 + \theta_{i4} x_2} - \theta_{i4} (\theta_{i3} - \theta_{i5}) e^{(\theta_{i1} + \theta_{i4}) x_1 + \theta_{i3} x_2}\\
\ge & \theta_{i4} (\theta_{i3} - \theta_{i1}) e^{(\theta_{i1} + \theta_{i4}) x_1 + \theta_{i3} x_2} - \theta_{i4} (\theta_{i3} - \theta_{i5}) e^{(\theta_{i1} + \theta_{i4}) x_1 + \theta_{i3} x_2}\\
= & \theta_{i4} (\theta_{i5} - \theta_{i1}) e^{(\theta_{i1} + \theta_{i4}) x_1 + \theta_{i3} x_2} >  0.
\end{align*}
Then it shows that $E_i > 0$.
Note that
\begin{align*}
D_{i1} &= \theta_{i3} \theta_{i4} \left(e^{\theta_{i4} x_1 + \theta_{i3} x_2} - e^{\theta_{i3} x_1 + \theta_{i4} x_2}\right) + \theta_{i3} \theta_{i5} e^{\theta_{i3} x_1 + \theta_{i4} x_2} - \theta_{i4} \theta_{i5} e^{\theta_{i4} x_1 + \theta_{i3} x_2}\\
& \quad + (\theta_{i3} - \theta_{i2}) (\theta_{i4} - \theta_{i5}) e^{(\theta_{i2} + \theta_{i3}) x_1+ \theta_{i4} x_2} + (\theta_{i2} - \theta_{i4}) (\theta_{i3} - \theta_{i5}) e^{(\theta_{i2} + \theta_{i4}) x_1 + \theta_{i3} x_2}\\
& < \theta_{i3} \theta_{i4} \left(e^{\theta_{i4} x_1 + \theta_{i3} x_2} - e^{\theta_{i3} x_1 + \theta_{i4} x_2}\right).
\end{align*}
Since $0 < x_1 \le x_2$, $\theta_{i4} x_1 + \theta_{i3} x_2 \ge \theta_{i3} x_1 + \theta_{i4} x_2$, thus $e^{\theta_{i4} x_1 + \theta_{i3} x_2} - e^{\theta_{i3} x_1 + \theta_{i4} x_2} \ge 0$, from which we obtain that $D_{i1} \le 0$. Therefore, $A_{i1}$ given in \eqref{eq:ABs} is positive.

Next, to show $B_{i3} < 0$, knowing that, for $x \in [x_2, \infty)$,
$$
V_i(x) = \frac{\bar l}{\rho_i} + B_{i3} e^{\theta_{i5} x} \leq \frac{\bar l}{\rho_i},
$$
where the inequality comes directly from Proposition \ref{corollary:Vi_bounded}, it can be inferred that $B_{i3} \le 0$.

Suppose $B_{i3} = 0$. If $0 < x_1 < x_2$, according to smooth fit principle, we have $V_i(x_2-) = V_i(x_2) = \frac{\bar l}{\rho_i} > 0$ and $V^\prime_i(x_2-) = V^\prime_i(x_2) = B_{i3} \theta_{i5} e^{\theta_{i5} x_2} = 0$. Using expression of $V_i$ given in \eqref{eq:Vi_solution},
\begin{equation*}
\left\{\begin{aligned}
& V_i(x_2-) = A_{i2} e^{\theta_{i3} x_2-} + B_{i2} e^{\theta_{i4} x_2-} > 0,\\
& V^\prime_i(x_2-) = A_{i2} \theta_{i3} e^{\theta_{i3} x_2-} + B_{i2} \theta_{i4} e^{\theta_{i4} x_2-} = 0.
\end{aligned}\right.
\end{equation*}
Multiply the first inequality in above by $\theta_{i3}$ and $\theta_{i4}$ respectively, and then subtract the second equality:
\begin{equation*}
\left\{\begin{aligned}
& A_{i2} (\theta_{i4} - \theta_{i3}) e^{\theta_{i3} x_2-} < 0,\\
& B_{i2} (\theta_{i3} - \theta_{i4}) e^{\theta_{i4} x_2-} > 0,
\end{aligned}\right.
\end{equation*}
from which one obtains that $A_{i2} > 0$ and $B_{i2} > 0$. Therefore, for $x \in [x_1, x_2)$, $V^{\prime\prime}_i(x) =  A_{i2} \theta_{i3}^2 e^{\theta_{i3} x} + B_{i2} \theta_{i4}^2 e^{\theta_{i4} x} > 0$ and thus $V^{\prime}_i(x_1) < V^{\prime}_i(x_2-) = 0$.

Recall that $A_{i1} > 0$ as shown before and $V^{\prime}_i(x_1-) = A_{i1} \theta_{i1} e^{\theta_{i1} x_1-} + B_{i1} \theta_{i2} e^{\theta_{i2} x_1-} = V^{\prime}_i(x_1) < 0$, $B_{i1}$ can only be strictly positive. Therefore, using the same argument we obtain that $V^{\prime\prime}_i(x) > 0$ and $V^{\prime}_i(x) < V^{\prime}_i(x_1) < 0$ for $x \in [0, x_1)$, it then follows that $V_i(0) = 0 > V_i(x_1) = A_{i2} e^{\theta_{i3} x_1} + B_{i2} e^{\theta_{i4} x_1} > 0$, which leads to the contradiction. Then $B_{i3} < 0$.
The proof of the case of $0 < x_1 = x_2$ is similar.
\qed

Because $B_{i3} < 0$, $V_i$ is concave on interval $[x_2, \infty)$. The concavity of $V_i$ on the whole domain is proved with the help of Lemma \ref{lemma:concave} inspired by \citet[Lemma 4.1]{shreve1984optimal} and \citet[Proposition 3.2]{bai2023reinforcement}.

\begin{lemma}\label{lemma:concave}
Suppose that $v_i(x)$, $i = 1, 2$, are solutions to the linear ordinary differential equations:
\begin{equation}\label{eq:ode_concave}
\frac{1}{2} \sigma^2 v_i^{\prime\prime}(x) + (\mu + \alpha) v_i^{\prime}(x) - \rho_i v_i(x) - \phi \alpha = 0, \quad x \in [a,b),
\end{equation}
where $\mu, \alpha > 0$, $0 < \rho_1 \le \rho_2$, and $\phi \geq 1$ are constants. Then $v(x) = \omega v_1(x) + (1-\omega)v_2(x)$ is a concave function if $v^{\prime\prime}(b-) < 0$ and $v_i(x)$, $i=1, 2$, satisfy either of the following conditions:
\begin{enumerate}
\item[(i)] $v_i^{\prime}(x) > 0$, $\forall x \in [a,b)$;
\item[(ii)] $v^{\prime}(b-) = \phi$, $v_i^{\prime}(b-) > 0$, and $v_i(a) = 0$.
\end{enumerate}
\end{lemma}

\proof
First, differentiate \eqref{eq:ode_concave} and let $m(x) := v^\prime(x)$ and $m_i(x) := v_i^\prime(x)$, then it can be obtained that
$$
m_i^{\prime\prime}(x) = \frac{2 \rho_i}{\sigma^2} m_i(x) - \frac{2 (\mu + \alpha)}{\sigma^2} m_i^{\prime}(x), \quad x \in [a,b).
$$
Then we introduce a change of variable $\varphi(x) := \int_a^x \exp\left(-2 (\mu + \alpha) u/\sigma^2\right) du$, $x \in [a, b)$,
and define $l(y) := \omega l_1(y) + (1-\omega)l_2(y)$ with $l_i(y) := m_i(\varphi^{-1}(y))$, $y\in[0,\bar y)$, where $\bar y = \int_a^b \exp\left(-2 (\mu + \alpha) u/\sigma^2\right) du$. Note that
$$
l_i^{\prime}(y) = \frac{1}{\varphi^\prime(\varphi^{-1}(y))} m^{\prime}_i(\varphi^{-1}(y)) = \frac{m^{\prime}_i(\varphi^{-1}(y))}{\exp\left(-2 (\mu + \alpha) \varphi^{-1}(y)/\sigma^2\right)},
$$
then
$$
l^{\prime}(y) = \frac{\omega m^{\prime}_1(\varphi^{-1}(y)) + (1-\omega) m^{\prime}_2(\varphi^{-1}(y))}{\exp\left(-2 (\mu + \alpha) \varphi^{-1}(y)/\sigma^2\right)} = \frac{m^{\prime}(\varphi^{-1}(y))}{\exp\left(-2 (\mu + \alpha) \varphi^{-1}(y)/\sigma^2\right)},
$$
which implies that $v^{\prime\prime}(\varphi^{-1}(y)) = m^{\prime}(\varphi^{-1}(y))$ has same sign as $l^{\prime}(y)$, specifically, $l^{\prime}(\bar y-) < 0$ because $v^{\prime\prime}(b-) < 0$. Differentiating $l_i^{\prime}$, we obtain the following second-order linear ODEs for $l_i$:
\begin{align*}
l_i^{\prime\prime}(y) &= \left[\varphi^\prime(\varphi^{-1}(y))\right]^{-2} \left(m^{\prime\prime}_i(\varphi^{-1}(y)) - \frac{\varphi^{\prime\prime}(\varphi^{-1}(y))}{\varphi^\prime(\varphi^{-1}(y))} m^\prime_i(\varphi^{-1}(y))\right) \\
&= \left[\varphi^\prime(\varphi^{-1}(y))\right]^{-2} \left(\frac{2 \rho_i}{\sigma^2} m_i(\varphi^{-1}(y)) - \frac{2 (\mu + \alpha)}{\sigma^2} m_i^{\prime}(\varphi^{-1}(y)) + \frac{2(\mu + \alpha)}{\sigma^2} m^\prime_i(\varphi^{-1}(y))\right) \\
& = \left[\varphi^\prime(\varphi^{-1}(y))\right]^{-2} \frac{2 \rho_i}{\sigma^2} l_i(y).
\end{align*}
Hence,
\begin{align}\label{eq:l_ode}
l^{\prime\prime}(y) &= \omega l_1^{\prime\prime}(y) + (1-\omega) l_2^{\prime\prime}(y)\notag\\
&= \left[\varphi^\prime(\varphi^{-1}(y))\right]^{-2} \frac{2}{\sigma^2} \left(\omega \rho_1 l_1(y) + (1-\omega) \rho_2 l_2(y) \right).
\end{align}

On one hand, if $m_i(x) = v_i^{\prime}(x) > 0$, $\forall x \in [a,b)$, then for any $y \in [0,\bar y)$, $l_i(y) = m_i(\varphi^{-1}(y)) > 0$ and $l^{\prime\prime}(y) > 0$ by \eqref{eq:l_ode}, which implies that $l^{\prime}(y)$ is increasing on interval $[0,\bar y)$, and thus $l^{\prime}(y) < l^{\prime}(\bar y-) < 0$ for any $y \in [0,\bar y)$. As $v^{\prime\prime}(x)$ has same sign as $l^{\prime}(\varphi(x))$, we prove that $v^{\prime\prime}(x)< 0$, $\forall x \in [a, b)$.

On the other hand, assume that condition $(ii)$ holds. Suppose there exists $y^* \in [0, \bar y)$ such that $l^\prime(y^*) = 0$ and $l^\prime(y) < 0$ for $y \in (y^*, \bar y)$. Then $l^{\prime\prime}(y^*) < 0$ and by \eqref{eq:l_ode},
\begin{equation}\label{ineq1:l}
\omega \rho_1 l_1(y^*) + (1-\omega) \rho_2 l_2(y^*) < 0.
\end{equation}
Since $l^\prime(y) < 0$ for $y \in (y^*, \bar y)$, we have $l(y^*) > l(\bar y-) = m(b-) = v^\prime(b-) > 0$, that is,
\begin{equation}\label{ineq2:l}
\omega l_1(y^*) + (1-\omega) l_2(y^*) > 0.
\end{equation}
Let inequality \eqref{ineq2:l} multiply $-\rho_1$ on both sides and then add up with \eqref{ineq1:l}, we obtain $(1-\omega) (\rho_2 - \rho_1) l_2(y^*) < 0$, which implies that $l_1(y^*) > 0$ and $l_2(y^*) < 0$. Combining with the fact that $l_2(\bar y-) = v_2^\prime(b-) > 0$ and $l^\prime_2$ has no zero in $[0, \bar y)$ (thanks to \citet[Lemma 4.1]{shreve1984optimal}), it indicates that $l^\prime_2(y) > 0$, $\forall y \in [0, \bar y)$. Therefore,
\begin{equation}\label{ineq:lprime}
\omega \rho_1 l_1^{\prime}(y) + (1-\omega) \rho_2 l_2^{\prime}(y) > \omega \rho_1 l_1^{\prime}(y) + (1-\omega) \rho_1 l_2^{\prime}(y) = \rho_1 l^{\prime}(y).
\end{equation}
We then claim that $\omega \rho_1 l_1(y) + (1-\omega) \rho_2 l_2(y) < 0$ (i.e., $l^{\prime\prime}(y) < 0$) for all $y\in[0, y^*]$. Otherwise, if there exists $y_1$ such that $y_1 := \sup\{y\in[0, y^*]: l^{\prime\prime}(y_1) = 0$\}, since $l^{\prime\prime}(y^*) < 0$, we have $l^{\prime\prime}(y) < 0$ for any $y\in(y_1, y^*]$. Consequently, $l^{\prime}(y_1) > l^{\prime}(y^*) = 0$, and inequality \eqref{ineq:lprime} tells that $\omega \rho_1 l_1^{\prime}(y_1) + (1-\omega) \rho_2 l_2^{\prime}(y_1) > \rho_1 l^\prime(y_1) > 0$. Then according to \eqref{eq:l_ode}, $l^{\prime\prime}(y_1) = 0$ and $l^{\prime\prime}(y_1+) < 0$ lead to $\omega \rho_1 l_1(y_1) + (1-\omega) \rho_2 l_2(y_1) = 0$ and $\omega \rho_1 l_1(y_1+) + (1-\omega) \rho_2 l_2(y_1+) < 0$, that is, $\omega \rho_1 l_1^{\prime}(y_1) + (1-\omega) \rho_2 l_2^{\prime}(y_1) < 0$, which is a contradiction.

Since $v^{\prime\prime}(x^*)$ has same sign as $l^{\prime}(y^*) = 0$ and $v^{\prime}(x^*) = l(y^*) > l(\bar y -) = v^{\prime}(b-) = \phi$, where $x^* = \varphi^{-1}(y^*)$, combining with \eqref{eq:ode_concave} leads to
\begin{equation}\label{ineq:v_x*}
\omega \rho_1 v_1(x^*) + (1-\omega)\rho_2 v_2(x^*) = \frac{1}{2} \sigma^2 v^{\prime\prime}(x^*) + (\mu + \alpha) v^{\prime}(x^*) - \phi \alpha > 0.
\end{equation}
However, $\omega \rho_1 v^{\prime}_1(x) + (1-\omega) \rho_2 v^{\prime}_2(x) = \omega \rho_1 l_1(\phi(x)) + (1-\omega) \rho_2 l_2(\phi(x)) < 0$ for all $x \in [a, x^*]$, resulting in $\omega \rho_1 v_1(x^*) + (1-\omega) \rho_2 v_2(x^*) < \omega \rho_1 v_1(a) + (1-\omega) \rho_2 v_2(a) = 0$, which contradicts with \eqref{ineq:v_x*}, therefore, there cannot exist $y^* \in [0, \bar y)$ such that $l^\prime(y^*) = 0$, that is, $v^{\prime\prime}(x) < 0$ for all $x \in [a, b)$.
\qed

\pfof{Theorem \ref{thm:V_concave}}
First, consider $V_i(x)$ for $x \in [x_2, \infty)$. It is straightforward that $V^\prime_i(x) = B_{i3} \theta_{i5} e^{\theta_{i5} x} > 0$ and $V^{\prime\prime}_i(x) = B_{i3} \theta_{i5}^2 e^{\theta_{i5} x} < 0$ from Proposition \ref{prop:parameter_sign}.

Second, for $x \in [x_1, x_2)$, since $V^\prime_i(x_2-) = V^\prime_i(x_2) > 0$, using expression \eqref{eq:Vi_solution}, it is equivalent to $A_{i2} \theta_{i3} e^{\theta_{i3} x_2-} + B_{i2} \theta_{i4} e^{\theta_{i4} x_2-} > 0$. We claim that $V^\prime_i(x) > 0$ for all $x \in [x_1, x_2)$ by considering different signs of $A_{i2}$ and $B_{i2}$. If $A_{i2} < 0$, then in order to make $V^\prime_i(x_2-) > 0$, we can only have $B_{i2} < 0$, hence $V^{\prime\prime}_i(x) = A_{i2} \theta_{i3}^2 e^{\theta_{i3} x} + B_{i2} \theta_{i4}^2 e^{\theta_{i4} x} < 0$, indicating $V^\prime_i(x) > V^\prime_i(x_2-) > 0$, $\forall x \in [x_1, x_2)$. If $A_{i2} > 0$ and $B_{i2} < 0$, then $V^\prime_i(x) = A_{i2} \theta_{i3} e^{\theta_{i3} x} + B_{i2} \theta_{i4} e^{\theta_{i4} x} > 0$. If $A_{i2} > 0$ and $B_{i2} > 0$, suppose $V_i^\prime(x_1) < 0$. According to the smooth principle, $V_i^\prime(x_1-) = A_{i1} \theta_{i1} e^{\theta_{i1} x_1-} + B_{i1} \theta_{i2} e^{\theta_{i2} x_1-} = V_i^\prime(x_1) < 0$. As $A_{i1} > 0$ by Proposition \ref{prop:parameter_sign}, $B_{i1}$ can only be positive, resulting in $V_i^{\prime\prime}(x) = A_{i1} \theta_{i1}^2 e^{\theta_{i1} x} + B_{i1} \theta_{i2}^2 e^{\theta_{i2} x} > 0$. Thus, $V_i^\prime(x) < V_i^\prime(x_1-) < 0$ for all $x \in [0, x_1)$, leading to $V_i(0) > V_i(x_1-) = V_i(x_1) = A_{i2} e^{\theta_{i3} x_1} + B_{i2} e^{\theta_{i4} x_1} > 0$, which contradicts with the boundary condition $V_i(0) = 0$. Therefore, $V_i^\prime(x_1)$ is positive and $V^{\prime\prime}_i(x) = A_{i2} \theta_{i3}^2 e^{\theta_{i3} x} + B_{i2} \theta_{i4}^2 e^{\theta_{i4} x} > 0$, which confirms that $V_i^\prime(x) > V_i^\prime(x_1) > 0$, $\forall x \in [x_1, x_2)$.

Next, we show $V^{\prime\prime}(x) = \omega V^{\prime\prime}_1(x) + (1-\omega) V_2^{\prime\prime}(x) < 0$, $\forall x \in [x_1, x_2)$. Since $V^\prime(x_2-) = V^\prime(x_2) = 1$ and $V^{\prime\prime}(x_2) < 0$, according to ODEs \eqref{eq:Vi_ODE},
\begin{align*}
V^{\prime\prime}(x_2-) &= \omega V^{\prime\prime}_1(x_2-) + (1-\omega) V_2^{\prime\prime}(x_2-)\\
&= \frac{2}{\sigma^2} \left[-\mu V^\prime(x_2-) + \left(\omega \rho_1 V_1(x_2-) + (1-\omega) \rho_2 V_2(x_2-)\right)\right]\\
&= \frac{2}{\sigma^2} \left[-(\mu-\bar l) V^\prime(x_2) - \bar l + \left(\omega \rho_1 V_1(x_2) + (1-\omega) \rho_2 V_2(x_2)\right)\right]\\
&= \omega V^{\prime\prime}_1(x_2) + (1-\omega) V_2^{\prime\prime}(x_2) = V^{\prime\prime}(x_2) < 0.
\end{align*}
Then by Lemma \ref{lemma:concave} $(i)$ with $\alpha$ being $0$, $V^{\prime\prime}(x)$ should always be negative on interval $[x_1, x_2)$.

Finally, for $x \in [0, x_1)$, it is already known that $V_i(0) = 0$, $V^\prime_i(x_1-) = V^\prime_i(x_1) > 0$, and $V^\prime(x_1-) = V^\prime(x_1) = \phi$. Using the same argument as above,
\begin{align*}
V^{\prime\prime}(x_1-) &= \frac{2}{\sigma^2} \left[-(\mu+\bar r) V^\prime(x_1-) + \phi \bar r + \left(\omega \rho_1 V_1(x_1-) + (1-\omega) \rho_2 V_2(x_1-)\right)\right]\\
&= \frac{2}{\sigma^2} \left[-\mu V^\prime(x_1) + \left(\omega \rho_1 V_1(x_1) + (1-\omega) \rho_2 V_2(x_1)\right)\right]\\
&= \omega V^{\prime\prime}_1(x_1) + (1-\omega) V_2^{\prime\prime}(x_1) = V^{\prime\prime}(x_1) < 0.
\end{align*}
By Lemma \ref{lemma:concave} $(ii)$ with $\alpha = \bar r$, $V^{\prime\prime}(x) < 0$, which shows $V^{\prime}(x) > V^{\prime}(x_1-) > 0$ for all $x \in [0, x_1)$.
\qed

\pfof{Theorem \ref{thm:V_classification}}
It remains to prove {\it (ii)} and {\it (iii)}. For Case {\it (ii)}, since there doesn't exist such $0 < x_1 \le x_2 < \infty$ that satisfy $V^\prime(x_1) = \phi \ge 1$ and $V^\prime(x_2) = 1$, we assume that there exists $x_2 > 0$ such that $V^\prime(x_2) = 1$. Then the equilibrium strategy becomes
\begin{equation*}
\hat u_s = (\hat l(x), \hat r(x)) = \left\{\begin{aligned}
& (0, 0), \qquad & 0 \leq x < x_2,\\
& (\bar l, 0), & x \geq x_2.
\end{aligned}\right.
\end{equation*}
Following a similar argument as in Case {\it (i)} with boundary condition $V_i(0) = 0$, we obtain that
$$
V_i(x) = \begin{cases}
A_{i2} \left(e^{\theta_{i3} x} - e^{\theta_{i4} x}\right), & x \in [0, x_2),\\
\frac{\bar l}{\rho_i} + B_{i3} e^{\theta_{i5} x}, & x \in [x_2, \infty),
\end{cases}
$$
where, by letting $x_1$ = 0 in \eqref{eq:CDEs},
$$
A_{i2} = -\frac{\theta_{i5}}{(\theta_{i5} - \theta_{i4}) e^{\theta_{i4} x_2}  + (\theta_{i3} - \theta_{i5}) e^{\theta_{i3} x_2}} \cdot \frac{\bar l}{\rho_i} > 0,
$$
and
$$
B_{i3} = - \frac{\theta_{i3} e^{\theta_{i3} x_2} - \theta_{i4} e^{\theta_{i4} x_2}}{(\theta_{i5} - \theta_{i4}) e^{(\theta_{i4} + \theta_{i5}) x_2}  + (\theta_{i3} - \theta_{i5}) e^{(\theta_{i3} + \theta_{i5}) x_2}} \cdot \frac{\bar l}{\rho_i} < 0.
$$
The concavity of $V = \omega V_1 + (1-\omega) V_2$ can be obtained using a similar argument as in the proof of Theorem \ref{thm:V_concave} with $A_{i2} > 0$ and $B_{i3} < 0$. Moreover, the existence of $x_2 > 0$ is guaranteed by the condition that there exists $x_2 > 0$ that solves \eqref{eq:x2} given $x_1 = 0$.

For Case {\it (iii)}, since we cannot find a positive $x_1$ such that $V^\prime(x_1) = \phi$ nor a positive $x_2$ such that $V^\prime(x_2) = 1$, which means $V^\prime(x) \le 1$ for $x \in [0, \infty)$. Thus the equilibrium strategy becomes
$$
\hat u_s = (\hat l(x), \hat r(x)) = (\bar l, 0),  \qquad x \geq 0,
$$
and the boundary condition leads to
$$
V_i(x) =  \frac{\bar l}{\rho_i} \left(1 - e^{\theta_{i5} x}\right), \qquad x \in [0, \infty),
$$
with $\theta_{i5} < 0$, which also shows the concavity of $V$.
\qed

\end{appendices}

\end{document}